\documentclass[11pt]{article}
\pdfoutput=1 %

\usepackage[T1]{fontenc}
\usepackage{lmodern}
\rmfamily
\DeclareFontShape{T1}{lmr}{b}{sc}{<->ssub*cmr/bx/sc}{}
\DeclareFontShape{T1}{lmr}{bx}{sc}{<->ssub*cmr/bx/sc}{}

\usepackage[margin=1in]{geometry}
\usepackage{amsfonts,amssymb,amsmath,amsthm,mathtools,thmtools, mathrsfs}
\usepackage{thm-restate}

\usepackage{microtype,xspace,bm,floatrow}
\usepackage{comment}
\usepackage{mleftright}

\usepackage{doi}
\usepackage{hyperref}
\usepackage{makecell}
\usepackage[dvipsnames]{xcolor}
\hypersetup{
	colorlinks=true,
	urlcolor=MidnightBlue,
	linkcolor=MidnightBlue,
	citecolor=MidnightBlue,
	unicode
}

\usepackage{tocbibind} 

\usepackage{bold-extra}

\usepackage{mdframed}

\usepackage{enumitem}
\setlist{  
  listparindent=\parindent,
  parsep=0pt,
}

\usepackage[capitalise,nameinlink,noabbrev]{cleveref}
\crefname{equation}{}{}

\usepackage{tikz}
\usepackage{tkz-graph}  
\usepackage{graphicx}
\usetikzlibrary{positioning,arrows.meta,math,shapes.geometric,decorations.pathmorphing,decorations.pathreplacing,calc,trees}

\usepackage[font=small,labelfont=bf]{caption} %
\usepackage{subcaption}
\usepackage[percent]{overpic}

\usepackage[english]{babel}
\addto\extrasenglish{}
\addto\extrasenglish{}
\addto\extrasenglish{}

\usepackage{lpic}
\usepackage{algorithm, algpseudocodex}
\usepackage{xfrac}

\tikzset{v/.style={circle, draw, minimum size=3em, scale=1}}
\tikzset{shade/.style={fill=gray}}
\tikzset{farrow/.style={-{Latex},thick}}
\tikzset{garrow/.style={-{Latex},thick,dashed}}
\tikzset{circlearrow/.style={farrow,decorate,decoration={zigzag,amplitude=0.1em}}}

\allowdisplaybreaks

\makeatletter
\DeclareRobustCommand\bfseries{%
  \not@math@alphabet\bfseries\mathbf
  \fontseries\bfdefault\selectfont\boldmath}
\makeatother

\theoremstyle{definition}
\declaretheorem[name=Theorem, numberwithin=section]{theorem}

\declaretheorem[name=Lemma,sibling=theorem]{lemma}
\declaretheorem[name=Corollary, sibling=theorem]{corollary}
\declaretheorem[name=Claim,sibling=theorem]{claim}
\declaretheorem[name=Definition, sibling=theorem, style=definition]{definition}

\declaretheorem[name=Remark,sibling=theorem]{remark}
\declaretheorem[name=Example,sibling=theorem]{example}

\newenvironment{claimproof}[1][\proofname]
{\renewcommand\qedsymbol{$\diamond$}\proof[#1]}
{\endproof}

\makeatletter
\newcommand{\mylabel}[2]{\begingroup
\def\@itemlabel{\textbf{(#1)}}
\def\@currentlabel{\textbf{(#1)}}
\phantomsection\label{#2}
\hfill{\textbf{(#1)}}
\endgroup
}
\makeatother

\newcommand{\poly}{\textrm{\upshape poly}}

\newcommand{\polylog}{\textrm{\upshape polylog}}

\newcommand{\newclass}[2]{\newcommand{#1}{{\text{\upshape\sffamily #2}}\xspace}}
\renewcommand{\P}{{\text{\upshape\sffamily P}}\xspace}
\newclass{\NP}{NP}
\newclass{\AC}{AC}
\newclass{\coNP}{coNP}
\newclass{\FP}{FP}
\newclass{\TFNP}{TFNP}
\newclass{\TFAP}{TFAP}
\newclass{\TFSigma}{TF$\Sigma$}
\newclass{\PLS}{PLS}
\newclass{\PPA}{PPA}
\newclass{\PPAD}{PPAD}
\newclass{\PPADS}{PPADS}
\newclass{\PPP}{PPP}
\newclass{\PAP}{PAP}
\newclass{\SAP}{SAP}
\newclass{\PWPP}{PWPP}
\newclass{\CLS}{CLS}
\newclass{\EOPL}{EOPL}
\newclass{\SOPL}{SOPL}
\newclass{\UEOPL}{UEOPL}
\newclass{\RAMSEY}{RAMSEY}
\newclass{\BIRAM}{BIRAMSEY}
\newclass{\cA}{A}
\newclass{\cB}{B}
\newclass{\BPP}{BPP}
\newclass{\PLC}{PLC}
\newclass{\UPLC}{UPLC}
\newclass{\PiH}{PiH}
\newclass{\rwPHP}{rwPHP}
\newclass{\ZPP}{ZPP}
\newclass{\Lossy}{LOSSY}
\newclass{\E}{E}
\newclass{\APEPP}{APEPP}

\newcommand{\newprob}[2]{\newcommand{#1}{{\text{\upshape\scshape #2}}\xspace}}
\newprob{\leaf}{Leaf}
\newprob{\eol}{EoL}
\newprob{\eolLong}{End-of-Line}
\newprob{\sol}{SoL}
\newprob{\solLong}{Sink-of-Line}
\newprob{\iter}{Iter}
\newprob{\sod}{SoD}
\newprob{\sodLong}{Sink-of-Dag}
\newprob{\kkt}{KKT}
\newprob{\pA}{A}
\newprob{\pB}{B}
\newprob{\eopl}{EoPL}
\newprob{\eoplLong}{End-of-Potential-Line}
\newprob{\ueoplLong}{Unique-EoPL}
\newprob{\ueopl}{UEoPL}
\newprob{\eoml}{EoML}
\newprob{\eomlLong}{End-of-Metered-Line}
\newprob{\sopl}{SoPL}
\newprob{\soplLong}{Sink-of-Potential-Line}
\newprob{\pigeon}{Pigeon}
\newprob{\wpigeon}{WeakPigeon}
\newprob{\lonely}{Lonely}
\newprob{\reversiblepigeon}{RPigeon}
\newprob{\reversiblepigeonlong}{Reversible-Pigeon}
\newprob{\factoring}{Factoring}
\newprob{\nash}{Nash}
\newprob{\Or}{Or}
\newprob{\ramsey}{Ramsey}
\newprob{\biramsey}{BiRamsey}
\newprob{\sunflower}{Naive Sunflower}
\newprob{\wsol}{wSoL}
\newprob{\wsolLong}{weak-Sink-of-Line}
\newprob{\Search}{Search}

\newprob{\emptychild}{Empty-Child}
\newprob{\ECnR}{Empty-Child-no-Root}
\newprob{\ECwH}{Empty-Child-w-Height}
\newprob{\RECwH}{Relaxed-Empty-Child-w-Height}
\newprob{\REC}{Relaxed-Empty-Child}
\newprob{\pEC}{Promise-Empty-Child}
\newprob{\ppEC}{Path-Promise-Empty-Child}
\newprob{\ILP}{Nephew}
\newprob{\ILPwithInverse}{Nephew-w-Inverse}
\newprob{\ILPdegree}{Low-Degree-Nephew}
\newprob{\NDILP}{ND-Nephew}
\newprob{\NDILPwS}{ND-Nephew-w-Son}
\newprob{\BTreeLeaf}{Leaf-of-Rooted-Tree}
\newprob{\lossycode}{Lossy-Code}
\newprob{\lossycodep}{Bij-Lossy-Code}
\newprob{\BadTC}{Bad-2-Coloring}
\newprob{\wBadTC}{Weak-B2C}
\newprob{\AMGMLC}{AMGM-LC}
\newprob{\BEC}{Binary-Empty-Child}
\newprob{\DLO}{Dense-Linear-Ordering}
\newprob{\PHP}{PHP}
\newprob{\Avoid}{Avoid}
\newprob{\weopl}{Weak-End-of-Potential-Line}

\newcommand{\Enc}{\mathsf{Enc}}
\newcommand{\Dec}{\mathsf{Dec}}

\newcommand{\newdt}[2]{\newcommand{#1}{{\text{\upshape\ttfamily #2}}\xspace}}

\newdt{\FindChildren}{Find-Children}
\newdt{\FindChildrenAndParent}{Find-Children-and-Parent}
\newdt{\FindPair}{Find-Pair}
\newdt{\EndOfPath}{End-of-Path}
\newdt{\LeafOfPath}{Leaf-of-Path}
\newdt{\emptychildreduction}{Reduce-EC-to-Nephew}
\newdt{\emptychildreductionInv}{Reduce-EC-to-NwI}
\newdt{\Checksol}{CheckSol}
\newdt{\Traverse}{Traverse}

\newclass{\wSA}{wSA}

\newcommand{\APC}{\mathsf{APC}}

\def\calU{\mathcal{U}}
\def\med{\mathsf{med}}

\newcommand{\Jerabek}{Je\v{r}\'{a}bek\xspace}

\usepackage[T1]{fontenc}
\newcommand{\LL}{\mathtt{L}}
\newcommand{\RR}{\mathtt{R}}

\newcommand{\class}[1]{\textup{\textsf{#1}}\xspace}

\newcommand{\TFZPP}{\class{TFZPP}}

\newcommand{\N}{\mathbb{N}}

\definecolor{betterYellow}{RGB}{255,255,0}

\newcommand{\EmptyChilditem}{
    \item[$\emptychild$.] Given a set $V$ of vertices and three functions $F, L, R: V \rightarrow V$, where $F(u)$ is the father of $u$, and $L(u), R(u)$ are the left and the right child of $u$ respectively, a solution is one of the following.
    \begin{description}
        \item[s1.] $u \in V$ such that $F(L(u)) \neq u$ or $F(R(u)) \neq u$ or $L(u) = R(u) \neq u$; \hfill (Empty child)
        \item[s2.] $1$, if $L(1) = 1$ or $R(1) = 1$ or $F(1)\ne 1$. \hfill (Wrong root)
    \end{description}
}

\newcommand{\Nephewitem}{
\item[$\ILP$.] Given a set $V$ of vertices and two functions $f : V \rightarrow V$ and $g : V \rightarrow V$. Think of $f(v)$ as the \emph{father} of $v$ and $g(v)$ as the \emph{nephew} of $v$. A solution is one of the following.
\begin{description}

    \item[s1.] $v \in V$ such that $f(f(g(v))) \not= f(v)$ \hfill (your nephew's grandparent is not your parent)

    \item[s2.] $v \in V$ such that $f(g(v)) = v$\hfill (you are your nephew's parent)
\end{description}
}

\newlength{\nodeline}
\newlength{\arrowline}
\definecolor{color_gadget_PHP}{RGB}{219, 48, 122}
\colorlet{color_gadget_PHP_inner}{color_gadget_PHP!10!white}
\colorlet{color_gadget_PHP_label}{color_gadget_PHP!70!black}

\definecolor{color_gadget_SOPL}{RGB}{255, 153, 0}
\colorlet{color_gadget_SOPL_inner}{color_gadget_SOPL!10!white}
\colorlet{color_gadget_SOPL_label}{color_gadget_SOPL!70!black}

\definecolor{color_gadget_PATHPHP}{RGB}{0, 153, 255}
\colorlet{color_gadget_PATHPHP_inner}{color_gadget_PATHPHP!10!white}
\colorlet{color_gadget_PATHPHP_label}{color_gadget_PATHPHP!70!black}

\definecolor{color_gadget_ITER}{RGB}{153, 204, 0}
\colorlet{color_gadget_ITER_inner}{color_gadget_ITER!10!white}
\colorlet{color_gadget_ITER_label}{color_gadget_ITER!70!black}

\definecolor{color_gadget_EOPL}{RGB}{0, 153, 255}
\colorlet{color_gadget_EOPL_inner}{color_gadget_EOPL!10!white}
\colorlet{color_gadget_EOPL_label}{color_gadget_EOPL!70!black}

\tikzstyle{node} = [circle, line width=\nodeline, draw = black, fill = white, inner sep = 0mm, minimum size = 3.5mm]
\tikzstyle{node_small} = [node, circle, line width = 0.25mm, minimum size = 2.5mm]
\tikzstyle{solution} = [fill=red!90!,draw=black!50!red]
\tikzstyle{side} = [fill=Goldenrod,draw=Brown]
\tikzstyle{node_text} = [] %

\tikzstyle{node_regular} = [node]
\tikzstyle{node_regular_small} = [node_small]

\tikzstyle{node_solution} = [node, solution]
\tikzstyle{node_green} = [node, fill=Green!20!LimeGreen,draw=black!80!Green]
\tikzstyle{node_notice} = [node, draw = Green!20!LimeGreen, line width=2pt, fill=none, minimum size = 5.93mm,dotted]
\tikzstyle{node_solution_small} = [node_small,solution]

\tikzstyle{node_a} = [node, side, rectangle, minimum size = 4.3mm]
\tikzstyle{node_a_solution} = [node_a, solution]
\tikzstyle{node_a_small} = [node_small, side, rectangle, minimum size = 2.15mm]
\tikzstyle{node_a_solution_small} = [node_a_small,solution]

\tikzstyle{node_b} = [node, side, diamond, minimum size = 6.2mm]
\tikzstyle{node_b_small} = [node_a_small, diamond, minimum size = 3.1mm]

\tikzstyle{node_notice_small} = [node_notice, line width=2*\nodeline, minimum size = 8mm]

\tikzstyle{naive} = [minimum size = 4.5mm,rectangle]
\tikzstyle{node_regular_intro} = [node,fill=Gray!10!white]

\tikzstyle{edge} = [-{Latex[round]}, line width=\arrowline]

\tikzstyle{edge_regular} = [edge]
\tikzstyle{edge_regular_small} = [-{Latex[round]}, line width = 0.25mm]

\tikzstyle{edge_php} = [edge, color=color_gadget_PHP!70!black]
\tikzstyle{edge_php_small} = [edge_php, edge_regular_small]

\tikzstyle{edge_eopl} = [edge, color=color_gadget_EOPL!70!black]
\tikzstyle{edge_eopl_small} = [edge_eopl, edge_regular_small]

\tikzstyle{edge_iter} = [edge, color=color_gadget_ITER!70!black]
\tikzstyle{edge_iter_small} = [edge_iter, edge_regular_small]

\tikzstyle{edge_pathphp_small} = [line width = 0.25mm, -{Latex[round]}, decorate, decoration={snake, segment length=2.5mm, amplitude=1mm, pre length=7pt,post length=8pt}, shorten < = 3pt, shorten >=3pt, color=color_gadget_PATHPHP_label]

\tikzstyle{gadget} = [rounded corners, line width = 0.4mm, dashed]

\tikzstyle{gadget_PHP} = [gadget, draw = color_gadget_PHP, fill=color_gadget_PHP_inner]
\tikzstyle{gadget_PHP_small} = [gadget_PHP, line width = 0.2mm]

\tikzstyle{gadget_SOPL} = [gadget, draw = color_gadget_SOPL, fill=color_gadget_SOPL_inner]
\tikzstyle{gadget_SOPL_small} = [gadget_SOPL, line width = 0.2mm]

\tikzstyle{gadget_PATHPHP} = [gadget, draw = color_gadget_PATHPHP, fill=color_gadget_PATHPHP_inner]
\tikzstyle{gadget_PATHPHP_small} = [gadget_PATHPHP, line width = 0.2mm]

\tikzstyle{gadget_ITER} = [gadget, draw = color_gadget_ITER, fill=color_gadget_ITER_inner]
\tikzstyle{gadget_ITER_small} = [gadget_ITER, line width = 0.2mm]

\tikzstyle{gadget_EOPL} = [gadget, draw = color_gadget_EOPL, fill=color_gadget_EOPL_inner]
\tikzstyle{gadget_EOPL_small} = [gadget_EOPL, line width = 0.2mm]

\begin{document}

\begin{center}
{\huge Total Search Problems in \ZPP}
\\[15mm] \large

\setlength\tabcolsep{1em}
\begin{tabular}{cccc}
Noah Fleming&
Stefan Grosser&
Siddhartha Jain&
Jiawei Li\\[-1mm]
\small\slshape Lund \& Columbia  &
\small\slshape McGill &
\small\slshape UT Austin&
\small\slshape UT Austin

\end{tabular}
\begin{tabular}{ccc}\\[0.1em]
Hanlin Ren & Morgan Shirley & Weiqiang Yuan
\\[-1mm]
\small\slshape IAS  &
\small\slshape Lund  &
\small\slshape EPFL  
\end{tabular}

\vspace{9mm}

\large
\today

\vspace{5mm}

\normalsize
\end{center}
\begin{abstract}
    We initiate a systematic study of ${\sf TFZPP}$, the class of total ${\sf NP}$ search problems solvable by polynomial time randomized algorithms. ${\sf TFZPP}$ contains a variety of important search problems such as \textsc{Bertrand-Chebyshev} (finding a prime between $N$ and $2N$), refuter problems for many circuit lower bounds, and $\textsc{Lossy-Code}$. The $\textsc{Lossy-Code}$ problem has found prominence due to its fundamental connections to derandomization, catalytic computing, and the metamathematics of complexity theory, among other areas.
    
    While ${\sf TFZPP}$ collapses to ${\sf FP}$ under standard derandomization assumptions in the white-box setting, we are able to separate ${\sf TFZPP}$ from the major ${\sf TFNP}$ subclasses in the black-box setting. In fact, we are able to separate it from every uniform ${\sf TFNP}$ class assuming that $\NP$ is not in quasi-polynomial time. To do so, we extend the connection between proof complexity and black-box ${\sf TFNP}$ to randomized proof systems and randomized reductions.
    
    Next, we turn to developing a taxonomy of ${\sf TFZPP}$ problems. We highlight a problem called $\textsc{Nephew}$, originating from an infinity axiom in set theory. We show that $\textsc{Nephew}$ is in $\mathsf{PWPP}\cap \mathsf{TFZPP}$ and conjecture that it is not reducible to $\textsc{Lossy-Code}$. Intriguingly, except for some artificial examples, most other black-box ${\sf TFZPP}$ problems that we are aware of reduce to $\textsc{Lossy-Code}$:
    \begin{itemize}
        \item We define a problem called $\textsc{Empty-Child}$ capturing finding a leaf in a rooted (binary) tree, and show that this problem is \emph{equivalent} to $\textsc{Lossy-Code}$. We also show that a variant of $\textsc{Empty-Child}$ with ``heights'' is complete for the intersection of $\sf SOPL$ and $\textsc{Lossy-Code}$.
        \item We strengthen $\textsc{Lossy-Code}$ with several combinatorial inequalities such as the AM-GM inequality. Somewhat surprisingly, we show the resulting new problems are still reducible to $\textsc{Lossy-Code}$. A technical highlight of this result is that they are proved by \emph{formalizations in bounded arithmetic}, specifically in Jeřábek's theory $\mathsf{APC}_1$ (JSL 2007).
        \item Finally, we show that the \textsc{Dense-Linear-Ordering} problem reduces to $\textsc{Lossy-Code}$.
    \end{itemize}

\end{abstract}

\newpage
\setcounter{tocdepth}{2}
\tableofcontents
\newpage

\section{Introduction}

Total search problems are abundant in theoretical computer science. The formal study of these problems has been highly impactful to a wide range of areas including game theory~\cite{DGC09, DBLP:conf/focs/ChenD06}, cryptography~\cite{Hubacek20, Foletal24, BGS25}, proof complexity~\cite{beame1995relative, DavisR23, PR24, GHJ+22,Thapen24, GoosKRS19, LiLR24, FlemingMD25}, and recently in the study of explicit construction problems and derandomization \cite{DBLP:conf/innovations/KleinbergKMP21, Kor21, Korten25}. Central to the latter has been the \emph{Range Avoidance} (or simply $\Avoid$) problem.

\begin{description}
\item[$\Avoid$.]  Given a circuit $D:\{0,1\}^{n-1} \rightarrow \{0,1\}^{n}$, find $x \in \{0,1\}^{n}$ such that for every $y\in \{0,1\}^{n-1}$, $D(y) \neq x$.
\end{description}

$\Avoid$ captures the explicit construction problems for many combinatorial objects whose existence follows from the probabilistic method. Notable examples include functions with high circuit complexity, rigid matrices, Ramsey graphs, strong error correcting codes, and many more~\cite{Kor21, Jerabek07, DBLP:journals/toct/GuruswamiLW25, DBLP:conf/approx/GajulapalliGNS23}. By developing algorithms for $\Avoid$, a line of work has shown circuit lower bounds against a variety of classes~\cite{RenSW22, ChenHLR23, CHR24, Li24}. 

$\Avoid$ belongs to the class $\mathsf{TF}\Sigma_2^P$, the second level of the total function polynomial hierarchy. If one \emph{Herbrandizes}\footnote{Herbrandization is a basic construction in logic; see \autoref{appendix: Herbrandization} for a short overview.} the $\mathrm{Avoid}$ problem, then one obtains its $\TFNP$ sibling, the \lossycode problem (see \cite{Korten25} for a survey). This problem asks to find an element that is not in the range of a pair of compressing and decompressing maps $C$ and $D$.

\begin{description}
\item[$\lossycode$.]  Given a pair of circuits $C:\{0,1\}^n \rightarrow \{0,1\}^{n-1}$ and $D:\{0,1\}^{n-1} \rightarrow \{0,1\}^{n}$, find $x \in \{0,1\}^{n}$ such that $D(C(x)) \neq x$.
\end{description}
$\lossycode$ was originally defined by Je{\v r}\'abek in \cite{Jerabek07, Jerabek-independence}, under the name \emph{retraction weak pigeonhole principle}, showing that it is equivalent to the set of problems whose totality is provable in $\APC_1$. Since then, it has been considered predominantly through the lens of bounded arithmetic as a $\TFNP$ problem and as a combinatorial principle \cite{Muller21, KolodziejczykT22}.
Korten~\cite{DBLP:conf/coco/Korten22} asked to understand the set of $\TFNP$ problems that are reducible to $\lossycode$. 
Besides being an interesting problem on its own, $\lossycode$ also arises naturally in a few other places, further motivating its study: %
\begin{itemize}
    \item {\bf Derandomization.} In the recent \emph{certified derandomization} framework~\cite{PyneRZ23}, the derandomization algorithm is required to either output the correct answer, or report that the underlying circuit lower bound assumption is false by providing a small circuit violating the assumption. It turns out that certified derandomization is characterized by $\lossycode$~\cite{DBLP:conf/focs/LiPT24}.
    
    Such derandomization ideas are particularly explored in the context of \emph{catalytic computing}~\cite{DBLP:conf/stoc/BuhrmanCKLS14, DBLP:journals/eatcs/Mertz23} in a framework known as ``compress-or-random''~\cite{Pyne24, CookLMP25, KMPS25, AgarwalaM25}: If the contents of the catalytic tape is ``incompressible'' (which usually means that it is not a solution of a certain $\lossycode$ instance), then it can be used for derandomization; otherwise we can compress the catalytic tape and obtain more free space. As a result, although we are currently unable to prove that $\mathsf{CL}$ (catalytic logspace) is in $\P$, we can show that $\mathsf{CL}$ reduces to $\lossycode$~\cite{CookLMP25}.
    \item {\bf Metamathematics of complexity theory.} It turns out that (variants of) $\lossycode$ captures the complexity of many \emph{refuter} problems~\cite{CJSW24}, which are natural total search problems reflecting the metamathematical complexity of proving lower bounds. Many lower bounds in circuit complexity and communication complexity have refuter problems equivalent to $\lossycode$~\cite{DBLP:conf/coco/Korten22, CLO24}, and the refuter complexity for some proof complexity lower bounds is captured by variants of $\lossycode$ as well~\cite{LiLR24}.
    \item {\bf Bounded arithmetic.} A basic theory of bounded arithmetic for approximate counting and reasoning about randomized computation is $\APC_1$, developed in a series of papers by Jeřábek~\cite{Jerabek04, Jerabek-phd, Jerabek07}. Wilkie's witnessing theorem~\cite{Thapen-PhD, Jerabek04} implies that $\lossycode$ is ``complete'' for $\APC_1$ in the following sense: $\APC_1$ proves the totality of $\lossycode$, and every $\TFNP$ problem provably total in $\APC_1$ reduces to $\lossycode$.
\end{itemize}

$\lossycode$ belongs to the class $\TFZPP$, the subclass of $\TFNP$ containing the total search problems that admit polynomial-time randomized algorithms, introduced in \cite{buresh2006tfnp}.
Since we are dealing with total $\NP$ search problems, every randomized algorithm that may make mistakes can be turned into one that does not make any mistakes.\footnote{This was observed by \Jerabek~\cite{Jerabek16}; in his terminology, we have $\sf TFRP = TFZPP$.} Hence, it seems that $\TFZPP$ is the only natural (semantic) subclass of $\sf TFNP$ capturing randomized polynomial time. %

Besides $\lossycode$, there are a variety of important total search problems that sit inside $\TFZPP$. We list two of them that we think reflect the importance of $\TFZPP$ the best:

\begin{mdframed}[innertopmargin=0]\footnotesize
\begin{example}
    The Bertrand--Chebyshev theorem states that for every integer $N\ge 1$, there is a prime number between $N$ and $2N$. This motivates the following total search problem called \textsc{Bertrand-Chebyshev}: Given an integer $N$ (represented in binary), output a prime number between $N$ and $2N$. In fact, the Prime Number Theorem implies that there are $\Theta(N/\log N)$ such prime numbers, and the AKS primality test~\cite{AKS04} provides a deterministic method for verifying solutions, hence \textsc{Bertrand-Chebyshev} is in $\TFZPP$.
    
    The complexity of \textsc{Bertrand-Chebyshev} remains elusive. Unless one makes strong assumptions such as Cramér's conjecture~\cite{cramer1936order} or $\P = \BPP$~\cite{ImpagliazzoW97}, the best known deterministic algorithm needs to spend $\approx N^{1/2}$ time~\cite{LagariasO87, BHP01} (which is \emph{exponential} in the input length). Improving this time bound was exactly the focus of the Polymath 4 project~\cite{TaoCH12}; however, despite much effort, no unconditional progress was made. This problem is also the flagship problem in the study of \emph{pseudodeterministic algorithms}~\cite{GatG11, OliveiraS17, LuOS21, ChenLORS23}.
    
    On the complexity-theoretic side, the only upper bound known for $\textsc{Bertrand-Chebyshev}$ is that it reduces to $\lossycode^{\textsc{Factoring}}$, i.e., the $\lossycode$ problem where both input circuits $C, D$ have access to a $\textsc{Factoring}$ oracle~\cite{ParisWW88, DBLP:conf/coco/Korten22}. It is unclear if it belongs to any standard $\TFNP$ subclasses such as $\PLS$ or $\PPA$~\cite{GoldbergP18, GhentiyalaLi25}. %
\end{example}
\end{mdframed}

\def\ckt{\mathscr{C}}
\def\Refuter{\textsc{Refuter}}
\begin{mdframed}[innertopmargin=0]\footnotesize
\begin{example}
    A family of important total search problems is \emph{refuter problems}~\cite{CLW20, CJSW24, CTW23} for complexity lower bounds. Let $\ckt$ be a circuit class and $L$ be a hard problem for $\ckt$, the refuter problem, $\Refuter(L\not\in\ckt)$, is the following total search problem: given a small $\ckt$ circuit $C$ attempting to compute $L$, the goal is to output an instance $x$ such that $C(x) \ne L(x)$. The complexity of these refuter problems are closely related to the provability of complexity lower bounds~\cite{CLO24, LiLR24}.
    
    Such refuter problems are often in $\TFZPP$: In fact, if $L$ is \emph{average-case hard} against $\ckt$ (and both $\ckt$ and $L$ are in polynomial-time), then $\Refuter(L\not\in\ckt)\in\TFZPP$ as the algorithm for the refuter problem can repeatedly sample inputs from the hard distribution until it finds a solution $x$ where $C(x) \ne L(x)$. Even though average-case lower bounds against $\AC^0[p]$ circuits have been proved for nearly 40 years~\cite{Razborov87, Smolensky87}, we are not aware of any non-trivial $\TFNP$ upper bound for the problem $\Refuter(\mathrm{MAJ}\not\in\AC^0[2])$.
    
    Another example is when $L = \textrm{search-SAT}$ and $\ckt$ is the family of polynomial-size circuits: Given a circuit $C$ attempting to solve $\textrm{search-SAT}$, the refuter problem asks to find a formula $\varphi$ (along with a satisfying assignment $a$ of $\varphi$) such that $C(\varphi)$ fails to satisfy $\varphi$. The complexity of this problem is of significant interest to the bounded arithmetic community~\cite{Krajicek95, BussSurvey, Pich15, PichS23} as its hardness would imply the unprovability of $\NP\not\subseteq\P/_\poly$. However, this problem is in $\TFZPP$ under the existence of one-way functions against non-uniform adversaries, hence it is unclear how its complexity sheds light on the aforementioned unprovability question. %
\end{example}
\end{mdframed}

Finally, an additional motivation for studying $\TFZPP$ is its connection to $\Avoid$ and $\APEPP$ (the class of total search problems mapping reducible to $\Avoid$~\cite{DBLP:conf/innovations/KleinbergKMP21}): it is the ``projection'' of $\Avoid$ to $\TFNP$ in the following sense: %

\begin{restatable}{theorem}{avoidIntersect}
\label{thm:avoidIntersect}
    $\TFZPP=\TFNP \cap \APEPP$.
\end{restatable}

\subsection*{Our Contributions.}
In this work, we initiate a formal study of $\TFZPP$ as a class of total search problems. Analogous to the setting of decision problems, we expect that $\TFZPP = \FP$. Indeed, this follows from the same assumption used in \cite{ImpagliazzoW97}---namely that $\E$ requires circuits of exponential size. Moreover, $\TFZPP=\FP$ appears to be weaker than a full derandomization of $\BPP$.

However, we show that this is not the case in the \emph{black-box} setting in a very strong sense. In the black-box setting one only has access to the input via an oracle; black-box classes are denoted by a $dt$ superscript (for ``\emph{d}ecision \emph{t}rees''). We say that a $\TFNP^{dt}$ class is \emph{uniformly generated} if it has a complete problem $R=\{R_n\}_{n \in \mathbb{N}}$ such that there is a Turing Machine which on input $1^n$ outputs $R_n$ in polynomial time. Note that all of the major $\TFNP^{dt}$ subclasses in the literature are uniformly generated. Under the assumption that $\NP$ is not in quasi-polynomial time ($\mathsf{QP}$), we show that no uniformly generated $\TFNP^{dt}$ class contains $\TFZPP^{dt}$.

\begin{theorem}
\label{thm:sepALL}
    $\TFZPP^{dt} \not \subseteq {\cal C}$ for every uniformly generated class ${\cal C} \subseteq \TFNP^{dt}$, unless $\NP \subseteq \mathsf{QP}$.   
\end{theorem}

 To prove these separations, we employ a close connection between total search problems and proof complexity \cite{beame1995relative, GoosKRS19, BussFI23, FlemingMD25}. This connection shows that, in the black-box setting, a search problem belongs to a class if and only if an associated proof system can prove the totality of that search problem. In this case, we say that the class is \emph{characterized} by that proof system. To prove our separations, we first show that $\TFZPP^{dt}$ is characterized by the random tree-like resolution proof systems of Buss et al.~\cite{BussKT14}.

 \begin{theorem}
\label{thm:TFZPPChar}
    $\TFZPP^{dt}$ is characterized by random tree-like resolution. 
\end{theorem}
More generally, we show that if a class $\cal C$ of total search problems is characterized by a proof system $\Pi$, then the class of problems that are efficiently \emph{randomized} reducible to a complete problem in $\cal C$ is characterized by the proof system random $\Pi$. \autoref{thm:sepALL} then follows by combining the following two results:  (1) Buss et al.~\cite{BussFI23} showed that every uniformly generated $\TFNP$ class has a characterizing proof system, and (2) Pudl\'ak and Thapen \cite{PudlakT19} showed that a propositional proof system simulating random tree-like resolution would imply faster algorithms for $\NP$. 

\paragraph{A Highly Unsatisfiable Cook-Reckhow Program.}
\autoref{thm:sepALL} is striking as it suggests that $\TFZPP$ problems can be arbitrarily hard in the black-box model. This motivates an interesting direction of research: find explicit $\TFZPP$ problems that are hard for stronger and stronger subclasses of $\TFNP^{dt}$. By the close connection between $\TFNP^{dt}$ and proof complexity \cite{BussFI23}, this can be seen as a Cook-Reckhow program for highly unsatisfiable formulas: for increasingly more expressive proof systems, exhibit a highly unsatisfiable CNF formula which is hard for that system. %

Towards this program, we provide explicit separations of $\TFZPP^{dt}$ from every major $\TFNP$ class defined in the 1990s~\cite{DBLP:journals/jcss/JohnsonPY88, DBLP:journals/jcss/Papadimitriou94}. Note that all those $\TFNP$ classes are contained in $\PLS$, $\PPP$, and $\PPA$. The separation from $\PPA$ was shown by Beame et al.~\cite{beame1995relative}. Leveraging the recent work of Hopkins and Lin~\cite{HopkinsL22}, we are able to show the following.

\begin{theorem}
\label{thm:sepExplicit}
    There are \emph{explicit} (polynomial-time constructable) problems in $\TFZPP^{dt}$ which are not in either $\PPP^{dt}$ or $\PLS^{dt}$.
\end{theorem}

\subsection*{A $\TFZPP$ Zoo.} 
We now turn to studying the structure of problems inside of $\TFZPP$.  \autoref{fig:tfzpp-zoo} shows the  zoo of problems within $\TFZPP$ that we consider, as well as their relationships. 

Like $\ZPP$, $\TFZPP$ is a semantic class and therefore it is unlikely to admit complete problems unless $\FP = \TFZPP$. However, we observe an interesting phenomenon: almost every $\TFZPP$ problem that has been studied in the literature is reducible to $\lossycode$!\footnote{One exception is the Bertrand--Chebyshev problem for which it is not known whether it is reducible to $\lossycode$. However, it is unclear how to define the Bertrand--Chebyshev problem in the \emph{black-box} model, and we are unable to separate any natural black-box $\TFZPP$ problem from $\lossycode$.}  This raises the question: are there ``natural'' $\TFZPP$ problems which are not reducible $\lossycode$ in the black-box setting, and what do they look like? \autoref{thm:sepALL} and \autoref{thm:sepExplicit} already provide examples of problems that are not reducible to $\lossycode$; however, we do not consider these problems natural---we are looking for problems that would be studied outside of the context of proving such separations. 

While we are unable to resolve this question---indeed, many of our conjectured separating examples turned out to be reducible to $\lossycode$ in surprising ways!---we provide natural $\TFZPP$ problems which we conjecture witness a separation, and which we believe are of independent interest. As well, we show several surprising reductions to $\lossycode$.%

\begin{figure}
    \centering
    \begin{tikzpicture}[scale=1.1]
\tikzset{inner sep=0,outer sep=3}

\tikzstyle{a}=[inner sep=6pt, inner ysep=6pt,outer sep=0.5pt,
draw=black!20!white, fill=Cerulean!2!white, very thick, rounded corners=6pt, align=center]
\tikzstyle{b}=[inner sep=4pt, inner ysep=4pt,outer sep=0.5pt,
draw=black!20!white, fill=Cerulean!2!white, thick, align=center]

\begin{scope}[yscale=1.145]
    \large
    
    \node[a] (LOSSY) at (3,2) {\small\upshape{Lossy-Code}};
    \node[a] (DLO) at (0.6,-0.2) {\small\upshape{Dense-Linear-Ordering}};
    \node[a] (INJLOSSY) at (3, -1) {\small\upshape{Bij-Lossy-Code}};
    \begin{scope}[local bounding box=amgmbox]
        \node[a] (AMGM) at (-2,3) {\small\upshape{AMGM-LC}};
        \node[a] (DualAMGM) at (-2, 1) {\small\upshape{Dual-AMGM-LC}};
        \node[a] (IncExc) at (-2, 2) {\small\upshape{Inclusion-Exclusion}};
    \end{scope}
    \begin{scope}[local bounding box=nephewbox]
        \node[a] (PWPP) at (2,6) {$\PWPP$};
        \node[a] (NEPHEW) at (3,4.5) {\upshape{Nephew}};
        \node[a] (TFZPP) at (4,6) {$\TFZPP$};
    \end{scope}
    \node[a] (PPADS) at (-1.5,6) {\PPADS};
    \node[a] (PPP) at (0,7) {$\PPP$};
    \node[a] (TFNP) at (3,8) {$\TFNP$};
    \begin{scope}[local bounding box=ecbox]
        \begin{scope}[local bounding box=nwibox]
            \node[a] (NEPHEWwSon) at (6,3) {\footnotesize\upshape{Nephew-w-Inverse}};
        \end{scope}
        \node[a] (EMPTYC) at (6, 2) {\small\upshape{Empty-Child}};
        \begin{scope}[local bounding box=ecwhbox]
            \node[a] (EMPTYwH) at (9, 1) {\footnotesize\upshape{Empty-Child-w-Height}\\ \footnotesize($=\Lossy\cap\mathsf{SOPL}$)};
        \end{scope}
        \begin{scope}[local bounding box=becbox]
            \node[a] (BEC) at (6, 0) {\footnotesize \upshape{Binary-Empty-Child}};
        \end{scope}
        \begin{scope}[local bounding box=becwhbox]
            \node[a] (BECwH) at (9, -1) {\footnotesize\upshape{Binary-Empty-Child-w-Height}};
        \end{scope}
    \end{scope}
    \node[a] (PLS) at (6,7) {$\PLS$};
    \node[a] (SOPL) at (9, 6) {$\mathsf{SOPL}$};
    \node[b] (TFZPPinfo) at (6.7,5.2){
        \small \textcolor{MidnightBlue}{$=$ Random TreeRes}\\
        \small \textcolor{YellowOrange}{$=\TFNP\cap\mathsf{APEPP}$}
    };
    \node[below=0 of TFZPPinfo] {\footnotesize{\cref{sec: prelim,sec: tfzpp general results}}};
\end{scope}

\path[-{Stealth[length=6pt]},line width=.6pt,gray]
(NEPHEW) edge (PWPP)
(NEPHEW) edge (TFZPP)
(PWPP) edge (PPP)
(TFZPP) edge (TFNP)
(LOSSY) edge (NEPHEW)
(DLO) edge node [above, sloped] {\textcolor{black}{\footnotesize\autoref{sec: DLO}}} (LOSSY)
(AMGM) edge (LOSSY)
(LOSSY) edge (AMGM)
(DualAMGM) edge (LOSSY)
(LOSSY) edge (DualAMGM)
(IncExc) edge (LOSSY)
(LOSSY) edge (IncExc)
(EMPTYwH) edge (EMPTYC)
(EMPTYC) edge (LOSSY)
(LOSSY) edge (EMPTYC)
(EMPTYwH) edge (SOPL)
(BEC) edge (EMPTYC)
(BECwH) edge (EMPTYwH)
(BECwH) edge (BEC)
(INJLOSSY) edge (LOSSY)
(INJLOSSY) edge (BEC)
(SOPL) edge (PLS)
(PLS) edge (TFNP)
(LOSSY) edge (PPADS)
(PPADS) edge (PPP)
(PPP) edge (TFNP)
(NEPHEWwSon) edge (NEPHEW)
(EMPTYC) edge (NEPHEWwSon)
(NEPHEWwSon) edge (EMPTYC)
(TFZPP) edge[bend left](TFZPPinfo);

\newcommand{\zoosep}{0.15}

\draw[thick, dotted] 
    ($(nephewbox.south west) + (-\zoosep,-\zoosep)$) -- 
    ($(nephewbox.north west) + (-\zoosep,\zoosep)$) -- 
    node[above, near start, align=center]{\footnotesize \autoref{sec: Nephew}} 
    ($(nephewbox.north east) + (\zoosep,\zoosep)$) -- 
    ($(nephewbox.south east) + (\zoosep,-\zoosep)$) -- 
    ($(nephewbox.south west) + (-\zoosep,-\zoosep)$);

\draw[thick, dotted] 
    ($(amgmbox.north west) + (-\zoosep,\zoosep)$) -- 
    ($(amgmbox.south west) + (-\zoosep,-\zoosep)$) -- 
    node[below, align=center] {\footnotesize \autoref{sec: APC1}} ($(amgmbox.south east) + (\zoosep,-\zoosep)$) -- 
    ($(amgmbox.north east) + (\zoosep,\zoosep)$) -- 
    ($(amgmbox.north west) + (-\zoosep,\zoosep)$);

\draw[thick, dotted] 
    let \p1=(becwhbox.south west), \p2=(becbox.south west), \p3=(ecbox.south east), 
        \p4=(ecwhbox.north east), \p5=(nwibox.north east) in 
    ($(ecbox.north west) + (-\zoosep,\zoosep)$) -- 
    ($(\p2) + (-\zoosep,-\zoosep)$) -- 
    ($(\x1, \y2) + (-\zoosep,-\zoosep)$) -- 
    ($(\p1) + (-\zoosep,-\zoosep)$) -- 
    ($(\p3) + (\zoosep,-\zoosep)$) -- 
    ($(\x3, \y4) + (\zoosep,\zoosep)$) -- 
    ($(\x5, \y4) + (\zoosep,\zoosep)$) -- 
    ($(\p5) + (\zoosep,\zoosep)$) -- 
    node[above, align=center] {\footnotesize\autoref{sec: Empty-Child}} ($(ecbox.north west) + (-\zoosep,\zoosep)$);

\small
\hypersetup{hidelinks}
\tikzset{new/.style={-{Stealth[length=6pt]},dashed,line width=1pt,YellowOrange}}

\end{tikzpicture}
    \caption{The $\TFZPP$ zoo.}
    \label{fig:tfzpp-zoo}
\end{figure}
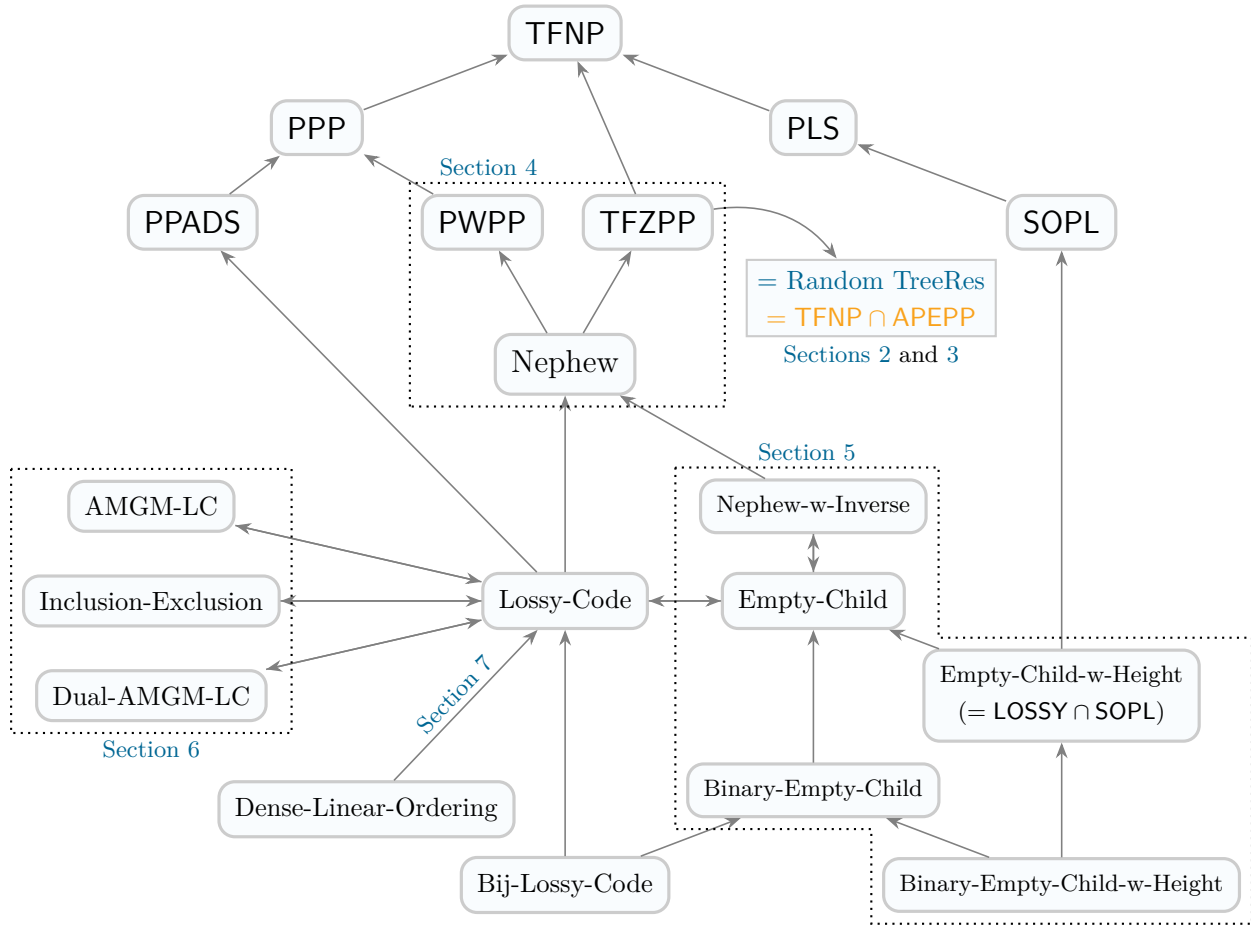

\paragraph{Nephew.} Our primary candidate is the following.

\begin{description}
    \Nephewitem
\end{description}

It may not be immediately obvious that $\ILP$ is a total search problem. A proof may be found in the textbook of Börger, Grädel, and Gurevich~\cite[Proposition 6.5.5]{Borger01}; we sketch the argument here. Create a directed graph $G_f$ with vertex set $V$ and with an edge from $u$ to $v$ if $f(u) = v$. Then one can assign to each vertex a ``level'' that represents its distance to the core\footnote{We intentionally leave ``core'' undefined for this brief sketch.} of $G_f$. Let $v^*$ be a vertex with maximum level $\ell_{\text{max}}$. However, if $v^*$ is not a solution to $\ILP$, it must be that $g(v^*)$ has level $\ell_{\text{max}} + 1$, a contradiction. A similar intuition will be used in \cref{sec: Nephew} for other proofs involving $\ILP$.

$\ILP$ is derived from the minimal \emph{axioms of infinity} in model theory. %
One method for constructing a total search problem is to begin with a sentence in logic that has an infinite model but no finite model. Such a sentence is known as an \emph{axiom of infinity}, since any model for it must be infinite. Axioms of infinity are classified under the number of quantifiers and predicate and function symbols of certain arity that they contain, and there are 10 minimal classes (\cite[Theorem 6.5.4]{Borger01}). Each of these classes corresponds to a total search problem, and in most cases, this problem is complete for a well-known $\TFNP$ class. The only exception is the one to which $\ILP$ belongs; $\ILP$ can be interpreted as the Herbrandization of
\begin{equation} \label{eq: ilp-from-borger}
    \forall x \exists y \left( F(F(y)) = F(x) \land F(y) \neq x \right).
\end{equation}

The proof that $\ILP$ belongs to $\TFZPP$ is highly non-trivial. Furthermore, our best upper bound on the complexity of $\ILP$ is that it is contained within $\PWPP$, the problems reducible to the weak pigeonhole principle, a relaxation of the $\lossycode$ problem.
\begin{restatable}{theorem}{inTFZPPandPWPP} \label{thm: tfzpp-and-pwpp-containment}
    $\ILP \in\TFZPP \cap \PWPP$.
\end{restatable}

To obtain evidence that $\ILP$ is not reducible $\lossycode$, we attempt to isolate the potential hardness in $\ILP$. In doing so, we define a number of natural problems in $\TFZPP$ which may be of independent interest.

The proof that $\ILP$ is in $\TFZPP$ proceeds by reducing it to the observation that a leaf in a binary tree can be found in logarithmic time in expectation. We define a total search problem whose membership in $\TFZPP$ formalizes this observation.

\begin{description}
    \EmptyChilditem
\end{description}

Surprisingly, $\emptychild$ is equivalent to $\lossycode$ under decision tree reductions, denoted $=_{dt}$. Thus, if $\ILP$ is indeed not reducible to $\lossycode$, the hardness of $\ILP$ does not come from this portion of the reduction.

\begin{theorem}\label{thm:intro_lossy_empty_eq}
    $\emptychild =_{dt} \lossycode$.
\end{theorem}

As a warm-up to the techniques needed to prove this theorem, we consider a variant of $\emptychild$ which includes an additional ``height'' function that outputs the height of a given node in the tree. We show that this is a complete problem for the class $\Lossy \cap \PLS = \Lossy \cap \SOPL$, where $\Lossy$ is the class of problems efficiently reducible to $\lossycode$. The proof resembles previous intersection theorems from $\TFNP$ \cite{FearnleyGHS23, GoosHJMPRT24}.

Then, we use $\emptychild$ as an intermediate problem to study the relationship between $\ILP$ and $\lossycode$.

\begin{theorem}
    $\emptychild \leq_{dt} \ILP$.
\end{theorem}
Thus, combining with Theorem \ref{thm:intro_lossy_empty_eq}, we have $\lossycode$ reduces to $\ILP$.

\paragraph{AM-GM Lossy-Code.} One of our original (and failed) candidates for separation from $\lossycode$ was a problem called \emph{AM-GM Lossy-Code}, obtained by combining $\lossycode$ itself with the \emph{AM-GM Inequality}: $\frac{a+b}{2} \ge \sqrt{ab}$. This problem was inspired by the \textsc{Bad $2$-Coloring} problem in~\cite{PPY23}: Given an undirected graph $G = (V, E)$ with $|V| = 2N$ vertices and $|E| = N^2 + 1$ edges along with a $2$-coloring $C: V \to \{0, 1\}$ of $V$, find an edge $(x, y)\in E$ that is not colored properly. This problem is total exactly because of the AM-GM inequality: suppose there are $a$ black vertices and $b=2N-a$ white vertices and every edge is colored properly, then there are at most $ab \le \mleft(\frac{a+b}{2}\mright)^2 = N^2$ edges, contradicting $|E| > N^2$.

To compose this problem with $\lossycode$, we need to make two adaptations: First, to put it inside $\TFZPP$, the number of edges needs to be \emph{much} larger than $N^2$, say $(1+\varepsilon)N^2$; second, we are given a function $F$ from $[(1+\varepsilon)N^2]$ to the set of properly colored edges and its purported inverse $G$ and we need to find $x\in [(1+\varepsilon)N^2]$ such that $G(F(x)) \ne x$. We arrive at the following problem:

\begin{description}
\item [$c$-\AMGMLC.]
        Let $c > 1$ be a constant, $V := [2N]$ and $P := [c\cdot N^2]$.  
        The input is a coloring function $C : V \rightarrow \{0,1\}$ and two mappings $F : P \rightarrow V \times V$, $G: V \times V \rightarrow P$. Let $H := C^{-1}(0) \times C^{-1}(1)$. The goal is to find solutions of either type:

        \begin{description}
            \item[s1.] a pigeon $x \in P$ such that $G(F(x)) \neq x$; \hfill (Wrong Encoding-Decoding)
            \item[s2.] a pigeon $x \in P$ such that $F(x) \notin H$; \hfill (Invalid Hole)
        \end{description}
\end{description}

Indeed, it seems unclear how to massage a $\lossycode$ instance of the shape $[cN^2]\rightleftharpoons H$ into a standard $\lossycode$ instance of the form $[2M]\rightleftharpoons [M]$, even though the AM-GM inequality implies that $cN^2 \gg |H|$. However, it turns out that such a reduction is possible (although highly non-trivial)! We encourage the reader to take a moment to think about how to reduce $\AMGMLC$ to $\lossycode$.

\begin{restatable}{theorem}{AMGMreducesLC}
\label{thm:AMGMreducesLC}
   For every constant $c > 1$, $c$-\AMGMLC$\leq_{dt}$~\lossycode.
\end{restatable}

Our reduction goes through \emph{bounded arithmetic}: We formalize the totality of $c$-$\AMGMLC$ in \Jerabek's theory $\APC_1$~\cite{Jerabek07}, and Wilkie's witnessing theorem for $\APC_1$~\cite{Thapen-PhD, Jerabek04} implies a reduction from $c$-$\AMGMLC$ to $\lossycode$. In particular, like every formalization in $\APC_1$, our reduction makes use of the \emph{Nisan--Wigderson generator}~\cite{NisanW94}.

Moreover, the techniques underlying $\APC_1$~\cite{Jerabek07} allow us to reduce problems to $\lossycode$ in a \emph{systematic} way; we provide two additional examples later ($c$-Dual-\AMGMLC in \autoref{sec: dual AM GM LC} and a problem capturing the Inclusion-Exclusion principle in \autoref{sec: incl excl}). A secondary goal of expounding these reductions is to introduce the ideas of $\APC_1$ to audiences who are less familiar with bounded arithmetic.

\begin{remark}
    Unfortunately, it seems unclear how to formalize the totality of $\ILP$ in $\APC_1$, hence the bounded arithmetic approach does not seem to provide a reduction from $\ILP$ to $\lossycode$. For example, our proof that $\ILP \in \TFZPP$ requires reasoning about the \emph{level} of each node, which seems to be global reasoning that is infeasible in $\APC_1$.
\end{remark}

\paragraph{Linear Ordering.}
Finally, we consider one more natural problem in $\TFZPP$. The Linear Ordering Principle has a storied history in proof complexity~\cite{Krishnamurthy85, BonetG01, Potechin20} and bounded arithmetic~\cite{DBLP:journals/jsyml/ChiariK98, Hanika-PhD, BussKT14, DBLP:journals/tocl/AtseriasT14}. It has been studied in the context of total search problems as well~\cite{KortenP24, HV25}, where it was used in order to construct new algorithms for $\Avoid$. A line of works in proof complexity \cite{Riis01, AtseriasD08, Gryaznov19, DBLP:conf/focs/ConnerydRNPR23} has also considered a dense variant of this problem defined as follows, which lies in $\TFZPP$.

\begin{description}
\item[$\DLO$.]   The input consists of the descriptions of a linear ordering $\prec$ over $N$ elements and a \emph{median function} $\med: [N]\times [N] \to [N]$. Without loss of generality, we may assume that for $x\ne y \in [N]$, exactly one of $(x\prec y)$ and $(y\prec x)$ is true, and that $\med(x, y) = \med(y, x)$. (That is, $\prec$ is represented by a string of $\binom{N}{2}$ bits and $\med$ is represented by a list of $\binom{n}{2}$ elements in $[N]$.) A solution is one of the following.
        \begin{description}
            \item[s1.] $x, y, z \in [N]$ such that $x\prec y$, $y\prec z$, and $z\prec x$; or \hfill (Transitivity violation)
            \item[s2.] $x, y\in [N]$ such that $x\prec y$, but neither $x\prec \med(x, y)$ nor $\med(x, y)\prec y$. \hfill (Invalid median) 
        \end{description}
\end{description}

    While $\Avoid$ reduces to the Linear Ordering Principle \cite{KortenP24}, we show a converse in the dense setting.

     \begin{theorem}
        $\DLO \leq_{dt}\lossycode$.
    \end{theorem}

\section{Preliminaries} \label{sec: prelim}

\subsection{Basics of \texorpdfstring{$\TFNP$}{TFNP}}
$\TFNP$ contains all search problems which are (i) total: a solution is guaranteed to exist, and (ii) in $\NP$: there is an efficient procedure to check whether a candidate solution is valid. It is believed that $\TFNP$ does not admit complete problems \cite{Pudlak15} and much of the research in this area has focused on studying syntactic subclasses (those with complete problems) which capture many of the total search problems of interest. These classes are typically defined by simple existence principles that capture the totality of the problems within that class. These naturally give rise to total search problems. For example, $\PWPP$ is the class of all search problems whose totality is witnessed by the existence principle: any map from $2N$ to $N$ must have a collision~\cite{Jerabek16}. To make these problems non-trivial, the input is presented succinctly as a circuit $C$ that on input $i$ outputs the $i$-th bit of the search problem. For example, the existence principle for $\PWPP$ gives rise to the following (white-box) total search problem.

\begin{description}
    \item[$\textsc{Weak-Pigeon}$.] Given $P:\{0,1\}^n \rightarrow \{0,1\}^{n-1}$, a solution is $x \neq y$ such that $P(x)=P(y)$.
\end{description}
$\PWPP$ is then the class of all total search problems which are efficiently reducible to $\textsc{Weak-Pigeon}$. 
 
A major thrust of this line of work is to understand the relationships between these classes. However, a separation between classes would imply $\P \neq \NP$. As a proxy, and as natural objects in their own right, researchers have studied total search problems in the \emph{black-box} model. In this setting, the input $C$ is given as a black box which can be queried, but we no longer have access to the description of $C$. 

In this setting a \emph{(query) search problem} is a sequence of relations $R_n \subseteq \{0,1\}^n \times {\cal O}_n$, for each $n \in \mathbb{N}$. It is total if for every $x \in \{0,1\}^n$ there is $o \in {\cal O}_n$ such that $(x,o) \in R_n$. We think of the input $x \in \{0,1\}^n$ as being accessed by querying the individual bits, and we will measure the complexity of solving $R_n$ as the number of bits that must be queried to determine some suitable $o \in {\cal O}_n$. An efficient algorithm is one that makes at most $\poly(\log n)$-many queries\footnote{As the input is succinctly encoded, this corresponds to looking at a polynomial part of the entire input.}; these problems belong to the class $\FP^{dt}$, where $dt$ indicates that it is the black-box version of the class. Similarly, $R \in \TFNP^{dt}$ if for every $n \in \mathbb{N}$ and each $o \in {\cal O}_n$ there exists a $\poly(\log n)$-depth decision tree $T_o :\{0,1\}^n \rightarrow \{0,1\}$ such that $T_o(x)=1 $ iff $(x,o) \in R_n$. 
While search problems are formally defined as a sequence $R = (R_n)_{n \in \mathbb{N}}$, we will often want to speak about individual members of this sequence. For readability, we will abuse notation and refer to elements $R_n$ in the sequence as total search problems. Furthermore, we will often drop the subscript $n$ and rely on context to differentiate. 

We compare the complexity of total search problems by reductions between them; the following is the black-box (decision tree) analogue of a deterministic polynomial-time reduction between search problems. 

\begin{definition}
    For total search problems $R \subseteq \{0,1\}^n \times {\cal O}_n$ and $S \subseteq \{0,1\}^m \times {\cal O}'_m$, there is an $S$-formulation of $R$ if for every $i \in [m]$ and $o \in {\cal O}'_m$ there are functions $f_i :\{0,1\}^n \rightarrow \{0,1\}$ and $g_o:\{0,1\}^n \rightarrow {\cal O}_n$ such that 
    \begin{align} \label{eq:red} (f(x),o) \in S) \implies (x,g_o(x)) \in R, \end{align}
    where $f(x):= (f_1(x),\ldots, f_m(x))$. The \emph{depth} of the $S$-formulation is 
    \[ d~:=~ \max\left( \{ \mathsf{depth}(f_i): i \in [m]\} \cup \{\mathsf{depth}(g_o):o \in {\cal O}'_m \}\right),\]
    where $\mathsf{depth}(f)$ denotes the minimum depth of any decision tree which computes $f$. The \emph{size} of the $S$-formulation is $m$, the number of input bits to $S$. The \emph{complexity} of an $S$-formulation is $\log m+ d$ and the complexity of reducing $R$ to $S$ is the minimum complexity of any $S$-formulation of $R$.

    This definition extends to sequences naturally. If $S=(S_n)$ is a sequence and $R_n$ is a single search problem, then the complexity of reducing $R_n$ to $S$ is the minimum over $m$ of the complexity of reducing $R_n$ to $S_m$. For two sequences $S=(S_n)$ and $R=(R_n)$, the complexity of reducing $R$ to $S$ is the complexity of reducing $R_n$ to $S$ for each $n$. 
    We say that a reduction from $R$ to $S$ is \emph{efficient} if its complexity is $\poly(\log (n))$ and denote this by $R \leq_{dt} S$.
\end{definition}

\subsection{\texorpdfstring{$\TFZPP$}{TFZPP}}

In this work, we will be particularly interested in the total search problems which are solvable in \emph{randomized} polynomial time. Formally,  $R \subseteq \{0,1\}^*\times \{0,1\}^* \in \TFZPP$ if there is a distribution ${\cal D}$ over polynomial-time Turing Machines $A$ with range $\{0,1,\bot\}$, such that $\Pr_{A \sim {\cal D}}[A(x)=\bot] \leq 1/3$ and

$\TFZPP$ is defined semantically and it is unlikely to have complete problems. However, we show that it is exactly the $\TFNP$ problems in the $\mathsf{TF}\Sigma_2^P$ class $\mathsf{APEPP}$, where $\mathsf{APEPP}$ is the class of total search problems that are reducible to $\Avoid$, as defined in~\cite{DBLP:conf/innovations/KleinbergKMP21}.

\avoidIntersect*

\begin{proof}
    Let $L\in\TFNP \cap \mathsf{APEPP}$. This means that given an instance $x$ of $L$, there are deterministic polynomial-time algorithms $V$, $C$, and $R$, such that:
    \begin{itemize}
        \item {\bf $V$ is a $\TFNP$ verifier for $L$.} For every string $z$ of length polynomial in $|x|$, $V(x, z) = 1$ if and only if $z$ is a valid solution for $x$.
        \item {\bf $(C, R)$ is a reduction from $x$ to $\Avoid$.} The output of $C(x)$ is a circuit $C_x$ mapping $\ell$ input bits to $\ell+1$ output bits, where $\ell \le \poly(|x|)$; given any $y \in \{0, 1\}^{\ell + 1}\setminus \mathrm{Range}(C_x)$, $R(x, y)$ outputs a valid solution for $x$.
    \end{itemize}
    Then we can solve $L$ in $\TFZPP$ via the following procedure. Guess $y \in \{0, 1\}^{\ell+1}$ uniformly at random and compute $z := R(x, y)$. If $V(x, z)$ accepts, then we output $z$; otherwise, we output $\bot$. By the correctness of $V$, if we did not output $\bot$, then our output is a valid solution of $x$. On the other hand, since at least a $1/2$ fraction of strings $y \in \{0, 1\}^{\ell+1}$ are valid outputs of $\Avoid$ on the instance $C_x$, by the correctness of $(R, C)$, we will output a valid solution w.p.~at least $1/2$.

    Now we prove the converse direction. If $L \in \TFZPP$ then clearly $L\in\TFNP$; hence we only need to show that there is a mapping reduction from $L$ to $\Avoid$. Let $U(x, r)$ be the zero-error randomized algorithm for $L$, i.e., $U(x, r)$ outputs a valid solution for $x$ w.p.~at least $1/2$ over its randomness $r$, and it outputs $\bot$ whenever it fails to output a valid solution.

    \def\PRG{\mathsf{PRG}}
    \def\TT{\mathsf{TT}}
    By standard results in derandomization~\cite{NisanW94, ImpagliazzoW97, Umans03}, there exist absolute constants $c, d\ge 1$ and a deterministic polynomial-time algorithm $\PRG$ such that the following holds. For every truth table $f$ of length $s^{10c}$, if the circuit complexity of $f$ is at least $s^c$, then $\PRG(f)$ outputs a list of $s^d$ strings that $(1/s^2)$-fools every size-$s^2$ circuit. That is, for every circuit $C:\{0, 1\}^s \to \{0, 1\}$ of size at most $s^2$,
    \[\mleft|\Pr_{x\sim \{0, 1\}^s}[C(x) = 1] - \Pr_{x\sim \PRG(f)}[C(x) = 1]\mright| \le 1/s^2.\]
    Now, let $s \le \poly(|x|)$ be the circuit complexity of $U$. Consider the \emph{truth table generator} $\TT: \{0, 1\}^{O(s^c\log s)} \to \{0, 1\}^{s^{10c}}$ that takes the description of a size-$s^c$ circuit $C: \{0, 1\}^{10c\log s} \to \{0, 1\}$ as input and outputs the length-$s^{10c}$ truth table of $C$. We treat $\TT$ as an instance for $\Avoid$ and reduce $x$ to $\TT$.\footnote{An unusual aspect of this reduction is that $\TT$ does \emph{not} depend on $x$!}

    It remains to show how to solve the instance $x$ deterministically given a non-output of $\TT$. Note that if $f \in \{0, 1\}^{s^{10c}}$ is a non-output of $\TT$, then the circuit complexity of $f$ is at least $s^c$, hence $\PRG(f)$ outputs a list of $\poly(s)$ strings that $(1/s^2)$-fools every size-$s^2$ circuits. This implies that
    \[\Pr_{r\sim \PRG(f)}[U(x, r) \ne \bot] \ge \Pr_{r\sim \{0, 1\}^s}[U(x, r) \ne \bot] - 1/s^2 \ge 1/2 - 1/s^2 > 0,\]
    and in particular, there exists at least one string $r \in \PRG(f)$ such that $U(x, r)\ne\bot$. We can solve the instance $x$ by cycling through every $r \in \PRG(f)$ and outputting $U(x, r)$ whenever we encounter such a good $r$.
\end{proof}

Note that this equivalence holds even in the black-box model, since its proof is relativizing.

\begin{definition}[$\TFZPP$]
    A total $\NP$ search problem $R \subseteq \{0,1\}^* \times \{0,1\}^*$ is in $\TFZPP$ if there is a distribution $\cal D$ over polynomial-time algorithms $A$ with output in $\{0,1,\bot\}$ such that:
    \begin{enumerate}
        \item For every $x \in \{0,1\}^*$ and every $A \sim {\cal D}$, if $A(x) \neq \bot$ then $(x,A(x)) \in R$,
        \item For every $x \in \{0,1\}^*$, 
        \[ \Pr_{A \sim {\cal D}}[A(x) = \perp] \leq 1/3. \]
    \end{enumerate}

    Similarly, $R \in \TFZPP^{dt}$ if there is a family of distributions ${\cal D} = \{{\cal D}_n\}_{n \in \mathbb{N}}$  over $\poly \log (n)$-depth decision trees with leaves labeled  in $\{0,1, \bot\}$, where on input $x$ we sample a decision tree $A \sim {\cal D}_{|x|}$, and ${\cal D}$ satisfies (1) and (2).
\end{definition}

\subsection{\texorpdfstring{$\lossycode$}{Lossy-Code}}

As mentioned in the introduction, $\lossycode$ is the Herbrandization of $\Avoid$. Let $N < M$ be two parameters (think of $N \ll M$), (the black-box version of) $\lossycode$ is the following problem:
\begin{description}
\item[$\lossycode_{N \to M}$.]  Given query access to a pair of functions $f: [N] \to [M]$ and $g: [M] \to [N]$, find $x \in [M]$ such that $f(g(x)) \neq x$.
\end{description}

We need the following basic fact about $\lossycode$ that roughly states that the ``stretch function'' of $\lossycode$ does not influence its complexity as long as it is in the ``weak'' regime. This fact and similar statements for other variants of the weak pigeonhole principle have been very useful in bounded arithmetic~\cite{ParisWW88, Thapen-PhD, Krajicek04a, Jerabek04, Jerabek07, CLO24}, total search problems~\cite{Kor21, DBLP:conf/coco/Korten22, LiLR24}, and cryptography~\cite{DBLP:journals/jacm/GoldreichGM86, DBLP:conf/crypto/Merkle87}.

\def\eps{\varepsilon}
\begin{lemma}\label{lemma: robustness of lossy-code}
    Let $\eps > 0$ and $M > (1+\eps)N$. There is a decision tree reduction of complexity $O(\eps^{-1}\log (M/N))$ from $\lossycode_{N \to (1+\eps)N}$ to $\lossycode_{N \to M}$.
\end{lemma}

By \autoref{lemma: robustness of lossy-code}, $\lossycode_{N \to 1.01N}$, $\lossycode_{N \to 2N}$, and $\lossycode_{N \to N^{100}}$ are equivalent up to decision tree reductions of $\polylog(N)$ depth. In this paper, unless otherwise stated, $\lossycode$ always stands for $\lossycode_{N \to 2N}$.

We denote $\Lossy$ as the class of total search problems reducible to $\lossycode$.\footnote{Previous literature~\cite{DBLP:conf/focs/LiPT24, CookLMP25} defined $\Lossy$ as the class of \emph{decision problems} reducible to $\lossycode$. In our context, it is more natural to define $\Lossy$ as a class of total search problems.} It follows from \autoref{lemma: robustness of lossy-code} that $\Lossy$ is robust in the sense that it does not matter whether it is defined using $\lossycode_{N \to 1.01N}$ or $\lossycode_{N \to N^{100}}$ as the complete problem. In fact, $\Lossy$ is \emph{extremely} robust: it is closed under Turing reductions ($\FP^\Lossy = \Lossy$~\cite{buss2012propositional, DBLP:conf/focs/LiPT24}) and it is self-low ($\Lossy^\Lossy = \Lossy$~\cite{GhentiyalaLi25}).

We will also consider the \emph{bijective} version of $\lossycode$, called $\lossycodep$. Let $N < M$, define:
\begin{description}
\item[$\lossycodep_{N \to M}$.]  Given query access to a pair of functions $f: [N] \to [M]$ and $g: [M] \to [N]$, find either $x \in [M]$ such that $f(g(x)) \neq x$, or $y\in [N]$ such that $g(f(y)) \neq y$.
\end{description}

Clearly, $\lossycodep_{N\to M}$ reduces to $\lossycode_{N\to M}$. We also use $\lossycodep$ to denote $\lossycodep_{N \to 2N}$ by default.
It is not difficult to see that the ``strong'' versions of these problems, $\lossycode_{N\to N+1}$ and $\lossycodep_{N\to N+1}$, are complete for $\PPADS$ and $\PPAD$ respectively. %

\section{Randomized Proof Complexity and Explicit Separations} \label{sec: tfzpp general results}

We begin by describing the connection between black-box $\TFZPP$ and proof complexity, and how this can be leveraged to obtain explicit separations from other natural classes. Proof complexity is concerned with the efficient provability of propositional theorems (unsatisfiable CNF formulas) in various \emph{proof systems}---simply a verifier for the language $\mathrm{UNSAT}$ of unsatisfiable CNF formulas.

\begin{definition}
\label{def:proofSystem}
    A propositional proof system is a polynomial-time machine $\cal P$ such that for every CNF formula $F$, $F \in \mathrm{UNSAT}$ iff there exists a proof $\Pi \in \{0,1\}^*$ such that ${\cal P}(F,\Pi)=1$. We say that $\Pi$ is a $\cal P$-proof of $F$ and define the size of $\Pi$ to be $s(\Pi):=|\Pi|$. 

    When studying connections between proof systems and $\TFNP$ classes, it is standard to also consider an associated notion of the \emph{width} of a proof $w(\Pi)$. This is typically specific to the proof system---for example, in resolution (defined next), it is the maximum number of literals in a clause in $\Pi$, while for algebraic systems such as Sum-of-Squares, the width is the degree of the polynomials occurring in the proof. 

    With a definition of width, the \emph{complexity} of proving $F$ in $\cal P$ is
    \[ \mathcal{P}(F)\coloneqq\min_{{\cal P}\text{-proof } \Pi \text{ of } F} w(\Pi)+\log s(\Pi).\]
\end{definition}

A standard example is the \emph{resolution} proof system. A resolution proof of an unsatisfiable CNF formula $F$ consists of a sequence of clauses $C_1,\ldots, C_t = \emptyset$ ending with the empty clause which contains no literals, such that each clause $C_i$ either belongs to $F$ or is derived from earlier clauses in the sequence according to the \emph{resolution rule}.
\begin{center}
    \textbf{Resolution rule:} From two clauses with complementary literals $A \lor x$ and $B \lor \overline x$, derive $A \lor B$.
\end{center}

The size of a resolution proof is the number of clauses that it contains, while the width is the maximum number of literals within any clause in the proof.  A resolution proof $C_1,\ldots, C_t$ is \emph{tree-like} if each $C_i$ is used at most once as a premise for the resolution proof. They are named as such because the implication graph of such proofs is a tree. 
 
There is a long line of work connecting proof complexity and black-box $\TFNP$ \cite{beame1995relative, GoosKRS19,GHJ+22, BussFI23, FlemingMD25,HubacekKT24,Thapen24, DavisR23,PR24, FlemingGPR24}. These connections show that a total search problem is contained within a  class iff an associated proof system can prove the totality of that search problem. We can phrase the totality of any total search problem $R \subseteq \{0,1\}^n \times {\cal O}$ as an unsatisfiable CNF formula in the following way: for each $o \in {\cal O}$ let $V_o$ be a decision tree which checks whether $o$ is a solution; that is, $V_o(x)=1$ iff $(x,o) \in R$. A root-to-leaf path in $V_o$ is a \emph{$1$-path} if its leaf is labeled $1$. We will associate with any path $p$ the conjunction of literals that it follows. Then the totality of $R$ is expressed as 
\[ F_R~ \coloneqq~\neg \mleft( \bigvee_{o \in {\cal O}} \bigvee_{1\text{-path }p \in V_o} p \mright). \]
If $R \in \TFNP^{dt}$ then $V_o$ can be assumed to have depth $\poly \log (n)$, and hence the width of $F_R$ is also $\poly \log (n)$.

Similarly, we can associate with any unsatisfiable CNF formula $F = C_1 \wedge \ldots \wedge C_m$ a total search problem $\Search_F \subseteq \{0,1\}^n \times [m]$ such that $(x,o) \in \Search_F$ iff $C_o(x)=0$. Observe that whenever $F$ has $\poly(\log(n))$ width then $\Search_F \in \TFNP^{dt}$ and furthermore that $\Search_{F_R}$ is reducible to $R$ by decision trees of depth at most the width of $F_R$. 

For a syntactic class ${\cal C} \subseteq \TFNP^{dt}$ we will denote by ${\cal C}(R)$ the complexity of reducing $R$ to $S$, where $S$ is any complete problem for ${\cal C}$. We say that a proof system ${\cal P}$ is \emph{characterized} by a class ${\cal C} \subseteq \TFNP^{dt}$ if $R \in {\cal C}$ iff ${\cal P}(F) = \poly({\cal C}(R))$. A standard example is that $\FP^{dt}$ characterizes tree-like resolution. Said differently, decision trees are equivalent to tree-like resolution proofs.

We extend these characterizations to capture randomized reductions. We show that randomized reductions between total search problems give rise to proofs in randomized proof systems, a notion introduced by Buss, Kołodziejczyk, and Thapen~\cite{BussKT14}.

\begin{definition}
    Let $\cal P$ be any propositional proof system. A \emph{randomized} $\cal P$-proof, denoted $r \cal P$,  of an unsatisfiable formula $F$ is a distribution ${\cal D}$ supported on pairs $(\Pi,B)$, such that 
    \begin{enumerate}
        \item Each $B$ is a CNF formula over the variables of $F$,
        \item $\Pi$ is a $\cal P$ proof of $F \wedge B$,
        \item For any assignment $x \in\{0,1\}^n$, $\Pr_{(\Pi,B)\sim {\cal D}}[B(x)=1] \geq 2/3$.
    \end{enumerate}
    The \emph{size} $s({\cal D})$, and \emph{width} 
     $w({\cal D})$ of an $r\cal P$-proof $\cal D$ is the maximum width and size of a proof $\Pi$  in the support of $\cal D$. The \emph{complexity} of proving $F$ in $r \cal P$ is 
    \[r{\cal P}(F) := \min_{\textnormal{$r{\cal P}$-proof $\cal D$ of $F$}} w({\cal D})+\log s({\cal D}) \]
\end{definition}

Note that a randomized proof system is not a Cook-Reckhow proof system in the sense of \autoref{def:proofSystem} since its proofs typically cannot be polynomial-time verified~\cite{PudlakT19}.

The main theorem of this section, \autoref{thm:mainChar}, shows that  a proof system $\cal P$ is characterized by class ${\cal C}$ iff the totality of the total search problems $R$ which are randomly reducible to any complete problem for ${\cal C}$ is provable in $r{\cal P}$. The following definition is equivalent to the probabilistic reduction in \cite{Jerabek16}.

\begin{definition}
	A randomized (ZPP) reduction from a search problem $S \subseteq [t] \times {\cal O}$ to $R \subseteq [n] \times {\cal Q}$ is a distribution ${\cal D}$ over deterministic reductions ${\cal T} = (T,\{T_o\})$ such that each output decision tree is labeled  either by some $j \in {\cal O}$ or by $\bot$, and ${\cal D}$ satisfies
	\begin{enumerate}
		\item For every $x \in [t]$ and every ${\cal T} \sim {\cal D}$, if $(T(x),o) \in R$ then either $T_o(x) = \bot$ or $(x,T_o(x)) \in S$.
		\item For every $x \in [t]$,
			\[\Pr_{{\cal T} \sim {\cal D}}[\exists o \in {\cal O} : (T(x),o) \in R ~\wedge~ T_o(x)=\bot] \leq \varepsilon\]
	\end{enumerate}
\end{definition}

\begin{theorem}
\label{thm:mainChar}
    If a proof system ${\cal P}$ is characterized by the total search problems reducible to $R \in \TFNP^{dt}$, then $r{\cal P}$ is characterized by the total search problems that are randomized-reducible to $R$.
\end{theorem}

The intuition for this theorem is most clear in the case of randomized reductions to $\FP^{dt}$ (which is $\TFZPP^{dt}$) and random tree-like resolution. This is also the case that we will use to derive consequences about $\TFZPP^{dt}$. We leave the proof of \autoref{thm:mainChar} to the \autoref{sec:appendixPC}. 

We remark that \autoref{thm:mainChar} reduces the task of showing that a $\TFNP$ class is closed under randomized reduction to showing that the corresponding proof system is \emph{closed} in the sense that ${\cal P} = r{\cal P}$. We are not aware of any such proof system, and it would be interesting to exhibit one.

\begin{lemma}\label{lem: rtR equal TFZPP}
    There is a quasi-polynomial size random tree-like resolution proof of $F = C_1 \wedge \ldots \wedge C_m$ iff $\mathsf{Search}_F \in \TFZPP^{dt}$.
\end{lemma}

\begin{proof}
    Suppose that $\Search_F \in \TFZPP^{dt}$ and let $\cal D$ be a distribution over depth-$d$ decision trees solving $\Search_F$ as in the definition of $\TFZPP^{dt}$. We construct a $\varepsilon$-error randomized tree resolution proof $\cal P$, which is defined by the following sampling procedure:
\begin{enumerate}
    \item Sample $T \sim \cal D$.
    \item Let $B$ be the set of clauses obtained by taking the negation of each root-to-leaf path in $T$ ending in $\bot$,
    \[B := \{\neg p: p \in T \mbox{ is a root-to-$\bot$ path}\},\]
    where we think of a path as the conjunctions of the literals that appear along it (a queried variable is a positive literal if $p$ took the $1$-branch, and a negative literal if $p$ took the $0$-branch).
    \item To construct $\Pi$, we will use the equivalence between a decision tree solving the false-clause search problem and tree-resolution proofs. Let $T^*$ be obtained from $T$ by relabeling each path $p$ ending in $\bot$ by the clause $\neg p \in B$. Let $F \cup B$ be the CNF formula whose clauses are the clauses of $F$ and those in $B$, and observe that $T^*$ is a depth-$d$ decision tree solving $\Search_{F \cup B}$. Thus, there is a depth-$d$ tree resolution proof $\Pi$ of $F \cup B$.
\end{enumerate}
As $B$ states that we do not follow any root-to-$\bot$ path in $T$, for any $x \in \{0,1\}^n$, $T(x)\neq \bot$ iff $x$ satisfies all of $B$. Therefore, $\Pr_{T \sim {\cal D}}[T(x)=\bot]\leq \varepsilon$ implies that $\Pr_{(B, \Pi) \sim {\cal P}}[B(x)=1] \geq 1-\varepsilon$ for every $x \in \{0,1\}^n$.

In the other direction, suppose that $\cal P$ is a $\varepsilon$-random tree resolution proof of $F$ with complexity $c$. We construct a distribution $\cal D$ over decision trees solving $\Search_F$ by the following sampling procedure:
\begin{enumerate}
    \item Sample $(\Pi,B) \sim {\cal P}$.
    \item Let $T^*$ be the depth-$d$ decision tree solving $\Search_{F \cup B}$ obtained from $\Pi$ obtained by the equivalence between tree resolution proofs and decision trees in the same manner as in point (3) above. It is well-known that the depth of a resolution proof is bounded by its width, and hence $T^*$ has depth at most $c$.
    \item Let $T$ be the decision tree obtained from $T^*$ by relabeling each leaf of $\Pi$ that is labeled by a clause in $B$ by $\bot$. 
\end{enumerate}
As for any $x \in \{0,1\}^n$, $\Pr_{(B,\Pi) \sim \cal D}[B(x)=1] \geq 1-\varepsilon$, we have that $\Pr_{T \sim {\cal D}}[T(x)=\bot] \leq \varepsilon$.
\end{proof}

\subsection{Separations}

We now use \autoref{lem: rtR equal TFZPP} to show that $\TFZPP^{dt}$ is not contained within any uniformly generated $\TFNP^{dt}$ class unless $\NP$ is contained within quasi-polynomial time ($\mathsf{QP}$). 

\begin{theorem}
\label{thm:TFZPPSepALL}
    $\TFZPP^{dt} \not \subseteq {\cal C}$ for any uniformly generated class ${\cal C} \subseteq \TFNP^{dt}$ unless $\NP \subseteq \mathsf{QP}$. 
\end{theorem}

This theorem follows by combining the characterization of uniform $\TFNP^{dt}$ classes by proof systems of Buss et al.~\cite{BussFI23}, \autoref{lem: rtR equal TFZPP}, and the following result of Pudl\'ak and Thapen \cite{PudlakT19}.

This separation relies on a theorem of Pudl\'ak and Thapen \cite{PudlakT19} who showed that random resolution cannot be simulated by any propositional proof system unless $\P \neq \NP$. A straightforward examination of their theorem reveals that it also holds for tree-like resolution and can be stated in the following form.

\begin{theorem}[Proposition 10 in \cite{PudlakT19}]
\label{thm:noPSim}
    There is a family of unsatisfiable $3$-CNFs ${\cal F}$ such that:
    \begin{enumerate}
        \item There are $O(\log (n))$-complexity random tree-like resolution proofs of ${\cal F}$.
        \item If there is a propositional proof system which has $\poly \log (n)$-complexity proofs of ${\cal F}$ then $\NP \subseteq {\mathsf{QP}}$.
    \end{enumerate}
\end{theorem}
Indeed, to prove their theorem Pudl\'ak and Thapen observe that random (treelike) resolution has small proofs of any highly unsatisfiable formula (one for which any assignment falsifies many clauses), and that the PCP theorem can be used to put any $3$-CNF formula into this form (the family ${\cal F}$). Using this, they show that if there existed a propositional proof system which could (quasi-polynomially) simulate random tree-like resolution, then one could use its variability to decide SAT. 

Combining this theorem with \autoref{lem: rtR equal TFZPP} and the characterization of $\TFNP^{dt}$ classes by propositional proof systems due to Buss et al.~\cite{BussFI23} proves \autoref{thm:TFZPPSepALL}.  Say that $ R=\{R_n\} \in \TFNP^{dt}$ is \emph{uniformly generated} if there is a Turing Machine which on input $1^n$ outputs $R_n$, and say that a class is uniformly generated if it has a uniformly generated complete problem. Note that all major $\TFNP^{dt}$ subclasses are uniformly generated.

\begin{proof}[Proof of \autoref{thm:TFZPPSepALL}]
    Let $\cal C$ be a uniformly generated class such that $\TFZPP^{dt} \subseteq {\cal C}$.
    Buss et al.~\cite{BussFI23} showed that every uniformly generated $\TFNP^{dt}$ subclass is characterized by a proof system; let ${\cal P}$ be the system for $\cal C$. Consider the family of formulas ${\cal F}$ from \autoref{thm:noPSim}. As $\TFZPP^{dt} \subseteq {\cal C}$, there are $\poly \log (n)$-complexity ${\cal P}$-proofs of ${\cal F}$. Hence, \autoref{thm:noPSim} implies that $\NP \subseteq \mathsf{QP}$. 
\end{proof}

\subsection{Explicit Separations}
The separating examples in \autoref{thm:TFZPPSepALL} rely on an unproven hypothesis. We end this section by proving explicit separations between $\TFZPP^{dt}$ and every major $\TFNP^{dt}$ subclass which do not rely on any unproven assumptions. A separation of $\PPA^{dt}$ from $\TFZPP^{dt}$ was implicitly shown by Beame et.~al. \cite{beame1995relative}, who proved Nullstellensatz lower bounds for $\lossycode$.\footnote{More specifically, Beame et.~al.~\cite[Theorem 12]{beame1995relative} proved the Nullstellensatz degree lower bounds for $\textsc{Weak-Pigeon}$. They also showed that any Nullstellensatz degree lower bounds for $\textsc{Weak-Pigeon}$ implies the same Nullstellensatz degree lower bounds for $\lossycode$ in \cite[Lemma 10]{beame1995relative} (see also Definition 3.1, 3.2 in \cite{beame1995relative} for the definition of $\lossycode$ and $\textsc{Weak-Pigeon}$).}
The remaining major $\TFNP^{dt}$ classes are contained within $\PLS^{dt}$ and $\PPP^{dt}$. We show the following.

\begin{theorem}
\label{thm:explicitSeps}
    There exist explicit total search problems in $\TFZPP^{dt}$ which are not contained in $\PLS^{dt}$ nor $\PPP^{dt}$.
\end{theorem}

It is known that if a total search problem $\Search_F$ is in $\PLS^{dt}$ or $\PPP^{dt}$, then $F$ has a small low-degree \emph{Sum-of-Squares} (SoS) proof; see \cite{FlemingKP19} for an exposition on this proof system. Our hard instance is based on the recent work of Hopkins and Lin~\cite{HopkinsL22}, who exhibited the first explicit hard $3$-XOR instance for SoS.
\begin{theorem}[\cite{HopkinsL22}]\label{thm: SoS_lower_bound}
There exist constants $\mu_1,\mu_2\in (0,1)$ and a polynomial time algorithm which, given $1^n$ as input, outputs a $3$-XOR formula $F=C_1\land \cdots\land C_m$ on $n$ variables such that:
\begin{itemize}[noitemsep]
    \item For every $x\in \{0,1\}^n$, $\Pr_{i\sim [m]}[C_i(x)=1]\le 1-\mu_1$.
    \item Any Sum-of-Squares refutation of $F$ requires degree at least $\mu_2n$.
\end{itemize}
\end{theorem}

\begin{proof}[Proof of \autoref{thm:explicitSeps}]
    Let $F$ be as in~\Cref{thm: SoS_lower_bound}. We show that $R\coloneqq \Search_F$ satisfies the desired properties.
    We first prove $R\in \TFZPP^{dt}$.
    Consider the following simple algorithm: sample $i\sim [m]$ uniformly at random, and make three queries to check if $C_i(x)=0$.
    If so, output $i$; otherwise, output $\bot$.
    By the first item of~\Cref{thm: SoS_lower_bound}, the algorithm succeeds with probability at least $\mu_1$.
    Repeating the procedure $O(1/\mu_1)$ times boosts the success probability to at least $2/3$.

    To separate $R$ from $\PPP^{dt}$ and $\PLS^{dt}$, we only need to show there is no efficient black-box reduction from $R$ to a complete problem for one of these classes. If there was, then there would be an efficient SoS proof of $F$, contradicting the second item of~\Cref{thm: SoS_lower_bound}.
    \end{proof}

\section{\texorpdfstring{$\ILP$}{Nephew}}\label{sec: Nephew}

Recall the $\ILP$ problem, which is our main candidate for a total search problem in $\TFZPP$ but not reducible to $\lossycode$:

\begin{description}
    \Nephewitem 
\end{description}

The main result of this section is the inclusion theorem:

\inTFZPPandPWPP*

The intuition behind \cref{thm: tfzpp-and-pwpp-containment} is as follows:
\begin{itemize}
    \item We can treat certain vertices in a $\ILP$ instance like the root of a directed binary tree, where the leaves of the tree correspond to solutions of the $\ILP$ instance.
    \item Finding a leaf of a rooted binary tree is easy for both $\PWPP$ and $\TFZPP$ computations.
\end{itemize}
Once we have proven the above, we are nearly done. First, choose an arbitrary vertex. Perhaps that vertex is one of the many that we can treat as the root of a binary tree, and so from it we can find a solution. Otherwise, we use a procedure to find such a root vertex. In fact, we will be able to find two vertices, one of which must be a good root vertex. To show inclusion in $\TFZPP$, we just need to pick one of these two randomly. For $\PWPP$, we will have to consider both. We make an adjustment to the argument for a single root vertex so that it works for two potential root vertices.

To be more concrete, we will reduce $\ILP$ to the following \emph{promise} search problem.

\begin{description}
    \item[$\BTreeLeaf.$] An instance consists of a set $V$ of vertices, a special vertex $v^*$, and two functions $L: V \to V \cup\{\bot\}$ and $R: V \to V \cup\{\bot\}$. We define a subset $V^* \subseteq V$ recursively: $v^* \in V^*$ and, for every $v \in V^*$, we add $L(v)$ to $V^*$ if $L(v) \neq \bot$ and $R(v)$ to $V^*$ if $R(v) \neq \bot$. We promise that the induced subgraph on $V^*$ is a (directed) tree rooted at $v^*$ and that for all $v \in V^*$ either:
    \begin{itemize}
        \item $L(v) = R(v) = \bot$, or
        \item $L(v) \neq \bot$ and $R(v) \neq \bot$ and $L(v) \neq R(v)$.
    \end{itemize}

    A solution is a $(\lceil \log |V| \rceil + 1)$-length path, represented by a string $p \in \{L, R\}^{\lceil \log |V| \rceil+1}$, where starting at $v^*$ and descending by the functions specified by the characters of $p$ in order will at some point reach a vertex $v$ where $L(v) = R(v) = \bot$.
\end{description}

Note that instead of simply asking for a leaf, we require a root-to-leaf path for a solution. This is to confirm that the leaf is in the binary tree rooted at $v^*$.

\paragraph{Finding a root-to-leaf path in a rooted binary tree.} It is easy for a $\PWPP$ or $\TFZPP$ computation to find a solution to $\BTreeLeaf$. This follows from the simple observation that there are (many) more paths of length $(\lceil\log|V|\rceil + 1)$ than vertices in the tree, so most paths must be solutions.

\begin{lemma} \label{lem: solve using tfzpp}
    A solution to $\BTreeLeaf$ can be found with high probability using a randomized algorithm.
\end{lemma}

\begin{proof}
    The algorithm guesses a random path of length $\lceil \log |V|\rceil + 1$, which will be a solution with probability at least $5/6$. This is because if a path $p$ is not a solution, then following $p$ reaches a vertex $v_p$ with two children. The path $p$ is the only non-solution path of length $\lceil \log |V| \rceil + 1$ to reach $v_p$; furthermore, because $V^*$ is an induced tree, no other path ends at a vertex with $L(v_p)$ or $R(v_p)$ as children. This means that the set $\{v_p, L(v_p), R(v_p)\}$ are uniquely reached by $p$ out of all of the non-solution paths. Therefore, there are at least 3 times as many vertices in $V$ as there are non-solution paths. Let the fraction of non-solution paths be $\alpha$. Then

    \[ |V| \geq 3 \alpha \left|\{L, R\}^{\lceil \log |V| \rceil+1}\right| \geq 3 \alpha \cdot 2 |V|, \]
 and so $\alpha \leq 1/6$.
\end{proof}

To show inclusion in $\PWPP$, we reduce to the $\PWPP$-complete problem $\textsc{Weak-Pigeon}$. (Note that this is not a reduction between $\TFNP$ problems, as $\BTreeLeaf$ is a promise problem.) We recall its definition here:

\begin{description}
   \item[$\textsc{Weak-Pigeon}$.] Given $h:\{0,1\}^n \rightarrow \{0,1\}^{n-1}$ a solution is $x \neq y$ such that $h(x)=h(y)$.
\end{description}

\begin{lemma} \label{lem: solve using pwpp}
    $\BTreeLeaf$ reduces to $\textsc{Weak-Pigeon}$.
\end{lemma}

\begin{proof} 
    Let $n = \lceil \log |V| \rceil+1$. Let $v_p$ be the vertex reached by $p$ if $p$ is a non-solution path. Then we define $h$ as a map from $(\lceil \log |V| \rceil+1)$-length paths to vertices:

    \[ h(p) = \begin{cases} v_p & \mbox{if $p$ is not a solution}; \\ v^* & \mbox{otherwise}. \end{cases} \] 

    (Recall that since $v^*$ is the root, we have that $v_p\ne v^*$ for any path $p$.) Thus, paths $(x, y)$ are a collision if and only if $x$ and $y$ are both solutions to $\BTreeLeaf$, and so from any collision we can find a solution to $\BTreeLeaf$ by arbitrarily choosing one from the pair.
\end{proof}

\paragraph{The structure of $\ILP$ instances and finding a rooted binary tree.} For any $\ILP$ instance $(V, f, g)$, let $G_f$ be the directed graph with vertex set $V$ and where $(u, v)$ is an edge if and only if $v = f(u)$. Then $G_f$ is a directed graph with out-degree one, and therefore the connected components of $G_f$ have a simple structure: they are composed of a cycle (perhaps a self-loop) and trees (with edges oriented leaf-to-root) that are rooted at vertices in the cycle. For any vertex $v \in V$, define its \emph{level} (denoted $\ell(v)$) as the distance from $v$ to any vertex on the cycle of its connected component. So, any vertex on the cycle has level 0, any non-cycle vertex pointing to a cycle vertex has level 1, and so on. See \cref{fig: graph structure} for an illustration.

\begin{figure}
    \centering
    \begin{tikzpicture}[level1/.style={fill=Cerulean!10!white, text opacity=1}, level2/.style={fill=Cerulean!30!white, text opacity=1}, level3/.style={fill=Cerulean!40!white}, level0/.style={fill=Cerulean!2!white}]
        \node[v,level3] (n00) at (0, 0) {3};
        \node[v,level3] (n01) at (2, 0) {3};
        \node[v,level2] (n02) at (4, 0) {2};
        \node[v,level1] (n03) at (6, 0) {1};
        \node[v,level1] (n04) at (8, 0) {1};
        \node[v,level1] (n05) at (10, 0) {1};
        \node[v,level1] (n06) at (12, 0) {1};
        \node[v,level3] (n10) at (0, 2) {3};
        \node[v,level2] (n11) at (2, 2) {2};
        \node[v,level1] (n12) at (4, 2) {1};
        \node[v, level0] (n13) at (6, 2) {0};
        \node[v,level0] (n14) at (8, 2) {0};
        \node[v,level0] (n15) at (10, 2) {0};
        \node[v,level0] (n16) at (12, 2) {0};
        \node[v,level3] (n20) at (0, 4) {3};
        \node[v,level3] (n21) at (2, 4) {3};
        \node[v,level2] (n22) at (4, 4) {2};
        \node[v,level1] (n23) at (6, 4) {1};
        \node[v,level0] (n24) at (8, 4) {0};
        \node[v,level0] (n25) at (10, 4) {0};
        \node[v,level3] (n26) at (12, 4) {3};
        \node[v,level0] (n30) at (0, 6) {0};
        \node[v,level0] (n31) at (2, 6) {0};
        \node[v,level1] (n32) at (4, 6) {1};
        \node[v,level0] (n33) at (6, 6) {0};
        \node[v,level0] (n34) at (8, 6) {0};
        \node[v,level1] (n35) at (10, 6) {1};
        \node[v,level2] (n36) at (12, 6) {2};

        \draw[farrow] (n00) to (n11);
        \draw[farrow] (n01) to (n02);
        \draw[farrow] (n02) to (n12);
        \draw[farrow] (n03) to (n13);
        \draw[farrow] (n04) to (n15);
        \draw[farrow] (n05) to (n15);
        \draw[farrow] (n06) to (n16);
        \draw[farrow] (n10) to (n11);
        \draw[farrow] (n11) to (n12);
        \draw[farrow] (n12) to (n13);
        \draw[farrow] (n13) to (n24);
        \draw[farrow] (n14) to (n13);
        \draw[farrow] (n15) to (n14);
        \draw[farrow] (n16) to (n15);
        \draw[farrow] (n20) to (n11);
        \draw[farrow] (n21) to (n11);
        \draw[farrow] (n22) to (n23);
        \draw[farrow] (n23) to (n24);
        \draw[farrow] (n24) to (n25);
        \draw[farrow] (n25) to (n16);
        \draw[farrow] (n26) to (n36);
        \draw[farrow] (n30) to[bend right] (n31);
        \draw[farrow] (n31) to[bend right] (n30);
        \draw[farrow] (n32) to (n31);
        \draw[farrow] (n33) to[loop left] (n33);
        \draw[farrow] (n34) to[loop left] (n34);
        \draw[farrow] (n35) to (n34);
        \draw[farrow] (n36) to (n35);
        
    \end{tikzpicture}
    \caption{An example $G_f$ with levels marked.} \label{fig: graph structure}
\end{figure}
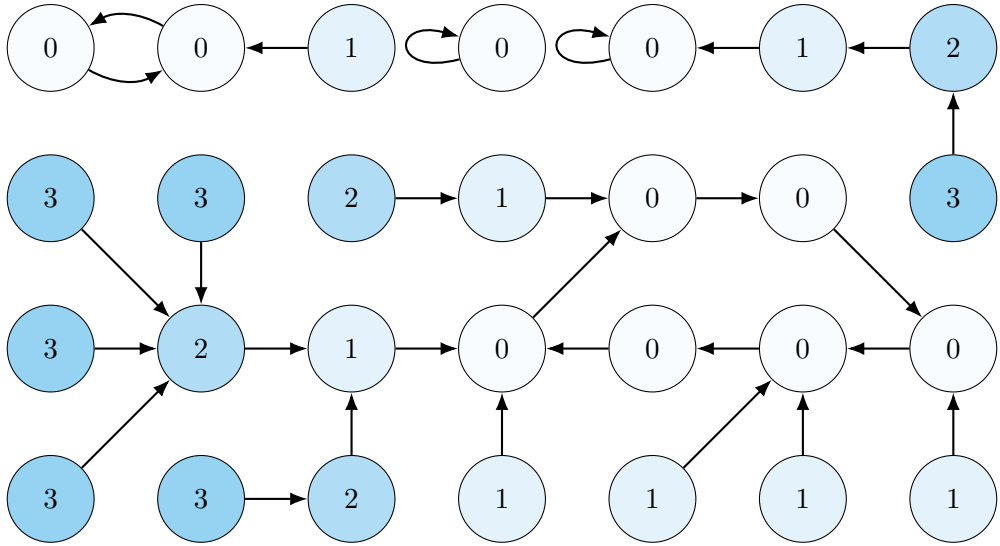

We give a reduction from $\ILP$ to $\BTreeLeaf$ under the assumption that we can find a vertex $v^*$ with $\ell(v^*) \geq 2$. After proving this, the hard part will be finding such a vertex. The main component of this reduction is the procedure $\FindChildren$ (\cref{alg: findchildren}), which is based on the functions $f$ and $g$ from an $\ILP$ instance. Define $\Checksol(u)$ to be the procedure that returns ${\sf True}$ iff $u$ is a solution to the $\ILP$ instance, that is, iff $f(f(g(u))) \neq f(u)$ or $f(g(u)) = u$.

\begin{algorithm}[h!] 
    \caption{Procedure $\FindChildren_{f, g}(v)$ } \label{alg: findchildren}
    \begin{algorithmic}[1]
        \State $h(v) \gets g(f(g(v)))$ 
        \Comment{Rename for notational brevity}
        \If{$\Checksol(v) \lor \Checksol\left(g (v)\right) \lor \Checksol \left(f(g(v))\right) \lor \Checksol \left(h(v)\right)$}
            \State \Return $(\bot, \bot)$ 
            \Comment{We have found a solution, so we can stop here}
        \Else
            \State \Return $(g(v), h(v))$
        \EndIf
    \end{algorithmic}
\end{algorithm}

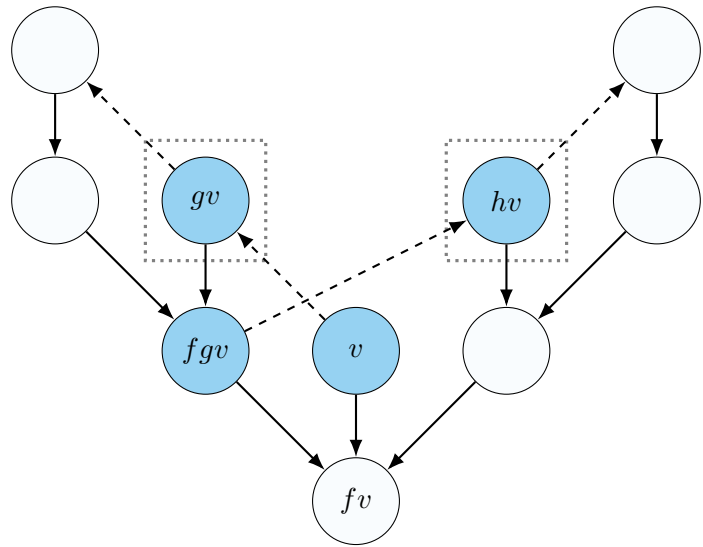
\begin{figure}
    \centering
    \begin{tikzpicture}[level2/.style={fill=Cerulean!40!white, text opacity=1},level0/.style={fill=Cerulean!2!white}] %
        \node[v,level0] (fv) at (0, 0) {$f v$};
        \node[v,level2] (v) at (0, 2) {$v$};
        \node[v,level2] (fw) at (-2, 2) {$f g v$};
        \node[v,level0] (fx) at (2, 2) {};
        \node[v, level2] (w) at (-2, 4) {$g v$};
        \node[v,level0] (leftw) at (-4, 4) {};
        \node[v, level2] (x) at (2, 4) {$h v$};
        \node[v,level0] (rightx) at (4, 4) {};
        \node[v,level0] (inlaww) at (-4, 6) {};
        \node[v,level0] (inlawx) at (4, 6) {};

        \draw[dotted, very thick,gray] (-2.8, 4.8) rectangle (-1.2, 3.2);
        \draw[dotted, very thick,gray] (2.8, 4.8) rectangle (1.2, 3.2);
        
        \draw[farrow] (v) to (fv);
        \draw[farrow] (fw) to (fv);
        \draw[farrow] (fx) to (fv);
        \draw[farrow] (w) to (fw);
        \draw[farrow] (leftw) to (fw);
        \draw[farrow] (x) to (fx);
        \draw[farrow] (rightx) to (fx);
        \draw[farrow] (inlaww) to (leftw);
        \draw[farrow] (inlawx) to (rightx);
        
        \draw[garrow] (v) to (w);
        \draw[garrow] (w) to (inlaww);
        \draw[garrow] (fw) to (x);
        \draw[garrow] (x) to (inlawx);
    \end{tikzpicture}
    \caption{The procedure performed by $\FindChildren_{f, g}(v)$.  Solid arrows represent $f$ and dashed arrows represent $g$. Parentheses are omitted in labels. The dotted boxes indicate that vertices $g(v)$ and $h(v)$ will be returned as the children of $v$. The procedure will check if the shaded vertices are $\ILP$ solutions, and in doing so will visit the unshaded vertices (but will not detect if these are solutions). Note that $f(h(v)) = v$ is possible, but $f(h(v)) = f(g(v))$ is not.} \label{fig: findchildren}
\end{figure}

See \cref{fig: findchildren} for an illustration. The idea is to construct the tree by defining the left and right children of $v$ to be two nodes reachable from $v$, unless the procedure finds a nearby solution, in which case we can make $v$ into a leaf, as it is easy to compute a solution given $v$. 

Although the intuition behind the reduction is straightforward, some work needs to be done to show that $\FindChildren$ gives a valid binary tree. The following properties will help.

\begin{lemma} \label{lem: findchildren properties}
    Let $\FindChildren_{f, g}(v) = (a, b)$. If $(a, b) \neq (\bot, \bot)$, then
    \begin{enumerate}[label=(\roman*)]
        \item \label{lem: findchildren properties distinct} $a \neq b$,
        \item \label{lem: findchildren properties distinct children} $f(a) \neq f(b)$,
        \item \label{lem: findchildren properties grandchildren} $f(f(a)) = f(f(b)) = f(v)$, and
        \item \label{lem: findchildren properties deeper nodes} if $\ell(v) \geq 2$, then $\ell(a) = \ell(b) = \ell(v) + 1$.
    \end{enumerate}
\end{lemma}
\begin{proof}
    \leavevmode %
    \begin{enumerate}[label=(\roman*)]
        \item Since $(a, b) \neq (\bot, \bot)$, we have that $f(g(v)) = f(a)$ is not an $\ILP$ solution, and in particular $f(g(f(a))) \neq f(a)$. This means that $a \neq b$, as $b = g(f(a))$.

        \item As $f(g(v))$ is not a solution, applying $f(g(u)) \neq u$ to $u = f(g(v))$ gives $f(b) = f(g(f(g(v)))) \neq f(g(v)) = f(a)$.
    
        \item As $v$ is not a solution, $f(f(a)) = f(f(g(v))) = f(v)$. As $f(g(v))$ is not a solution, applying $f(f(g(u))) = f(u)$ for $u = f(g(v))$ gives $f(f(b)) = f(f(g(f(g(v))))) = f(f(g(v))) = f(v)$.
    
        \item Under the assumption that $\ell(v) \geq 2$, we know that $\ell(f(v)) = \ell(v) - 1 > 0$, which implies that any vertex with $f(f(u)) = f(v)$ has $\ell(u) = \ell(f(v)) + 2 = \ell(v) + 1$. By \cref{lem: findchildren properties grandchildren}, this applies to $a$ and $b$. \qedhere
    \end{enumerate}
\end{proof}

\begin{remark}
    \cref{lem: findchildren properties}~\ref{lem: findchildren properties deeper nodes} does not hold without the hypothesis $\ell(v) \geq 2$, as if $\ell(f(v)) = 0$ there is no guarantee that vertices with edges pointing to $f(v)$ have level greater than 0. This is why it is necessary to assume $\ell(v^*) \geq 2$ for the simple reduction given here.
\end{remark}

\begin{lemma} \label{lem: findchildren reduction}
    Let $(V, f, g)$ be an instance of $\ILP$. Then define $L$ and $R$ by \[(L(v), R(v)) \gets \FindChildren_{f, g}(v).\] Given $v^* \in V$ with $\ell(v^*) \geq 2$, $(V, v^*, L, R)$ is a valid instance to $\BTreeLeaf$. 
\end{lemma}

\begin{proof}
    All we need to show is that $v^*$ is the root of an induced tree. Recall that $V^*$ denotes the set of vertices reachable from $v^*$ via $L$ and $R$ (i.e., using edges of the form $(u\to L(u))$ and $(u\to R(u))$). By \cref{lem: findchildren properties}~\ref{lem: findchildren properties deeper nodes}, the induced subgraph on $V^*$ can be assigned levels such that edges are directed only from lower levels to higher levels, i.e.\ it is a DAG. By \cref{lem: findchildren properties}~\ref{lem: findchildren properties distinct}, the DAG has outdegree 2. To prove this induced DAG is indeed a tree, it suffices to show that no two vertices share the same child. Indeed, if there are vertices $u, u'\in V^*$ such that $\{L(u), R(u)\} \cap \{L(u'), R(u')\} \ne \varnothing$, then by \cref{lem: findchildren properties}~\ref{lem: findchildren properties grandchildren}, we have $f(u) = f(u')$. Hence, it suffices to prove the following claim:

    \begin{claim} \label{claim: different children}
        If $u, u' \in V^*$ are two different vertices, then $f(u) \neq f(u')$.
    \end{claim}
    \begin{claimproof}[Proof of \cref{claim: different children}]
        We may assume that $\ell(u) = \ell(u')$, as otherwise $f(u) \neq f(u')$ by the observation about levels in $V^*$. The proof is by induction on levels. At level $\ell(v^*)$, there is only one vertex, hence the base case is trivially true. Assume that at level $k$ (where $k \geq \ell(v^*)$), no two vertices in $V^*$ map via $f$ to the same vertex. We will prove that it also holds for level $k+1$. 

        For two different vertices $u, u'$ with $\ell(u) = \ell(u') = k+1$, we may assume $u$ and $u'$ are the children of two different vertices in $V^*$: if they are children of the same vertex, by \cref{lem: findchildren properties}~\ref{lem: findchildren properties distinct children} $\{u, u'\} = \{L(w), R(w)\}$ implies $f(u) \neq f(u')$ and we are done. So, let $w \neq w'$ be such that $u \in \{L(w), R(w)\}$ and $u' \in \{L(w'), R(w')\}$. Then $\ell(w) = \ell(w') = k$ and by the inductive hypothesis $f(w) \neq f(w')$. 
        
        Assume that $f(u) = f(u')$. We will prove shortly that all vertices in $V^*$ that map to $u$ or $u'$ by $L$ or $R$ map to some common vertex $z$ by $f$, a contradiction with $f(w) \neq f(w')$. Indeed, set $z = f(f(u)) = f(f(u'))$. If a vertex $x \in V^*$ has $L(x) = g(x) = u$ (the same arguments will apply to $x$ that maps to $u'$), then \[z = f(f(u)) = f(f(g(x))) = f(x).\] 
        If a vertex $x \in V^*$ has $R(x) = g(f(g(x))) = u$, then \[z = f(f(u)) = f(f(g(f(g(x))))) = f(f(g(x))) = f(x), \] where we have applied $f(f(g(y))) = f(y)$ to $y = f(g(x))$ which we know is not a $\ILP$ solution because it was checked by $\FindChildren$.
    \end{claimproof}

    See \cref{fig: different children claim} for an illustration of the argument behind \cref{claim: different children}.
\end{proof}

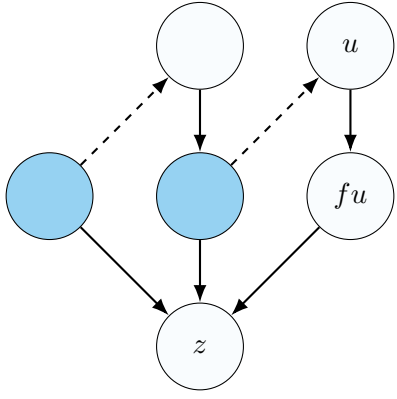
\begin{figure}
    \centering
    \begin{tikzpicture}[level2/.style={fill=Cerulean!40!white, text opacity=1},level0/.style={fill=Cerulean!2!white}] %
        \node[v, level2] (x) at (0, 0) {};
        \node[v, level2] (xx) at (-2, 0) {};
        \node[v,level0] (gxx) at (0, 2) {};
        \node[v,level0] (u) at (2, 2) {$u$};
        \node[v,level0] (fu) at (2, 0) {$f u$};
        \node[v,level0] (z) at (0, -2) {$z$};

        \draw[farrow] (x) to (z);
        \draw[farrow] (xx) to (z);
        \draw[farrow] (gxx) to (x);
        \draw[farrow] (u) to (fu);
        \draw[farrow] (fu) to (z);
        
        \draw[garrow] (xx) to (gxx);
        \draw[garrow] (x) to (u);
    \end{tikzpicture}
    \caption{The argument behind \cref{claim: different children}. Solid arrows represent $f$ and dashed arrows represent $g$. Parentheses are omitted in labels. The shaded vertices are possible locations of the $w \in V^*$ such that either $L(w) = u$ or $R(w) = u$.} \label{fig: different children claim}
\end{figure}

\paragraph{Completing the argument.} It remains to find a vertex $v^*$ with $\ell(v^*)\ge 2$. We do not know how to achieve this deterministically; instead, we will find a pair of vertices $(v, v')$ such that \emph{either} $\ell(v) \ge 2$ or $\ell(v') \ge 2$.

\begin{lemma} \label{lem: pair of vertices}
    Let $u \in V$ be an arbitrary starting vertex. Let $v = g(u)$ and $v' = g(f(u))$. Then either
    \begin{itemize}
        \item at least one of $\{u, v, v'\}$ is a solution, or
        \item at least one of $\{v, v'\}$ is at level at least 2.
    \end{itemize}
\end{lemma}
\begin{proof}
    Assume that there are no solutions among $\{u, v, v'\}$. There are three cases:
    \begin{itemize}
        \item {\bf Case I:} Suppose that $\ell(u) = 0$. We claim that $\ell(v) = 2$ in this case. Indeed, $u$ is the only vertex at level $0$ that maps via $f$ to $f(u)$. Since $f(f(v)) = f(f(g(u))) = f(u)$ but $f(v) = f(g(u)) \ne u$, $f(v)$ has to be at level $1$, hence $v$ is at level $2$.
        \item {\bf Case II:} Suppose that $\ell(u) = 1$. In this case, we have $\ell(f(u)) = 0$, hence the same argument as above applies to show that $v' = g(f(u))$ is at level $2$.
        \item {\bf Case III:} Suppose that $\ell(u) \ge 2$. In this case, $f(u)$ is at level at least $1$. Since $f(f(v)) = f(f(g(u))) = f(u)$, we have that $v$ is at level at least $2$.\qedhere
    \end{itemize}
\end{proof}

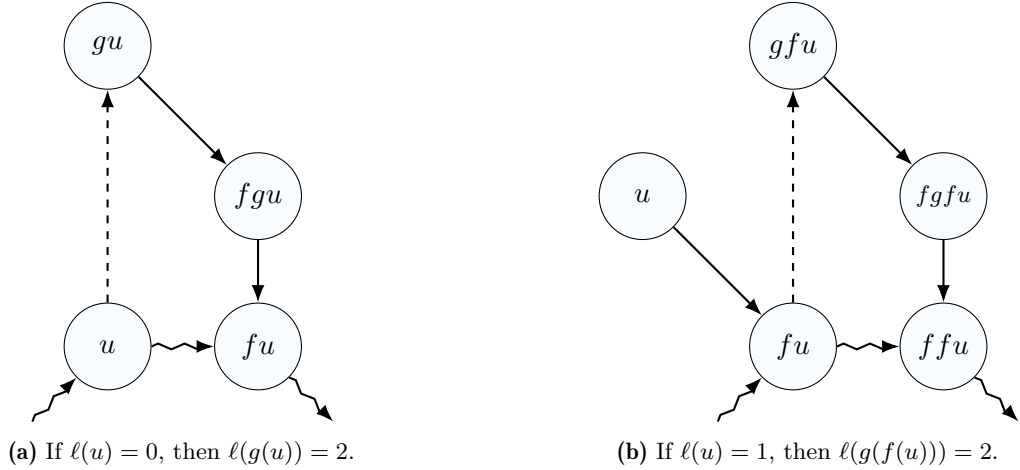
\begin{figure}
  \centering
  \begin{subfigure}[b]{0.45\textwidth}
    \centering
    \begin{tikzpicture}[level0/.style={fill=Cerulean!2!white}]
        \node[v,level0] (u) at (0, 0) {$u$};
        \node[v,level0] (fu) at (2, 0) {$f u$};
        \node[v,level0] (gu) at (0, 4) {$g u$};
        \node[v,level0] (fgu) at (2, 2) {$f g u$};

        \draw[circlearrow] (-1, -1) to (u);
        \draw[circlearrow] (u) to (fu);
        \draw[circlearrow] (fu) to (3, -1);

        \draw[farrow] (gu) to (fgu);
        \draw[farrow] (fgu) to (fu);

        \draw[garrow] (u) to (gu);
    \end{tikzpicture}
    \caption{If $\ell(u) = 0$, then $\ell(g(u)) = 2$.}
    \label{fig: level 0}
  \end{subfigure}
  \quad \quad 
  \begin{subfigure}[b]{0.45\textwidth}
    \centering
    \begin{tikzpicture}[level0/.style={fill=Cerulean!2!white}]
        \node[v,level0] (u) at (-2, 2) {$u$};
        \node[v,level0] (ffu) at (2, 0) {$f f u$};
        \node[v,level0] (fu) at (0, 0) {$f u$};
        \node[v,level0] (gfu) at (0, 4) {$g f u$};
        \node[v,level0] (fgfu) at (2, 2) {\footnotesize $f g f u$};

        \draw[circlearrow] (-1, -1) to (fu);
        \draw[circlearrow] (fu) to (ffu);
        \draw[circlearrow] (ffu) to (3, -1);

        \draw[farrow] (u) to (fu);
        \draw[farrow] (gfu) to (fgfu);
        \draw[farrow] (fgfu) to (ffu);

        \draw[garrow] (fu) to (gfu);
    \end{tikzpicture}
    \caption{If $\ell(u) = 1$, then $\ell(g(f(u))) = 2$.}
    \label{fig: level 1}
  \end{subfigure}
  \caption{Illustration of cases I and II of \cref{lem: pair of vertices}. Solid arrows represent $f$ and dashed arrows represent $g$. Zigzag $f$ arrows are between vertices of level 0. Parentheses are omitted in labels.}
  \label{fig: circle}
\end{figure}

See \cref{fig: circle} for an illustration of cases I and II of \cref{lem: pair of vertices}.

Now we are ready to prove our main results. Recall the statement of \cref{thm: tfzpp-and-pwpp-containment}:
\inTFZPPandPWPP*
\begin{proof}
    The proof of $\ILP \in \TFZPP$ is easy. Pick an arbitrary vertex $u$ and define $v = g(u)$ and $v' = g(f(u))$. By \autoref{lem: pair of vertices}, if none of $u, v, v'$ are solutions of $\ILP$, then at least one of $v$ and $v'$ is at level at least $2$. Therefore, we can randomly select one vertex in $\{v, v'\}$ as the root $v^*$ and run the randomized algorithm for $\BTreeLeaf$ on the instance $(V, v^*, L, R)$. The correctness is guaranteed by \autoref{lem: findchildren reduction}.

    Now we prove that $\ILP\in\PWPP$. Intuitively, we will run the reduction of \cref{lem: findchildren reduction} in parallel on (the direct product of) two graphs where the root is $v$ in one graph and is $v'$ in the other. Because we can only guarantee that our starting point is at level $\ge 2$ in only one of these graphs, we only know that one of the graphs has an induced tree rooted at its respective starting point by the reduction process. Fortunately, the overall graph on $V \times V$ will have a rooted induced tree.

    By \cref{lem: pair of vertices}, starting with an arbitrary vertex, we can obtain either a $\ILP$ solution (in which case we are done) or a pair $(v, v')$ where at least one is at level at least 2. We will reduce to $\BTreeLeaf$ on the vertex set $V \times V$. Set
    \[ (A(u, u'), B(u, u')) \gets \FindChildren_{f, g}(u), \]
    \[ (C(u, u'), D(u, u')) \gets \FindChildren_{f, g}(u'). \]
    Set $L(u, u') = (A(u, u'), C(u, u'))$, unless one of $A(u, u'), C(u, u')$ is $\bot$, in which case $L(u, u') = \bot$. Set $R(u, u') = (B(u, u'), D(u, u'))$, again unless one side of the pair is $\bot$ in which case $R(u, u') = \bot$. By the construction of $\FindChildren$, if one of $L, R$ is set to $\bot$ the other one will as well. The special vertex in $V \times V$ will be $(v, v')$. 
    
    Then $V^* \subset V \times V$ induces a binary tree (recall that $V^*$ is the set of vertices reachable from $(v, v')$ by $L, R$). Say that $\ell(v) \geq 2$; the other case follows by symmetry. Starting at $(v, v')$, $L$ and $R$ will yield a tree structure on the first member of the pair by the same argument as \cref{lem: findchildren reduction}. Let $p, p'$ be two distinct paths of any length which start at $(v, v')$ and traverse with $L$ and $R$. The vertex reached by $p$ will have a different first pair element than the vertex reached by $p'$ (by the tree structure on the first element of the pair), and thus $p$ and $p'$ arrive at distinct vertices. This means that there are no (undirected) cycles in the induced graph on $V^*$, and so it is a tree.
\end{proof}

\section{Finding a Leaf in a Binary Tree}\label{sec: Empty-Child}

In the previous section, we reduced $\ILP$ to the promise problem $\BTreeLeaf$, which captures the following principle:
\begin{center}
	\emph{Randomly walking down a binary tree will reach a leaf in logarithmic time in expectation.}
\end{center}

In this section, we study a family of \emph{total} (instead of promise) search problems that capture the same principle. We then study the power of these search problems. 

The problem $\emptychild$ captures the task of finding a leaf in a tree where each non-leaf vertex $v$ has at least two different children $L(v)$ and $R(v)$. This is made total (i.e.\ we do not need the promises as in $\BTreeLeaf$) by adding a parent function, denoted $F(v)$. 

\begin{description}
    \item[\emptychild.] The input is a set of vertices $V$ and three functions $F, L, R: V \rightarrow V$, where $F(u)$ is the father of $u$, and $L(u), R(u)$ are the left and the right child of $u$, respectively. There are two possible solutions:
	\begin{description}
		\item[s1.] $u \in V$ such that $F(L(u)) \neq u$ or $F(R(u)) \neq u$ or $L(u) = R(u) \neq u$. \mylabel{Empty Child}{item: (Empty Child)}
		\item[s2.] $1$, if $L(1) = 1$ or $R(1) = 1$ or $F(1)\ne 1$. \mylabel{Wrong Root}{item: (Wrong Root)}
	\end{description}
\end{description}

\paragraph*{Intuition behind $\emptychild$.} We think of $\emptychild$ as defining a directed graph on $V$. Draw an edge from $v$ to $w$ if $F(w) = v$ and either $L(v) = w$ or $R(v) = w$. If $v$ does not point to two distinct vertices in this graph, then $v$ has an ``empty child'' and $v$ is a leaf\hspace{0.07em}\footnote{We slightly stretch the definition of leaf to include vertices with only one child. Equivalently, we could redefine the graph such that a vertex may only have two or zero children.} of the graph.

The directed graph defined above may have multiple connected components. They will be of the following sorts:
\begin{itemize}
	\item Isolated vertices. We may have $F(v) = L(v) = R(v) = v$.
	\item Trees. 
	\item Cycles with trees rooted at each node in the cycle. 
\end{itemize}

Unless there is a solution of type \ref{item: (Wrong Root)}, we are forced to have at least one tree, rooted at vertex 1. This forces the graph to have structures other than just isolated vertices. Then we can find a leaf of that tree using the principle captured by $\BTreeLeaf$. An important observation is that in the third type of structure---cycles with trees---walking randomly will also reach a leaf in logarithmic time. We will show shortly that solution \ref{item: (Wrong Root)} can be relaxed to allow for graphs without trees, as long as at least one cycle-with-trees is present.

In $\emptychild$, it is possible that there are vertices $v, w$ such that $F(w) = v$ but $L(v) \neq w$ and $R(v) \neq w$. This is fine: structurally, we can interpret $w$ as being the root of a new tree. We will consider a variant of $\emptychild$ where such a situation is forbidden.

In $\BTreeLeaf$, we needed a root-to-leaf path to prove that the leaf vertex was in the promised tree. In contrast, here we enforce the tree structure syntactically, and hence do not require a path. Indeed, a solution to $\emptychild$ need not be in the connected component that contains the canonical root vertex.

\paragraph*{Results about $\emptychild$.} The main result of this section is the $\Lossy$-completeness of $\emptychild$:
\begin{theorem}\label{thm: empty child main}
	$\emptychild$ is equivalent to $\lossycode$.
\end{theorem}

The problem from the previous section, $\ILP$, is at least as strong as $\emptychild$: we give a reduction in \autoref{sec: EC: REC}. We conjecture that $\ILP$ is strictly more powerful than $\emptychild$. A variant of $\ILP$ is considered in \cref{sec: nephew with inverse} and is shown to be equivalent to $\emptychild$.

\paragraph*{Variants of $\emptychild$.} For a proof later in this section it will be useful to weaken \ref{item: (Wrong Root)}. Instead of enforcing that a \emph{root} vertex exists, we only enforce that an \emph{internal} vertex that is not a self-loop exists.

\begin{description}
    \item[$\emptychild'$.] Given the same inputs as in \emptychild, we accept the solutions \ref{item: (Empty Child)} and:
    \begin{description}
        \item[s2a.] $1$, if $L(1) = 1$ or $R(1) = 1$. \mylabel{Wrong Internal Vertex}{item: (Wrong Internal Vertex)}
    \end{description}
\end{description}

It turns out that $\emptychild'$ is equivalent to $\emptychild$ in terms of decision tree reductions, but we do not know of a direct way to prove this. Instead, we will use $\lossycode$ as an intermediate step in the reduction.

\begin{theorem} \label{thm: emptychild equiv to weakened version}
    $\emptychild$ and $\emptychild'$ are equivalent under black-box reductions.
\end{theorem}

It is easy to see that $\emptychild$ reduces to $\emptychild'$, as the difference between the two is just a weakening of one of the solutions. In \cref{sec: empty-child to lossycode}, we will prove that $\emptychild'$ reduces to $\lossycode$ (\cref{lemma: emptychild in lossy}). In \cref{sec: lossycode to empty-child}, we will prove that $\lossycode$ reduces to $\emptychild$ (\cref{thm: lossycode to emptychild}), completing the proof of \cref{thm: emptychild equiv to weakened version}.

We also consider a stricter variant where the $F$ function and the $L$ and $R$ functions are required to ``agree'':

\begin{description}
    \item[$\BEC$.] Given the same inputs as in \emptychild but in addition to \ref{item: (Empty Child)} and \ref{item: (Wrong Root)}, we also accept the following solution:
	\begin{description}
		\item[s3.] a vertex $u \in V \backslash \{1\}$ such that $u \notin \{L(F(u)), R(F(u))\}$. \mylabel{Wrong Father}{item: (Wrong Father)}
	\end{description}
\end{description}

Finally, we will consider variants of $\emptychild$ with added \emph{height} functions in \cref{sec: empty child with height}. The analysis of these will be a warm-up for the proof that $\lossycode$ reduces to $\emptychild$ in \cref{sec: lossycode to empty-child}.

    \subsection{\texorpdfstring{$\emptychild$}{Empty-Child} reduces to \texorpdfstring{$\lossycode$}{Lossy-Code}} \label{sec: empty-child to lossycode}
	It is easy to see that $\emptychild$ is in $\TFZPP$: Start from the root $1$ and walk down the tree using $L$ or $R$ randomly until we reach a leaf. In fact, the same proof reduces $\emptychild$ (in fact, also $\emptychild'$) to $\lossycode$:%
    \begin{lemma}\label{lemma: emptychild in lossy}
        $\emptychild'$ reduces to \lossycode via a black-box reduction.
    \end{lemma}
    \begin{proof}
        Consider an instance of $\emptychild'$ $(V, F, L, R)$ and let $\ell \coloneqq \lceil\log |V|\rceil$. We construct a \lossycode instance $(f,g)$ with $N \coloneqq 2^\ell$. Intuitively, $f: [2N] \rightarrow [N]$ maps a tree-path of length $\ell+1$ to a vertex in $u$ by travelling downwards in the tree, and $g: [N] \rightarrow [2N]$ maps a vertex in $V$ to a tree-path by travelling upwards.
        
		More formally, identify $[2N]$ with $\{\LL, \RR\}^{\ell+1}$, the space of length-$(\ell+1)$ strings over the alphabet $\{\LL, \RR\}$. Every such string $x\in \{\LL, \RR\}^{\ell+1}$ corresponds to a path of length $\ell+1$ starting from the root: Let $v_0 \coloneqq 1$ be the root and for each $i \ge 1$,
        \[v_i \coloneqq \begin{cases} L(v_{i-1}) & \text{if }x_i = \LL; \\ R(v_{i-1}) & \text{if }x_i = \RR.\end{cases}\]
        Then, $x$ corresponds to the path $p_x \coloneqq (v_0, v_1, \dots, v_{\ell+1})$. We define $f(x)$ to be $v_{\ell+1}$, the endpoint of this path. For instance, $f(\LL \LL \RR) = R(L(L(1)))$ when $\ell = 2$.
        
        Conversely, for any $u \in [N]$, we can construct a string $x \in \{\LL, \RR\}^{\ell+1}$ by traversing $\ell+1$ steps upwards from $u$: in the $i$-th step, we set $x_{\ell+2-i} \leftarrow \LL$ if $u = L(F(u))$, or $x_{\ell+2-i} \leftarrow \RR$ otherwise\footnote{It is possible that $u \notin \{L(F(u)), R(F(u))\}$.}; then update $u \leftarrow F(u)$. We set $g(u) = x$.

		Now we prove the correctness of the reduction. Let $x\in [2N]$ be such that $g(f(x)) \ne x$, we use a case analysis to find a solution of $\emptychild'$ from $x$. First we find the path $p_x = (v_0, v_1, \dots, v_{\ell+1})$ as defined above. Normally, two adjacent vertices in the path should be distinct (since a node should be different from its child); we deal with the abnormal case before proceeding. Suppose there exists some $i\ge 0$ such that $v_i = v_{i+1}$, and let $i$ be the smallest such index.
		\begin{itemize}
			\item If $i = 0$, then either $L(1) = 1$ or $R(1) = 1$, hence $1$ is a \ref{item: (Wrong Internal Vertex)}.%
			\item Otherwise, we have $v_{i-1} \ne v_i$. Without loss of generality, let us assume that $v_i = L(v_{i-1})$ and $v_{i+1} = L(v_i)$ ($= v_i$); the cases where we walk down the right child are completely symmetric. \begin{itemize}
				\item If $F(v_i) \ne v_{i-1}$, then $F(L(v_{i-1})) \ne v_{i-1}$ and hence $v_{i-1}$ has an \ref{item: (Empty Child)}.%
				\item Otherwise, $F(L(v_i)) = F(v_i) = v_{i-1} \ne v_i$ and hence $v_i$ has an \ref{item: (Empty Child)}.%
			\end{itemize}
		\end{itemize}
		
		Now we are in the ``normal'' case where for every $i\ge 0$, $v_i \ne v_{i+1}$. Let $x' \coloneqq g(f(x))$. Recall that the precise definition of $x'$ is as follows: let $u_{\ell + 1} \coloneqq v_{\ell+1}$ ($=f(x)$), and let $u_i \coloneqq F(u_{i+1})$ for each $i$ from $\ell$ down to $0$, then $x'_i = 0$ if $u_i = L(F(u_i))$ and $x'_i = 1$ otherwise.
		
		Let $i$ be the largest index such that either $x'_i \ne x_i$ or $v_{i-1} \ne u_{i-1}$. (Note that such $i$ exists because $x' = g(f(x)) \ne x$.) Then $v_i = u_i$. We claim that $v_{i-1}$ has an \ref{item: (Empty Child)}, which would finish the proof.
		\begin{itemize}
			\item If $v_{i-1} \ne u_{i-1} = F(v_i)$, then $v_{i-1} \ne F(L(v_{i-1}))$ or $v_{i-1} \ne F(R(v_{i-1}))$ depending on whether $x_i$ is $0$ or $1$. In either case, $v_{i-1}$ has an \ref{item: (Empty Child)}.
			\item If $x_i = 0$ and $x'_i = 1$, then $v_i = L(v_{i-1})$ but $u_i \ne L(F(u_i))$ (hence $v_i \ne L(F(v_i))$). Note that this implies that $v_{i-1} \ne F(v_i)$, hence $F(L(v_{i-1})) = F(v_i) \ne v_{i-1}$ and $v_{i-1}$ has an \ref{item: (Empty Child)}.
			\item If $x_i = 1$ and $x'_i = 0$, then $v_i = R(v_{i-1})$ but $u_i = L(F(u_i))$ (hence $v_i = L(F(v_i))$). This implies that $F(R(v_{i-1})) = F(v_i)$. If $F(v_i) \ne v_{i-1}$, then $v_{i-1}$ has an \ref{item: (Empty Child)}. Otherwise we have $L(v_{i-1}) = R(v_{i-1}) = v_i$ and $v_{i-1} \ne v_i$ (by the reasoning above) and $v_{i-1}$ has an \ref{item: (Empty Child)}.\qedhere
		\end{itemize}
    \end{proof}

    \subsection{Finding a Leaf with Heights}\label{sec: empty child with height}

	Before showing that $\emptychild$ is complete for $\Lossy$, we take a detour to study variants of $\emptychild$ with \emph{heights}. Indeed, the proof that $\emptychild$ is $\Lossy$-complete is inspired by the investigations of these variants. We first define the variants of $\emptychild$ and $\BEC$ with heights:

    \def\BECwH{\textsc{Binary-Empty-Child-w-Height}}
    \begin{description}
        \item[\ECwH.] The input is a set $V$ of vertices and four functions $F, L, R: V \rightarrow V$, $H: V \rightarrow [|V|]$. Besides \ref{item: (Empty Child)} and \ref{item: (Wrong Root)} listed in the definition of \emptychild, the following solution is also valid:
        \begin{description}
            \item[s4.] a vertex $u \in V \backslash \{1 \}$ such that $u\ne F(u)$ and $H(u) \neq H(F(u))+1$; or $1$, if $H(1) \neq 1$.
            
            \mylabel{Wrong Height}{item: (Wrong Height)}
        \end{description}
    \end{description}
    \begin{description}
        \item [$\BECwH$.] All of \ref{item: (Empty Child)}, \ref{item: (Wrong Root)}, \ref{item: (Wrong Father)}, and \ref{item: (Wrong Height)} are valid solutions.
    \end{description}

	Again, we remark that an isolated vertex $u$ with $F(u) = L(u) = R(u) = u$ is \emph{not} a solution of $\ECwH$ or $\BECwH$.

	Due to the presence of heights, these problems are in $\PLS$ now:
	\begin{lemma}\label{lemma: ECwH in SOPL}
		\ECwH is in \PLS, and therefore, \SOPL.
	\end{lemma}
	\begin{proof}
		We reduce the $\ECwH$ instance to the instance of $\textsc{Sink-Of-DAG}$ where for each node $v\in V$, the successor of $v$ is $L(v)$ and the potential of $v$ is $H(v)$. Let $v$ be a solution of $\textsc{Sink-Of-DAG}$, then one of the following cases happens:
		\begin{itemize}
			\item $v = 1$ and $L(v) = 1$. Then $1$ is a \ref{item: (Wrong Root)}.
			\item $L(v) \ne v$ and $L(L(v)) = L(v)$. This implies that either $F(L(v)) \ne v$ or $F(L(L(v))) \ne L(v)$, hence either $v$ or $L(v)$ is a node with \ref{item: (Empty Child)}.
			\item $L(v) \ne v$ and $H(L(v)) \le H(v)$. If $F(L(v)) \ne v$ then $v$ is a node with \ref{item: (Empty Child)}. Otherwise, let $u \coloneqq L(v)$, we have that $u\ne F(u) = v$ and that $H(u) \le H(F(u))$, hence $u$ has \ref{item: (Wrong Height)}.
		\end{itemize}

		Therefore, $\ECwH$ reduces to $\textsc{Sink-Of-DAG}$ and is in $\PLS$. Finally, since $\ECwH$ is also in $\Lossy \subseteq \PPADS$, it is in $\SOPL = \PLS \cap \PPADS$~\cite{GoosHJMPRT24}.
	\end{proof}

	In contrast, it is easy to show that $\emptychild$ (as well as $\lossycode$ and $\lossycodep$) is not in $\PLS$ using Prover-Delayer games and resolution width lower bounds (see, e.g., \cite[Proposition 3.4]{PudlakT19}).

	\begin{remark}
		Consider the following seemingly harder variant of $\ECwH$ where the definition of \ref{item: (Wrong Height)} is changed to
		\begin{description}
			\item [s4'.] a vertex $u \in V\setminus \{1\}$ such that \underline{$H(u) \le H(F(u))$}.\mylabel{Wrong Height'}{item: (Wrong Height')}
		\end{description}
		Call this variant $\ECwH'$. We can define $\BECwH'$ likewise.
        Clearly, the proof of \autoref{lemma: ECwH in SOPL} also shows that $\ECwH'$ is in $\PLS$ (hence $\SOPL$). Since $\ECwH$ is $\Lossy\cap\SOPL$-complete (as we will show in \autoref{thm: rwPHP-cap-EOPL to empty-child}), it follows that $\ECwH'$ is equivalent to $\ECwH$.

		The difference between \ref{item: (Wrong Height)} and \ref{item: (Wrong Height')} resembles the difference between \textsc{Sink-Of-Metered-Line} and \textsc{Sink-Of-Potential-Line}. Indeed, our proof below (that $\ECwH'$ reduces to $\ECwH$) starts from the fact that \textsc{Sink-Of-Metered-Line} is $\SOPL$-complete, hence it makes black-box use of the reduction from \textsc{Sink-Of-Potential-Line} to \textsc{Sink-Of-Metered-Line} \cite{FearnleyGMS20}.\footnote{Although \cite{FearnleyGMS20} only claims a reduction from \textsc{End-Of-Potential-Line} to \textsc{End-Of-Metered-Line}, the same reduction also proves that \textsc{Sink-Of-Potential-Line} reduces to \textsc{Sink-Of-Metered-Line}.}
	\end{remark}

	\def\SoML{\textsc{Sink-Of-Metered-Line}}
    \def\EoML{\textsc{End-Of-Metered-Line}}
	Next, we show that $\ECwH$ is complete for the class $\Lossy \cap \SOPL$. The proof is inspired by the recent intersection results in $\TFNP$: $\CLS = \PPAD \cap \PLS$~\cite{FearnleyGHS23}, $\EOPL = \PPAD\cap\PLS$, and $\SOPL = \PPADS \cap \PLS$~\cite{GoosHJMPRT24}. Given a $\PPAD$ instance $A$ and a $\PLS$ instance $B$, \cite{FearnleyGHS23} showed how to ``combine'' $A$ and $B$ into a $\CLS$ instance $C$ such that a solution of $C$ implies either a solution of $A$ or a solution of $B$. Our idea is similar here: Given a $\lossycode$ instance $(f, g)$ and a $\SoML$ instance $(S, P, V)$, we show how to create an instance $(F, L, R, H)$ of $\ECwH$ such that a solution of $(F, L, R, H)$ implies either a solution of $(f, g)$ or a solution of $(S, P, V)$. Roughly speaking, we use $(f, g)$ to implicitly define an exponentially large tree and use $(S, P, V)$ to define the heights on the tree.

    \begin{description}
        \item[$\SoML$ {\cite{HubacekY20,FearnleyGMS20}}.]  Given functions $S, P: [N] \to [N]$ and $V: [N] \to [N] \cup \{0\}$ there are four types of possible solutions:
        \begin{description}
            \item[s1.] $1$, if $P(1) \ne 1$ or $S(1) = 1$ or $V(1) \ne 1$,  \mylabel{Bad Source}{item: (Bad Source)}
            \item[s2.] A vertex $x\in [N]$ such that $P(S(x)) \ne x$, \mylabel{Sink of Line}{item: (Sink of Line)}
            \item[s3.] A vertex $x\in [N]\setminus \{1\}$ such that $V(x) = 1$, \mylabel{Bad Meter I}{item: (Bad Meter I)}
            \item[s4.] A vertex $x\in [N]$ such that ($V(x) > 0$ and $V(S(x)) - V(x) \ne 1$) or ($V(x) > 1$ and $V(x) - V(P(x)) \ne 1$). \mylabel{Bad Meter II}{item: (Bad Meter II)}
        \end{description}
    \end{description}

    \begin{theorem}\label{thm: rwPHP-cap-EOPL to empty-child}
        $\lossycode \cap \SoML$ reduces to $\ECwH$ by a black-box reduction.
    \end{theorem}
    \begin{proof}
        An input of $\lossycode \cap \SoML$ consists of an instance $(f, g)$ of $\lossycode$ and an instance $(S, P, V)$ of $\SoML$; either a solution for $(f, g)$ or a solution for $(S, P, V)$ is a valid solution. Now we show how to reduce this instance $(f, g, S, P, V)$ to an instance of $\ECwH$. Let $N, M$ denote the dimensions of the inputs, where $f: [N] \to [2N]$, $g: [2N] \to [N]$, $S, P: [M] \to [M]$, and $V: [M] \to [M] \cup \{0\}$. Without loss of generality, we may assume that $N = 2^n$ is a power of $2$ (this follows from the robustness of $\lossycode$ as shown in \autoref{lemma: robustness of lossy-code}).

        We will create a ``binary tree'' with $n + M$ levels. Roughly speaking, the first $n$ levels form a complete binary tree, and each of the last $M$ levels contains $N$ vertices. Let $(i, j)$ denote the $j$-th vertex in the $i$-the level; we abuse notation and write $(i, j)$ also as the \emph{index} of the node $(i, j)$, hence
       \[
        (i, j) = \begin{cases}
            j+2^{i-1}-1 & \text{if }i \le n,\\
            j+N\cdot (i-n)-1 & \text{otherwise}.
        \end{cases}
        \]

        We now specify the structure of the tree.
        \begin{itemize}
            \item The father of the root is $F(1, 1) \coloneqq (1, 1)$.
            
            \item For every $2\le i\le n + 1$ and $j\in [2^{i-1}]$, $F(i, j) \coloneqq (i-1, \lceil j/2\rceil)$.
            
            \item For every $2\le i \le M$ and $j\in [N]$ s.t.~$V(i)\ne 0$, $F(i + n, j) \coloneqq (P(i) + n, \lceil f(j) / 2\rceil)$.

            \item For every $2\le i \le M$ and $j\in [N]$ s.t.~$V(i)= 0$, $F(i + n, j) \coloneqq (i+n,j)$.
            
            \item For every $i\in [n]$ and $j \in [2^{i-1}]$, the left and right children of $(i, j)$ are $L(i, j) \coloneqq (i+1, 2j-1)$ and $R(i, j) \coloneqq (i+1, 2j)$ respectively.
            
            \item For every $i\in [M]$ and $j\in [N]$ s.t.~$V(i)\ne 0$, $L(i + n, j) \coloneqq (S(i) + n, g(2j-1))$ and $R(i + n, j) \coloneqq (S(i) + n, g(2j))$.

            \item For every $i\in [M]$ and $j\in [N]$ s.t.~$V(i)= 0$, $L(i + n, j) \coloneqq R(i+n,j)\coloneqq (i+n,j)$.
            
            \item For every $i\in [n]$, the height of every node in the $i$'th level is $i$.
            
            \item For every $i\in [M]$, the height of every node in the $(i+n)$'th level is $V(i) + n$.
        \end{itemize}

        \begin{figure}
            \centering
            \resizebox{\linewidth}{!}{
                \begin{tikzpicture}[
    aa/.style={draw, circle, inner sep=2pt, font=\small, minimum size = 1.1cm, fill=Cerulean!3!white},
    text label/.style={draw=none, font=\large, inner sep=0pt, minimum size = 0.7cm},
]

\node[aa] {1, 1} [sibling distance = 3cm, level distance = 1.2cm]
    child {node[aa] {2, 1}[sibling distance = 1.5cm]
        child {node[aa] {3, 1}}
        child {node[aa] {3, 2}}}
    child {node[aa] {2, 2}[sibling distance = 1.5cm]
        child {node[aa] {3, 3}}
        child {node[aa] {3, 4}}};
\node[text label, rotate=90, anchor=center] at (0, -3.5) {$\cdots$};
\node[aa] at (-3.5, -4.5) {\tiny $n+1, 1$} [level distance = 1.2cm]
    child {node {}};
\node[aa] at ( 3.5, -4.5) {\tiny $n+1, N$} [level distance = 1.2cm]
    child {node {}};

\node[aa] (IJ) at (1.2, -6.1) {\tiny $n+i, j$};
\draw (IJ) -- (1.2, -5);
\draw (IJ) -- (0.7, -7.2);
\draw (IJ) -- (1.7, -7.2);

\draw[dashed] (0.5, -5.4) -- (0.5, -6.8) -- (1.9, -6.8) -- (1.9, -5.4) -- (0.6, -5.4);

\node[text label, rotate=90, anchor=center] at (-3.5, -6.35) {$\cdots$};
\node[text label, rotate=90, anchor=center] at ( 3.5, -6.35) {$\cdots$};

\node at (-3.5, -7.1) {} [level distance = 1.2cm]
    child {node[aa] {}};
\node at ( 3.5, -7.1) {} [level distance = 1.2cm]
    child {node[aa] {}};

\node[text label, anchor=center] at (0, -4.5) {$\cdots$};
\node[text label, anchor=center] at (0, -8.3) {$\cdots$};

\draw[decoration={brace, amplitude=10pt, mirror}, decorate]
    (-4.2, 0) -- node[text label, left=0.4cm, align=center, anchor=east, font=\small] {complete\\ binary tree} (-4.2, -4.4);
\draw[decoration={brace, amplitude=10pt, mirror}, decorate]
    (-4.2, -4.6) -- node[text label, left=0.4cm, align=center, anchor=east, font=\small] {$M$ levels of\\ $N$ vertices each} (-4.2, -8.4);

\draw[thick, dashed] (7.5, 1) -- (19.5, 1) -- (19.5, -9) -- (7.5, -9) -- (7.5, 1);
\draw[thick, dotted] (1.9, -5.4) -- (7.5, 1);
\draw[thick, dotted] (1.9, -6.8) -- (7.5, -9);

\node [aa] (a) at (13, 0) {$\lceil f(j) / 2\rceil$};
\node [aa, dashed, fill=white] (b) at (13, -2) {$f(j)$};
\node [aa] (c) at (13, -4) {$j$};
\node [aa, dashed, fill=white] (dl) at (11.5, -6) {$2j-1$};
\node [aa, dashed, fill=white] (dr) at (14.5, -6) {$2j$};
\node [aa] (el) at (11.5, -8) {$g(2j-1)$};
\node [aa] (er) at (14.5, -8) {$g(2j)$};

\path [dashed, ->]
    (c) edge (b)
    (b) edge (a)
    (el) edge (dl)
    (dl) edge (c)
    (er) edge (dr)
    (dr) edge (c);

\node (L1) [text label, anchor = center] at (9, 0) {level $P(i) + n$};
\node (L2) [text label, anchor = center] at (9, -4) {level $i+n$};
\node (L3) [text label, anchor = center] at (9, -8) {level $S(i) + n$};

\path [->]
    (L2) edge node [text label, fill = white] {children} (L3)
    (L2) edge node [text label, fill = white] {parent} (L1);

\node (R1) [text label, anchor = center] at (17, 0) {$[N]$};
\node (R2) [text label, anchor = center] at (17, -2) {$[2N]$};
\node (R3) [text label, anchor = center] at (17, -4) {$[N]$};
\node (R4) [text label, anchor = center] at (17, -6) {$[2N]$};
\node (R5) [text label, anchor = center] at (17, -8) {$[N]$};

\path [->]
    (R2) edge node [text label, align = center, anchor = west, right = 0.1cm, fill = white] {$x\mapsto \lceil x/2\rceil$} (R1)
    (R3) edge [bend left = 30] node [text label, anchor = east] {$g$} (R2)
    (R2) edge [bend left = 30] node [text label, anchor = west] {$f$} (R3)
    (R3) edge [bend right = 30] node [text label, align = center, anchor = east, left = 0.2cm, fill = white] {$x\mapsto$\\ $2x-1$} (R4)
    (R3) edge [bend left = 30] node [text label, align = center, anchor = west, right = 0.2cm, fill = white] {$x\mapsto 2x$} (R4)
    (R4) edge [bend left = 30] node [text label, anchor = west] {$g$} (R5)
    (R5) edge [bend left = 30] node [text label, anchor = east] {$f$} (R4);

\end{tikzpicture}
            }
            \label{fig:Lossy-SoPL-completeness}
            \caption{The ``tree'' constructed in \autoref{thm: rwPHP-cap-EOPL to empty-child}. \textbf{\textit{Left}}: The first $n+1$ levels form a complete binary tree and the last $M$ levels are computed from the $\lossycode$ instance $(f, g)$ and the $\SoML$ instance $(S, P, V)$. \textbf{\textit{Right}}: For each node $(i+n, j)$ in the last $M$ levels, the levels of its parent and children are computed from $(S, P)$, and the positions of these nodes within their levels are computed from $(f, g)$.}
        \end{figure}
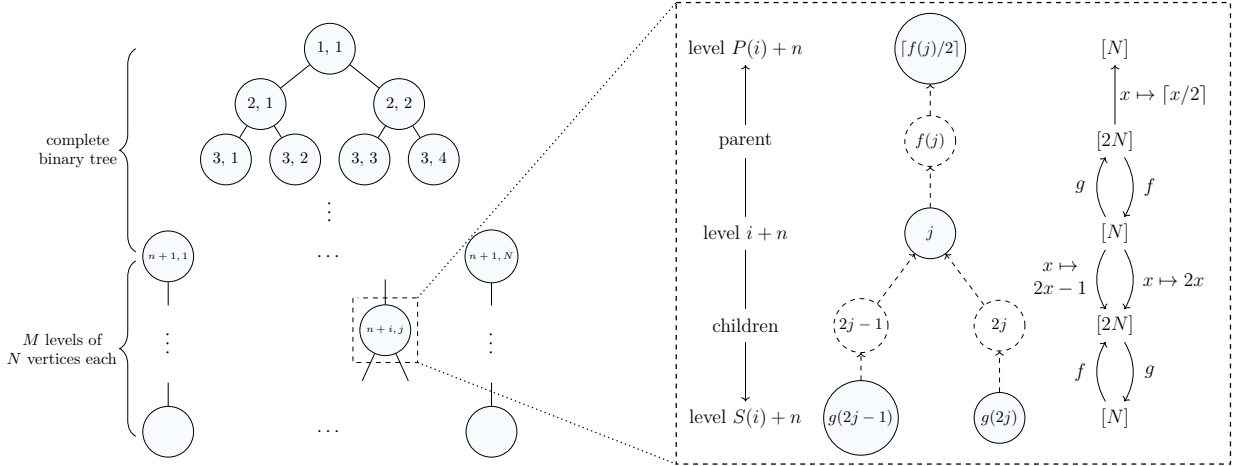

        Now, given any solution $(i', j)$ of $\ECwH$, we show how to obtain a solution of either $(f, g)$ (for $\lossycode$) or $(S, P, V)$ (for $\SoML$). In fact, from the above definitions of $F, L, R$ one can see that it must be the case that $i' > n$; moreover, if $(i', j)$ has \ref{item: (Wrong Height)}, then it must be the case that $i' > n + 1$ (unless $V(1) = 1$ and $1$ is a \ref{item: (Bad Source)} for $\SoML$). Now we define $i \coloneqq i' - n$, hence $i \in [M]$. We can see that 
        \begin{align*}
            F(L(i', j)) =&\, (P(S(i)) + n, \lceil f(g(2j-1))/2\rceil);&(\text{if }S(i)\ne 1)\\
            F(R(i', j)) =&\, (P(S(i)) + n, \lceil f(g(2j))/2\rceil);&(\text{if }S(i)\ne 1)\\
            H(i', j) - H(F(i', j)) =&\, V(i) - V(P(i)). &(\text{if }i > 1)
        \end{align*}
        
        Consider the following case analysis.\begin{itemize}
            \item Suppose that $(i', j)$ has an \ref{item: (Empty Child)}. Clearly, $V(i) \ne 0$. If $F(L(i', j)) \ne (i', j)$, then either $S(i) = 1$, or $P(S(i)) \ne i$, or $\lceil f(g(2j-1))/2\rceil \ne j$. If $F(R(i', j)) \ne (i', j)$, then either $S(i) = 1$, or $P(S(i)) \ne i$, or $\lceil f(g(2j))/2\rceil \ne j$. If $L(i', j) = R(i', j) \ne (i', j)$, then $g(2j-1) = g(2j)$. Hence, one of the following holds:\begin{itemize}
                \item $S(i) = 1$. If $P(1) \ne 1$ or $S(1) = 1$, then $1$ is a \ref{item: (Bad Source)} for $\SoML$; otherwise, since $P(S(i)) = 1 \ne i$, $i$ is a \ref{item: (Sink of Line)} for $\SoML$.
                \item $P(S(i)) \ne i$. In this case, $i$ is a \ref{item: (Sink of Line)} for $\SoML$.
                \item $\lceil f(g(2j-1))/2\rceil \ne j$. In this case, $f(g(2j-1)) \ne 2j-1$, hence $2j-1$ is an answer for $\lossycode$.
                \item $\lceil f(g(2j))/2\rceil \ne j$. In this case, $f(g(2j)) \ne 2j$, hence $2j$ is an answer for $\lossycode$.
                \item $g(2j-1) = g(2j)$. In this case, either $f(g(2j-1)) \ne 2j-1$ or $f(g(2j)) \ne 2j$, hence we can find an answer for $\lossycode$.
            \end{itemize}
            \item Clearly $(1,1)$ is not a \ref{item: (Wrong Root)}, so we do not need to care about this case.
            \item Suppose that $(i', j)$ has \ref{item: (Wrong Height)}. Then $V(i) - V(P(i)) \ne 1$ and $F(i',j)\ne (i',j)$, which implies that $V(i)\ge 1$.
            If $V(i)>1$, then $i$ is a \ref{item: (Bad Meter II)} for $\SoML$; otherwise $i$ is a \ref{item: (Bad Meter I)} for $\SoML$.%
        \end{itemize}
        It follows that given a solution for relaxed $\ECwH$, we could find a solution for either $\lossycode$ or $\SoML$.
    \end{proof}

    Since $\ECwH$ is hard for $\lossycode\cap \SoML$ (\autoref{thm: rwPHP-cap-EOPL to empty-child}), is in $\SOPL$ (\autoref{lemma: ECwH in SOPL}), and is in $\Lossy$ (\autoref{lemma: emptychild in lossy}), we have:
    \begin{corollary}
        $\ECwH$ is complete for $\SOPL\cap \Lossy$.
    \end{corollary}

    Similarly, we can also reduce $\lossycodep \cap \EoML$ to \BECwH. Recall that $\EoML$ is defined as follows:

    \begin{description}
        \item[$\EoML$.]  The input is functions $S, P: [N]\to [N]$ and $V: [N]\to [N]\cup\{0\}$ (i.e., same as $\SoML$). Besides the solutions of $\SoML$, we also accept the following solutions:
        \begin{description}
            \item [s5.] A vertex $x\in [N]\setminus \{1\}$ such that $S(P(x)) \ne x$.\mylabel{End of Line}{item: (End of Line)}
            \item [s6.] A vertex $x\in [N]$ such that ($V(x) > 0$ and $V(S(x)) - V(x) \ne 1$). \mylabel{Bad Meter III}{item: (Bad Meter III)}
        \end{description}
    \end{description}

    \begin{corollary}\label{cor: InjLC cap EOML to BECwH}
        $\lossycodep \cap \EoML$ can be reduced to \BECwH.
    \end{corollary}
    \begin{proof}[Proof Sketch]
        The reduction is exactly the same, except that we also need to handle the case that $(i', j)$ has a \ref{item: (Wrong Father)}. We can also see that $i' > n+1$, hence $i = i'-n \ge 2$. We also have that $V(i) \ne 0$. Hence,
        \begin{align*}
            L(F(i', j)) =&\, (S(P(i)) + n, g(2\lceil f(j) / 2\rceil - 1)); &(\text{since }i > 1)\\
            R(F(i', j)) =&\, (S(P(i)) + n, g(2\lceil f(j) / 2\rceil)). &(\text{since }i > 1)
        \end{align*}
        Hence, the case when $(i', j)$ has a \ref{item: (Wrong Father)} is argued as follows.
        \begin{itemize}
            \item Suppose that $(i', j)$ has a \ref{item: (Wrong Father)}. Letting $v \coloneqq 2\lceil f(j)/2\rceil$, then
            \[(i', j) \notin \{(S(P(i)) + n, g(v - 1)), (S(P(i)) + n, g(v))\}.\]
            Hence, one of the following holds: \begin{itemize}
                \item $i \ne S(P(i))$. In this case, $i$ is an \ref{item: (End of Line)} for $\EoML$.
                \item $j\notin \{g(v-1), g(v)\}$. However $f(j) \in \{v-1, v\}$, hence $g(f(j)) \ne j$ and $j$ is an answer for $\lossycodep$.\qedhere
            \end{itemize}
        \end{itemize}
    \end{proof}

	\subsection{\texorpdfstring{$\Lossy$}{Lossy}-Completeness of \texorpdfstring{$\emptychild$}{Empty-Child}} \label{sec: lossycode to empty-child}
    
	Next, we adapt the above proof to show that $\lossycode$ reduces to $\emptychild$, establishing the equivalence between the two problems. Roughly speaking, our reduction makes use of two ideas:
    \begin{enumerate}
        \item $\lossycode$ reduces to $\PPADS$. In fact, a $\lossycode$ instance $f: [N]\to [2N]$ and $g: [2N] \to [N]$ can be seen as a $\textsc{Sink-Of-Line}$ instance where $g$ is the successor function, $f$ is the predecessor function, and there are $N$ distinguished sources numbered from $N+1$ to $2N$.
        \item One can reduce $\lossycode \cap \PPADS$ to $\emptychild$ by slightly adapting the reduction in \autoref{sec: empty child with height}, using the $\PPADS$ instance to handle the ``heights''.
    \end{enumerate}

    Our reduction starts by constructing a forest with $2N$ levels, where each level contains $N$ vertices. The children and parents of each vertex can then be defined in the same way as those in the bottom $M$ levels of the tree from \Cref{thm: rwPHP-cap-EOPL to empty-child}, except that a vertex at the $i$-th level has its father at the $f(i)$-th level and children at the $g(i)$-th level. This is already a ``forest'' where every vertex in levels $N+1\sim 2N$ are distinguished roots (hence there are $N^2$ distinguished roots). %
    Finally, we create a complete binary tree of depth $\log(N^2)$ at the top to connect them.%
    
    \begin{theorem}\label{thm: lossycode to emptychild}
        There is a reduction from $\lossycode$ to $\emptychild$.
    \end{theorem}

    \begin{proof}
        Given $f:[N]\to [2N]$ and $g:[2N]\to [N]$ as the inputs to $\lossycode$, we construct an instance $(F,L,R)$ of $\emptychild$ as follows:
        As in the proof of~\Cref{thm: rwPHP-cap-EOPL to empty-child}, we may assume without loss of generality that $N=2^n$ is a power of $2$.
        The ``binary tree'' has $2(n+N)$ levels: the first $2n$ levels form a perfect binary tree, while each of the remaining $2N$ levels contains exactly $N$ vertices.
        Following the notation in the proof of~\Cref{thm: rwPHP-cap-EOPL to empty-child}, we write $(i,j)$ for the $j$-th vertex in the $i$-th level.
        Slightly abusing notation, we also use $(i,j)$ to denote the index of this vertex, namely
        \[
        (i, j) = \begin{cases}
            j+2^{i-1}-1 & \text{if }i \le 2n,\\
            j+N\cdot (i-2n-1)+N^2-1 & \text{otherwise}.
        \end{cases}
        \]

        The structure of the tree is described below:
        \begin{itemize}
            \item The father of the root is $F(1,1)\coloneqq (1,1)$.
            
            \item For every $2\le i\le 2n$ and $j\in [2^{i-1}]$, $F(i, j) \coloneqq (i-1, \lceil j/2\rceil)$.

            \item For every $1\le i\le N$ and $j\in [N]$, $F(i+2n+N,j)\coloneqq (2n, \lceil ((i-1)\cdot N+j)/2\rceil)$.

            \item For every $1\le i\le N$ and $j\in [N]$, $F(i + 2n, j) \coloneqq (f(i) + 2n, \lceil f(j) / 2\rceil)$.

            \item For every $i\in [2n-1]$ and $j\in [2^{i-1}]$, the left and right children of $(i, j)$ are $L(i, j) \coloneqq (i+1, 2j-1)$ and $R(i, j) \coloneqq (i+1, 2j)$ respectively.

            \item For every $1\le i\le N$ and $j\in [N/2]$, $L(2n,(i-1)N/2+j)\coloneqq (i+2n+N,2j-1)$ and $R(2n,(i-1)N/2+j)\coloneqq (i+2n+N,2j)$.

            \item For every $i\in [2N]$ and $j\in [N]$, $L(i+2n,j)\coloneqq (g(i)+2n,g(2j-1))$ and $R(i+2n,j)\coloneqq (g(i)+2n,g(2j))$.

        \end{itemize}

        \begin{figure}
            \centering
            \resizebox{\linewidth}{!}{
                \begin{tikzpicture}[
    aa/.style={draw, circle, inner sep=2pt, font=\small, minimum size = 1.1cm, fill=Cerulean!3!white},
    text label/.style={draw=none, font=\large, inner sep=0pt, minimum size = 0.7cm},
    edge from parent/.style={draw,-latex}
]

\node[aa] {1, 1} [sibling distance = 3cm, level distance = 1.2cm]
    child {node[aa] {2, 1}[sibling distance = 1.5cm]
        child {node[aa] {3, 1}}
        child {node[aa] {3, 2}}}
    child {node[aa] {2, 2}[sibling distance = 1.5cm]
        child {node[aa] {3, 3}}
        child {node[aa] {3, 4}}};
\node[text label, rotate=90, anchor=center] at (-0.826, -3.4) {$\cdots$};
\node[aa] (LNode) at (-3.5, -4.5) {\tiny $2n, 1$};
\node[aa] (MNode) at (-0.826, -4.5) {};
\node[aa] (RNode) at ( 3.5, -4.5) {\tiny $2n, N^2/2$};

\node[aa] (Y0) at (-2, -6) {\tiny $2n+1, 1$};
\node[aa] (Z0) at (2, -6) {\tiny $2n+1, N$};

\node[aa] (IJ) at (0.4, -7.6) {\tiny $2n+i, j$};
\draw (IJ) -- (0.4, -6.8);
\draw (IJ) -- (-0.1, -8.4);
\draw (IJ) -- (0.9, -8.4);

\draw[dashed] (-0.3, -6.9) -- (-0.3, -8.3) -- (1.1, -8.3) -- (1.1, -6.9) -- (-0.3, -6.9);

\node[text label, rotate=90, anchor=center] at (-2, -7.25) {$\cdots$};
\node[text label, rotate=90, anchor=center] at ( 2, -7.25) {$\cdots$};
\node[text label, rotate=90, anchor=center] at (-2, -11.6) {$\cdots$};
\node[text label, rotate=90, anchor=center] at ( 2, -11.6) {$\cdots$};

\node[aa] (Y1) at (-2, -8.6) {};
\node[aa] (Z1) at (2, -8.6) {};
\node[aa] (Y2) at (-2, -10.1) {};
\node[aa] (Z2) at (2, -10.1) {};
\node[aa] (Y3) at (-2, -13.1) {};
\node[aa] (Z3) at (2, -13.1) {};

\node[text label, anchor=center] at (-2.163, -4.5) {$\cdots$};
\node[text label, anchor=center] at (1.337, -4.5) {$\cdots$};
\node[text label, anchor=center] at (0, -6) {$\cdots$};
\node[text label, anchor=center] at (0, -8.6) {$\cdots$};
\node[text label, anchor=center] at (0, -10.1) {$\cdots$};
\node[text label, anchor=center] at (0, -13.1) {$\cdots$};

\path [->]
    (LNode) edge [bend right = 65] (Y2)
    (LNode) edge [bend right = 85] (-2.3, -11.3)
    (RNode) edge [bend left = 25] (2.3, -11.8)
    (RNode) edge [bend left = 35] (Z3)
    (MNode) edge (-0.9, -10.7)
    (MNode) edge (-0.4, -10.7);

\draw[decoration={brace, amplitude=10pt, mirror}, decorate]
    (-4.2, 0) -- node[text label, left=0.4cm, align=center, anchor=east, font=\small] {complete\\ binary tree} (-4.2, -4.5);
\draw[decoration={brace, amplitude=10pt, mirror}, decorate]
    (-2.7, -6) -- node[text label, left=0.4cm, align=center, anchor=east, font=\small] {$N$ levels} (-2.7, -8.7);
\draw[decoration={brace, amplitude=10pt, mirror}, decorate]
    (-2.7, -10) -- node[text label, left=0.4cm, align=center, anchor=east, font=\small] {$N$ levels} (-2.7, -13.1);

\fill[fill=gray, fill opacity = 0.1, text opacity = 1] (-2.9, -9.35) rectangle (2.9, -13.75) node[pos=.5, color=black, align = center] {distinguished\\roots};

\draw[thick, dashed] (7.3, -1) -- (19.5, -1) -- (19.5, -11) -- (7.3, -11) -- (7.3, -1);
\draw[thick, dotted] (1.1, -6.9) -- (7.3, -1);
\draw[thick, dotted] (1.1, -8.3) -- (7.3, -11);

\node [aa] (a) at (13, -2) {$\lceil f(j) / 2\rceil$};
\node [aa, dashed, fill=white] (b) at (13, -4) {$f(j)$};
\node [aa] (c) at (13, -6) {$j$};
\node [aa, dashed, fill=white] (dl) at (11.5, -8) {$2j-1$};
\node [aa, dashed, fill=white] (dr) at (14.5, -8) {$2j$};
\node [aa] (el) at (11.5, -10) {$g(2j-1)$};
\node [aa] (er) at (14.5, -10) {$g(2j)$};

\path [dashed, ->]
    (c) edge (b)
    (b) edge (a)
    (el) edge (dl)
    (dl) edge (c)
    (er) edge (dr)
    (dr) edge (c);

\node (L1) [text label, anchor = center] at (9, -2) {level $f(i) + 2n$};
\node (L2) [text label, anchor = center] at (9, -6) {level $i+2n$};
\node (L3) [text label, anchor = center] at (9, -10) {level $g(i) + 2n$};

\path [->]
    (L2) edge node [text label, fill = white] {children} (L3)
    (L2) edge node [text label, fill = white] {parent} (L1);

\node (R1) [text label, anchor = center] at (17, -2) {$[N]$};
\node (R2) [text label, anchor = center] at (17, -4) {$[2N]$};
\node (R3) [text label, anchor = center] at (17, -6) {$[N]$};
\node (R4) [text label, anchor = center] at (17, -8) {$[2N]$};
\node (R5) [text label, anchor = center] at (17, -10) {$[N]$};

\path [->]
    (R2) edge node [text label, align = center, anchor = west, right = 0.1cm, fill = white] {$x\mapsto \lceil x/2\rceil$} (R1)
    (R3) edge [bend left = 30] node [text label, anchor = east] {$g$} (R2)
    (R2) edge [bend left = 30] node [text label, anchor = west] {$f$} (R3)
    (R3) edge [bend right = 30] node [text label, align = center, anchor = east, left = 0.2cm, fill = white] {$x\mapsto$\\ $2x-1$} (R4)
    (R3) edge [bend left = 30] node [text label, align = center, anchor = west, right = 0.2cm, fill = white] {$x\mapsto 2x$} (R4)
    (R4) edge [bend left = 30] node [text label, anchor = west] {$g$} (R5)
    (R5) edge [bend left = 30] node [text label, anchor = east] {$f$} (R4);

\end{tikzpicture}
            }
            \label{fig:Lossy-completeness}
            \caption{The ``tree'' constructed in \autoref{thm: lossycode to emptychild} where solid arrows represent tree edges. The first $2n$ levels along with the last $N$ levels form the perfect binary tree. Every vertex in the last $N$ levels are distinguished roots of the bottom forest, and hence are the leaves of the perfect binary tree. The last $2N$ levels form the bottom forest, whose structure is computed from the $\lossycode$ instance $(f, g)$.}
        \end{figure}
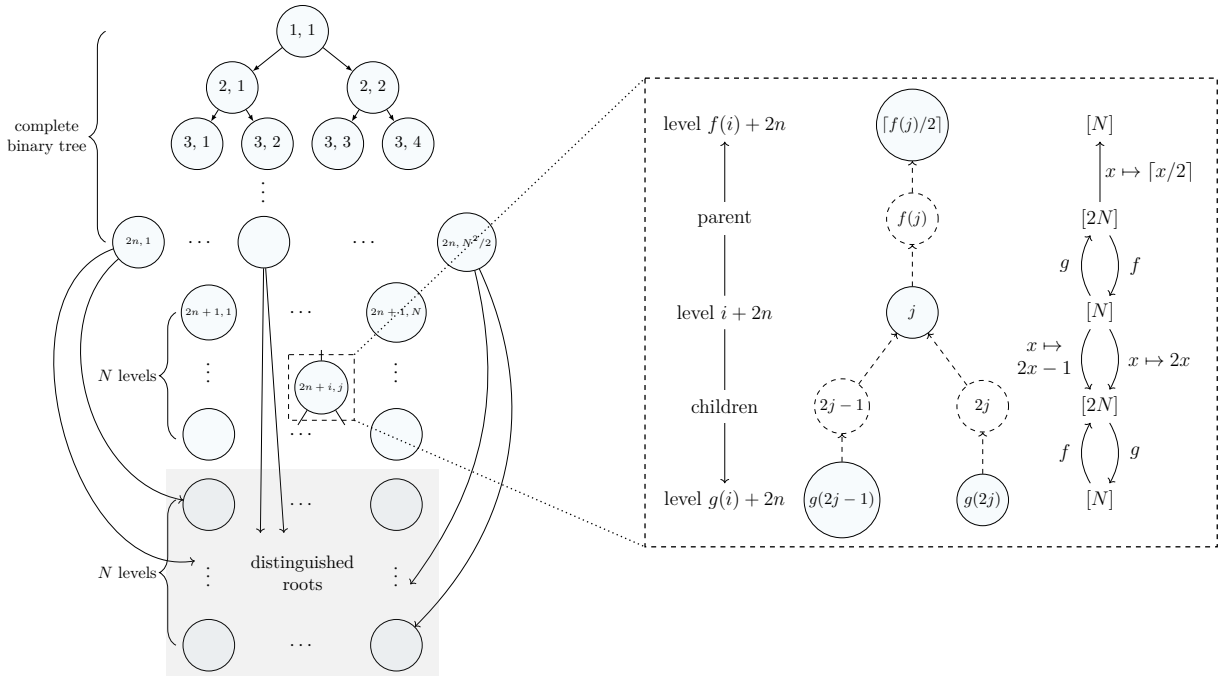

         Given any solution $(i', j)$ of $\emptychild$, we show how to obtain a solution of $\lossycode$. Note that our $\emptychild$ instance only has solutions with \ref{item: (Empty Child)}: \ref{item: (Wrong Father)} and \ref{item: (Wrong Height)} are not allowed, and \ref{item: (Wrong Root)} does not apply since our tree has a ``correct'' root. Hence, we have that $F(L(i', j)) \ne (i', j)$, or $F(R(i', j)) \ne (i', j)$, or $L(i', j) = R(i', j) \ne (i', j)$. It follows from the structure of our tree that $i' > 2n$; define $i \coloneqq i' - 2n$. We have
         \begin{align*}
            F(L(i', j)) =&\, F(g(i) + 2n, g(2j-1)) = (f(g(i)) + 2n, \lceil f(g(2j-1)) / 2\rceil);\\
            F(R(i', j)) =&\, F(g(i) + 2n, g(2j))   = (f(g(i)) + 2n, \lceil f(g(2j)) / 2\rceil).
        \end{align*}
        There are three cases:\begin{itemize}
            \item {\bf Case I:} $F(L(i', j)) \ne (i', j)$. Then either $f(g(i)) \ne i$, or $\lceil f(g(2j-1)) / 2\rceil \ne j$. It follows that either $i$ or $2j-1$ is a valid solution for $\lossycode$.
            \item {\bf Case II:} $F(R(i', j)) \ne (i', j)$. Then either $f(g(i)) \ne i$, or $\lceil f(g(2j)) / 2\rceil \ne j$. It follows that either $i$ or $2j$ is a valid solution for $\lossycode$.
            \item {\bf Case III:} $L(i', j) = R(i', j)$. Then $g(2j-1) = g(2j)$, which means that either $2j$ or $2j-1$ is a valid solution for $\lossycode$.\qedhere
        \end{itemize} 
        
    \end{proof}

    Similarly, $\lossycodep$ reduces to $\BEC$:
	\begin{corollary}\label{cor: InjLC to BEC}
		There is a reduction from $\lossycodep$ to $\BEC$.
	\end{corollary}

        \begin{proof}[Proof Sketch]
            The reduction is exactly the same as in \autoref{thm: lossycode to emptychild}, except that we also need to handle the case that $(i', j)$ has a \ref{item: (Wrong Father)}. Clearly, $i = i' - 2n \in [N]$. Letting $v := 2\lceil f(j)/2\rceil$, we have
        \begin{align*}
            L(F(i', j)) =&\, L(f(i)+2n, \lceil f(j) / 2\rceil) = (g(f(i))+2n, g(v-1));\\
            R(F(i', j)) =&\, R(f(i)+2n, \lceil f(j) / 2\rceil) = (g(f(i))+2n, g(v)).
        \end{align*}
        Since $(i', j)\not\in\{L(F(i', j)), R(F(i', j))\}$, there are two cases:
        \begin{itemize}
            \item {\bf Case I:} $i'\ne g(f(i))+2n$. This implies that $g(f(i)) \ne i$, hence $i$ is a valid solution for $\lossycodep$.
            \item {\bf Case II:} $j\not\in\{g(v-1), g(v)\}$. However $f(j) \in \{v-1, v\}$, hence $g(f(j))\ne j$ and $j$ is a valid solution for $\lossycodep$.\qedhere
        \end{itemize}
        \end{proof}

    It is unclear whether $\BEC$ is equivalent to $\lossycodep$. We are unable to reduce the former to the latter, and the proof strategy of \autoref{lemma: emptychild in lossy} does not work here. We conjecture that (in the black-box setting) $\BEC$ is strictly harder than $\lossycodep$.

\subsection{\texorpdfstring{\emptychild}{Empty-Child} Reduces to \texorpdfstring{\ILP}{Nephew}} \label{sec: EC: REC}

Recall once more the definition of $\ILP$:

\begin{description}
    \Nephewitem 
\end{description}

$\ILP$ is at least as powerful as $\emptychild$. We show a reduction from an $\emptychild$ instance $(V, F, L, R)$ to a $\ILP$ instance $(V', f, g)$, where $V' = V \times \{0, 1\}$. If there is an $\emptychild$ solution at $v \in V$, then both of $(v, 0), (v, 1)$ will be self-loops for both $f$ and $g$, giving a solution in the $\ILP$ instance. Otherwise, we will construct a valid $\ILP$ substructure. Special care will need to be taken in the case of isolated vertices in the $\emptychild$ instance, as well as for the distinguished root vertex $1 \in V$.

The procedure $\emptychildreduction$ in \cref{alg: emptychild reduction} is the reduction.  It takes in a vertex $(v, i) \in V'$ and returns a tuple containing $f(v, i)$ and $g(v, i)$. See \cref{fig: emptychild to nephew} for an illustration of the $\ILP$ instance returned by this procedure.

Before we prove its validity, we provide some intuition for the reduction. In the $\ILP$ instance, we construct a structure where, for each vertex $v \in V$, the vertices corresponding to $F(v)$ are pointed to via the $f$ function and the vertices corresponding to $L(v), R(v)$ are pointed to via the $g$ function. In order to do this, we need to double the number of vertices: $(v, 0)$ will point via $g$ to $R(v)$ and be pointed at via $f$ by $L(v)$; and $(v, 1)$ will point via $g$ to $L(v)$ and be pointed at via $f$ by $R(v)$. Isolated vertices and the root vertex will be handled separately.

\begin{algorithm}[h!] 
        \caption{Procedure $\emptychildreduction_{F, L, R}(v, i)$ } \label{alg: emptychild reduction}
        \begin{algorithmic}[1]
            \If{$F(L(v)) \neq v \lor F(R(v)) \neq v \lor L(v) = R(v) \neq v$}
                \State \Return $((v, i), (v, i))$ \Comment{\ref{item: (Empty Child)}, so self-loop}
            \ElsIf{$v = 1$ and $L(1) = 1 \lor R(1) = 1 \lor F(1) \neq 1$}
                \State \Return $((1, i), (1, i))$ \Comment{\ref{item: (Wrong Root)}, so self-loop}
            \ElsIf{$F(v) = L(v) = R(v) = v$} 
                \State \Return $((1, 0), (v, 1-i))$ \Comment{Isolated vertex, treat specially}
            \Else  
                \If{$v=1$}
                    \State $f \gets (1, 0)$
                \ElsIf{$v = R(F(v))$}
                    \State $f \gets (F(v), 1)$
                \Else \Comment{$v\ne 1$ and $v$ is not a right child}
                    \State $f \gets (F(v), 0)$
                \EndIf
                \If{$(v, i) = (1, 0)$}
                    \State $g \gets (L(L(v)), 0)$
                \ElsIf{$(v, i) = (1, 1)$}
                    \State $g \gets (1, 1)$
                \ElsIf{$i = 0$}
                    \State $g \gets (R(v), 0)$
                \Else \Comment{$v \neq 1$ and $i = 1$}
                    \State $g \gets (L(v), 0)$
                \EndIf
                \State \Return $(f, g)$
            \EndIf
        \end{algorithmic}
    \end{algorithm}

\begin{figure}
  \centering
  \begin{subfigure}[b]{0.45\textwidth}
    \centering
    \begin{tikzpicture}[level2/.style={fill=Cerulean!40!white, text opacity=1},level0/.style={fill=Cerulean!2!white}]

        \node[v,level0] (v0) at (0, 0) {$v, 0$};
        \node[v,level0] (v1) at (2, 0) {$v, 1$};
        \node[v,level0] (Lv) at (0, 2) {\footnotesize $Lv, 0$};
        \node[v,level0] (Rv) at (2, 2) {\footnotesize $Rv, 0$};
        \node[v,level0] (Fv) at (0, -2) {\footnotesize $Fv, i$};

        \draw[farrow] (Lv) to (v0);
        \draw[farrow] (Rv) to (v1);
        \draw[farrow] (v0) to (Fv);
        \draw[farrow] (v1) to (Fv);

        \draw[garrow] (v0) to (Rv);
        \draw[garrow] (v1) to (Lv);
    \end{tikzpicture}
    \caption{Local $\ILP$ structure for a typical node $v$. The value of $i$ depends on whether or not $v$ is a right child.}
    \label{fig: level 0}
  \end{subfigure}
  \quad \quad 
  \begin{subfigure}[b]{0.45\textwidth}
    \centering
    \begin{tikzpicture}[level2/.style={fill=Cerulean!40!white, text opacity=1},level0/.style={fill=Cerulean!2!white}]
        \node[v,level0] (vx0) at (0, 0) {};
        \node[v,level0] (vx1) at (2, 0) {};
        \node[v,level0] (Lv) at (0, 2) {\scriptsize $LL1, 0$};
        \node[v,level0] (Rv) at (2, 2) {};
        \node[v,level0] (one) at (0, -2) {$1, 0$};
        \node[v,level0] (oneone) at (2, -2) {$1, 1$};
        \node[v,level0] (v0) at (-2, 0) {$v, 0$};
        \node[v,level0] (v1) at (-2, -2) {$v, 1$};

        \draw[farrow] (Lv) to (vx0);
        \draw[farrow] (Rv) to (vx1);
        \draw[farrow] (vx0) to (one);
        \draw[farrow] (vx1) to (one);
        \draw[farrow] (v0) to (one);
        \draw[farrow] (v1) to (one);
        \draw[farrow] (one) to[loop right] (one);
        \draw[farrow] (oneone) to [bend right] (one);

        \draw[garrow] (vx0) to (Rv);
        \draw[garrow] (vx1) to (Lv);
        \draw[garrow] (one) to[bend left] (Lv);
        \draw[garrow] (v0) to[bend left] (v1);
        \draw[garrow] (v1) to[bend left] (v0);
        \draw[garrow] (oneone) to[loop above] (oneone);
    \end{tikzpicture}
    \caption{Local $\ILP$ structure around node 1, with a node $v$ that is a self-loop in $\emptychild$.} \label{fig: emptychild to nephew root}
    \label{fig: level 1}
  \end{subfigure}
  \caption{Illustration of $\emptychildreduction_{F, L, R}$. Solid arrows represent $f$ and dashed arrows represent $g$. Parentheses are omitted in labels.}
  \label{fig: emptychild to nephew}
\end{figure}
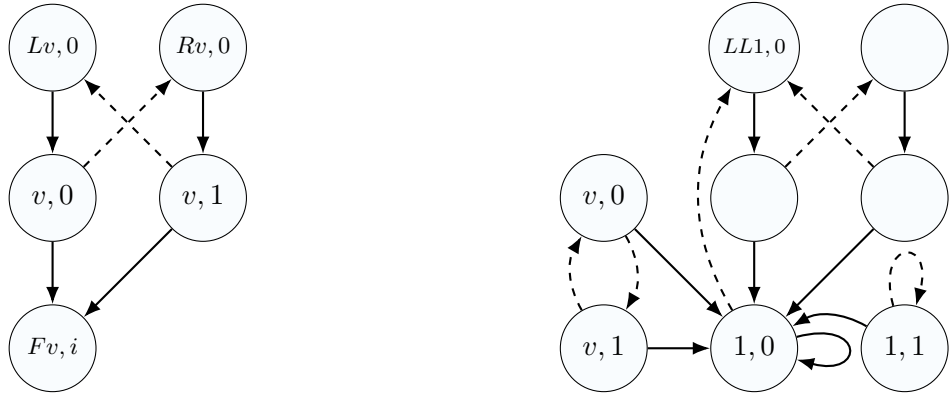

\begin{theorem} \label{thm: emptychild to nephew}
    $(f(v, i), g(v, i)) \gets \emptychildreduction_{F, L, R}$ is a reduction from $\emptychild$ to $\ILP$.
\end{theorem}

\begin{proof}
    Consider any $(v, i) \in V'$ such that $v$ is a solution to $\emptychild$. Then there is a self-loop on $(v, i)$ with both $f$ and $g$, meaning that $f(g(v, i)) = (v, i)$, and so $(v, i)$ is a solution to the $\ILP$ instance. Given such an $\ILP$ solution, then, it is easy to find a solution to $\emptychild$ by checking $v$. Additionally, it is possible for there to be a $\ILP$ solution if an $\emptychild$ solution exists at $L(v)$ or $R(v)$, or at $v = 1$ if a solution exists at $L(L(1))$. Again, it is easy to check if any of these is the case.

    We will show that those scenarios are the only solutions to the $\ILP$ instance. Consider any $(v, i) \in V'$ and assume there is no solution at $v$, $L(v)$, $R(v)$, or $L(L(1))$. Then we have the following cases:

    \begin{itemize}
        \item $F(v) = L(v) = R(v) = v$. Note then $v \neq 1$. 
        
        There is no solution of type 1: \[f(f(g(v, i))) = f(f(v,1-i)) = f(1,0) = (1, 0) = f(v, i),\] where the second-to-last equality comes from the fact that either $1$ is not a solution, so $f(1, 0) = (1, 0)$, or $1$ is a solution and $(1, 0)$ self-loops. 
        
        There is no solution of type 2: \[f(g(v, i)) = f(v, 1-i) = (1, 0) \neq (v, i).\]

        For the remaining cases, assume $F(v) = L(v) = R(v) = v$ does not hold.
        \item $(v, i) = (1, 0)$. 

        There is no solution of type 1: \begin{align*}
            f(f(g(1, 0))) &= f(f(L(L(1)), 0)) = f( F(L(L(1))), 0) = f( L(1), 0) = (F(L(1)), 0) = (1, 0) \\&= f(1, 0),
        \end{align*}
        where we used the fact that $L(1)$ and $L(L(1))$ are not solutions.

        There is no solution of type 2: \[f(g(1, 0)) = f(L(L(1)), 0) = (F(L(L(1))), 0) = (L(1), 0) \neq (1, 0). \]

        \item $(v, i) = (1, 1)$.

        There is no solution of type 1: \[f(f(g(1, 1))) = f(f(1, 1)) = f(1, 0) = (1, 0) = f(1, 1).\]

        There is no solution of type 2: \[f(g(1, 1)) = f(1, 1) = (1, 0) \neq (1, 1).\]

        \item $v \neq 1$ and $i = 0$. Observe that $R(F(R(v))) = R(v)$, as $F(R(v)) = v$.

        There is no solution of type 1: \[f(f(g(v, i))) = f(f(R(v), 0)) = f(F(R(v)), 1) = f(v, 1) = f(v, 0), \] where the last inequality comes from the fact that $(v, 0)$ and $(v, 1)$ will always map to the same value under $f$.

        There is no solution of type 2: \[f(g(v)) = f(R(v), 0) = (F(R(v)), 1) = (v, 1) \neq (v, 0).\]

        \item $v \neq 1$ and $i = 1$.

        There is no solution of type 1: \[f(f(g(v, i))) = f(f(L(v), 0)) = f(F(L(v)), 0) = f(v, 0) = f(v, 1), \] where the last inequality comes from the fact that $(v, 0)$ and $(v, 1)$ will always map to the same value under $f$.

        There is no solution of type 2: \[f(g(v, i)) = f(L(v), 0) = (F(L(v)), 0) = (v, 0) \neq (v, 1).\]
    \end{itemize}

    We have ruled out $\ILP$ solutions other than those listed before the case analysis. Thus, every solution in the $\ILP$ instance can be used to efficiently map back to an $\emptychild$ solution as required.
\end{proof}

Combining \cref{thm: emptychild to nephew} with \cref{thm: lossycode to emptychild}, we obtain the following.

\begin{corollary}
    There is a reduction from $\lossycode$ to $\ILP$.
\end{corollary}

\subsection{\texorpdfstring{\emptychild}{Empty-Child} and \texorpdfstring{\ILP}{Nephew} with Inverse} \label{sec: nephew with inverse}

It seems to be difficult to reduce $\ILP$ to $\emptychild$. It is possible that $\ILP$ is strictly more powerful than $\emptychild$. However, a proof of this seems elusive. To demonstrate this, we consider a natural modification to $\ILP$ by adding an \emph{inverse} to $f$. Interestingly, we show that this results in a problem that is \emph{equivalent} to $\emptychild$. If we believe that this modification is superficial, we could take this as evidence that $\ILP$ and $\emptychild$ are in fact equivalent. Alternatively, it could indicate that any proof that shows that $\ILP$ does not reduce to $\emptychild$ needs to argue how an inverse to $f$ makes $\ILP$ significantly easier.

We need to be a little bit careful in our definition of the inverse function. In the previous subsection we reduced $\emptychild$ to $\ILP$ and created vertices that were not solutions yet were not pointed at via the $f$ function: for an illustration see \cref{fig: emptychild to nephew root}. This situation only arises in the case where $f(f(v)) = f(v)$. Therefore we allow $f^{-1}(v) = \bot$ in instances where $f(f(v)) = f(v)$.

\begin{description}
    \item[\ILPwithInverse.] Given the same inputs as in $\ILP$ and a function $f^{-1}: V \rightarrow V \cup \{\bot\}$, in addition to the solutions of $\ILP$, we accept the following solutions:
    \begin{description}
        \item[s3.] $v \in V$ such that $f^{-1}(v) \neq \bot$ and $f(f^{-1}(v)) \neq v$. \mylabel{Wrong Inverse}{item: Inverse}
        \item[s4.] $v \in V$ such that $f^{-1}(v) = \bot$ and $f(f(v)) \neq f(v)$. \mylabel{Bad $\bot$}{item: Bad Bot}
    \end{description}
\end{description}

First we show that $\emptychildreduction$ can be augmented to include the inverse function, and therefore $\emptychild$ reduces to $\ILPwithInverse$. This augmented procedure is shown in \cref{alg: emptychild reduction inverse}. It returns a tuple $(f, f^{-1}, g)$. New additions are in blue and underlined.

\newcommand{\newmark}[1]{\textcolor{MidnightBlue}{\underline{#1}}}

\begin{algorithm}[ht] 
    \caption{Procedure $\emptychildreductionInv_{F, L, R}(v, i)$ } \label{alg: emptychild reduction inverse}
    \begin{algorithmic}[1]
        \If{$F(L(v)) \neq v \lor F(R(v)) \neq v \lor L(v) = R(v) \neq v$}
            \State \Return $((v, i), \newmark{(v, i)}, (v, i))$ \Comment{\ref{item: (Empty Child)}, so self-loop}
        \ElsIf{$v = 1$ and $L(1) = 1 \lor R(1) = 1 \lor F(1) \neq 1$}
            \State \Return $((1, i), \newmark{(1, i)}, (1, i))$ \Comment{\ref{item: (Wrong Root)}, so self-loop}
        \ElsIf{$F(v) = L(v) = R(v) = v$} 
            \State \Return $((1, 0), \newmark{\bot}, (v, 1-i))$ \Comment{Isolated vertex, treat specially}
        \Else  
            \If{$v=1 $}
                \State $f \gets (1, 0)$
            \ElsIf{$v = R(F(v))$}
                \State $f \gets (F(v), 1)$
            \Else \Comment{$v\ne 1$ and $v$ is not a right child}
                \State $f \gets (F(v), 0)$
            \EndIf
            \If{$(v, i) = (1, 0)$}
                \State $g \gets (L(L(v)), 0)$
                \State \newmark{$f^{-1} \gets (L(v), 0)$}
            \ElsIf{$(v, i) = (1, 1)$}
                \State $g \gets (1, 1)$
                \State \newmark{$f^{-1} \gets \bot$}
            \ElsIf{$i = 0$}
                \State $g \gets (R(v), 0)$
                \State \newmark{$f^{-1} \gets (L(v), 0)$}
            \Else \Comment{$v \neq 1$ and $i = 1$}
                \State $g \gets (L(v), 0)$
                \State \newmark{$f^{-1} \gets (R(v), 0)$}
            \EndIf
            \State \Return $(f, \newmark{f^{-1}}, g)$
        \EndIf
    \end{algorithmic}
\end{algorithm}

\begin{theorem}
    $\emptychild$ reduces to $\ILPwithInverse$.
\end{theorem}

\begin{proof}
    \cref{thm: emptychild to nephew} shows that $\emptychildreduction$ is a reduction from $\emptychild$ to $\ILP$. The analysis in that proof still holds, as the $f$ and $g$ pointers returned are unmodified by the augmentation to $\emptychildreductionInv$. Thus, all that is left to do is to prove that no additional solutions are introduced by the addition of the $f^{-1}$ pointer.

    As in the proof of \cref{thm: emptychild to nephew}, consider $(v, i) \in V'$ and assume there is no solution at $v, L(v), R(v)$, or $L(L(1))$. We know already that there are no solutions of types 1 or 2. All that is left to prove is that no solutions of type 3 \ref{item: Inverse} or type 4 \ref{item: Bad Bot} exist. 

    \begin{itemize}
        \item $F(v) = L(v) = R(v) = v$. Note then $v \neq 1$. $f^{-1}(v) = \bot$, so we need only to consider solutions of type 4. None exist:
        \[ f(f(v, i)) = f((1, 0)) =  (1, 0) = f(v, i). \]

        \item $(v, i) = (1, 0)$. $f^{-1}(v) \neq \bot$, so we need only to consider solutions of type 3. None exist:
        \[ f(f^{-1}(1, 0)) = f(L(1), 0) = (F(L(1)), 0) = (1, 0). \]

        \item $(v, i) = (1, 1)$. $f^{-1}(v) = \bot$, so we need only to consider solutions of type 4. None exist:
        \[ f(f(1, 1)) = f(1, 0) = (1, 0) = f(1, 1). \]

        \item $v \neq 1$ and $i = 0$. $f^{-1}(v) \neq \bot$, so we need only to consider solutions of type 3. None exist:
        \[ f(f^{-1}(v, 0)) = f(L(v), 0) = (F(L(v)), 0) = (v, 0). \]

        \item $v \neq 1$ and $i = 1$. $f^{-1}(v) \neq \bot$, so we need only to consider solutions of type 3. None exist:
        \[ f(f^{-1}(v, 1)) = f(R(v), 0) = (F(R(v)), 1) = (v, 1). \qedhere \]
    \end{itemize}
\end{proof}

We will now show that $\ILPwithInverse$ reduces to $\emptychild$. To do so, we will reduce it to $\emptychild'$.

\begin{theorem}
    $\ILPwithInverse$ reduces to $\emptychild'$.
\end{theorem}

\begin{proof}
    We use a variant of the $\FindChildren$ procedure named $\FindChildrenAndParent$ that in addition also returns the parent vertex. Here, we define $\Checksol(u)$ to be the procedure that returns $\mathsf{True}$ iff $u$ is a solution to the $\ILPwithInverse$ instance.

    \begin{algorithm}[ht] 
        \caption{Procedure $\FindChildrenAndParent_{f, f^{-1}, g}(v)$ } \label{alg: findchildren inverse}
        \begin{algorithmic}[1]
            \If{$f^{-1}(v) = \bot$}
                \State \Return(v, v; v) \Comment{Self-loop on vertices with no $f$-inverse.}
            \Else
                \State $v' \gets f^{-1}(v)$ \Comment{We will use the $f$-inverse of $v$ as our starting point.}
                \State $h(v') \gets g(f(g(v')))$ 
                \Comment{Rename for notational brevity}
                \If{$\Checksol(v) \lor \Checksol(v') \lor \Checksol \left(f(g(v'))\right)$}
                    \State \Return $(\bot, \bot; f(v))$ 
                    \Comment{We have found a solution, but \underline{always} return $f(v)$.}
                \Else
                    \State \Return $(f(g(v')), f(h(v')); f(v))$
                    \Comment{The two children of $v$ are $f(g(v'))$ and $f(h(v'))$; the parent of $v$ is $f(v)$.}
                \EndIf
            \EndIf
        \end{algorithmic}
    \end{algorithm}

    \begin{figure}
        \centering
        \begin{tikzpicture}[level2/.style={fill=Cerulean!40!white, text opacity=1},level0/.style={fill=Cerulean!2!white}] %
            \node[v,level0] (fv) at (0, 0) {$f v$};
            \node[v,level2] (v) at (0, 2) {$v$};
            \node[v,level2] (fgvv) at (-2, 4) {$f g v'$};
            \node[v,level2] (vv) at (0, 4) {$v'$};
            \node[v,level0] (fhvv) at (2, 4) {$f h v'$};
            \node[v,level0] (gvv) at (-2, 6) {$g v'$};
            \node[v,level0] (hvv) at (2, 6) {$h v'$};
            
            \draw[dotted, very thick,gray] (-2.8, 4.8) rectangle (-1.2, 3.2);
            \draw[dotted, very thick,gray] (2.8, 4.8) rectangle (1.2, 3.2);
            \draw[dotted, very thick,gray] (0.8, 0.8) rectangle (-0.8, -0.8);
            
            \draw[farrow] (v) to (fv);
            \draw[farrow] (vv) to[bend left] (v);
            \draw[farrow] (gvv) to (fgvv);
            \draw[farrow] (hvv) to (fhvv);
            \draw[farrow] (fgvv) to (v);
            \draw[farrow] (fhvv) to (v);

            \draw[farrow, dotted] (v) to[bend left] (vv);

            \draw[garrow] (vv) to (gvv);
            \draw[garrow] (fgvv) to (hvv);
        \end{tikzpicture}
        \caption{The procedure performed by $\FindChildrenAndParent_{f, f^{-1}, g}(v)$.  Solid arrows represent $f$, the dotted arrow represents $f^{-1}$, and dashed arrows represent $g$. Parentheses are omitted in labels. The dotted boxes indicate the vertices that will be returned. The procedure will check if the shaded vertices are $\ILPwithInverse$ solutions, and in doing so will visit the unshaded vertices (but will not detect if these are solutions). Note that $f(h(v')) = v'$ is possible.} \label{fig: findchildren inverse}
    \end{figure}
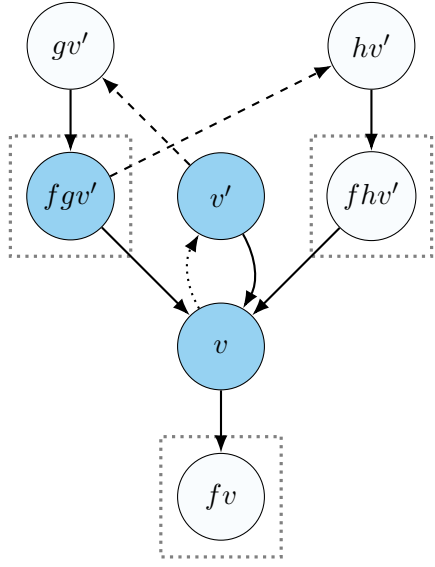

    $\FindChildrenAndParent$ is given in \cref{alg: findchildren inverse}. See \cref{fig: findchildren inverse} for an illustration. We show some useful properties of this procedure.

    \begin{claim} \label{claim: findchildren inverse properties}
        Let $\FindChildrenAndParent_{f, f^{-1}, g}(v) = (a, b; c)$. Suppose $f^{-1}(v) \neq \bot$. If $(a, b) \neq (\bot, \bot)$, then
        \begin{enumerate}[label=(\roman*)]
            \item \label{claim: findchildren inverse properties distinct} $a \neq b$,
            \item \label{claim: findchildren inverse properties same children} $f(a) = f(b) = v$.
        \end{enumerate}
    \end{claim}
    \begin{claimproof}[Proof of \cref{claim: findchildren inverse properties}.]
        \leavevmode %
        \begin{enumerate}[label=(\roman*),beginpenalty=10000] %
            \item As $a$ is not a solution, we know that $b = f(g(a)) \neq a$.

            \item As $v$ is not a solution, $f(v') = f(f^{-1}(v)) = v$. As $v'$ is not a solution, $f(a) = f(f(g(v'))) = f(v') = v$. As $a$ is not a solution, $f(b) = f(f(g(a))) = f(a) = v$. \qedhere
        \end{enumerate}
    \end{claimproof}

    Now we show the reduction. First, we need to choose a vertex from the $\ILPwithInverse$ instance to be our distinguished vertex 1 in $\emptychild'$. We will select it such that $f(f(1)) \neq f(1)$. 
    
    Let $v^{\star}$ be the lexicographically first $\ILPwithInverse$ vertex. Before performing the rest of the reduction, check the vertices $\{v^{\star}, f(v^{\star}), g(f(v^{\star})\}$ to see if they are solutions. If any are, halt the reduction by returning some easily-falsified $\emptychild'$ instance (e.g. a single vertex that self-loops on $F, L, R$). Given such an $\emptychild'$ instance we can easily find a solution to the $\ILPwithInverse$ instance by ensuring that our reduction always checks these specified vertices.

    \begin{claim} \label{claim: one not bot}
        If none of $\{v^{\star}, f(v^{\star}), g(f(v^{\star})\}$ are $\ILPwithInverse$ solutions, then either $f(f(v^{\star})) \neq f(v^{\star})$ or $f(f(g(f(v^{\star}))) \neq f(g(f(v^{\star}))$.
    \end{claim}
    \begin{claimproof}
        Assume otherwise. $f(v^{\star})$ is not a solution, so $f(f(g(f(v^{\star})))) = f(f(v^{\star}))$ and $f(g(f(v^{\star}))) \neq f(v^{\star})$. Then
        \[ f(f(v^{\star})) = f(f(g(f(v^{\star})))) = f(g(f(v^{\star}))) \neq f(v^{\star}), \]
        a contradiction.
    \end{claimproof}
    
    Set vertex 1 to be $v^{\star}$ if $f(f(v^{\star})) \neq f(v^{\star})$ and $g(f(v^{\star}))$ otherwise. In the following, assume we have relabeled $V$ such that the vertex chosen above is vertex 1.
    
    For every vertex $v$, if $\FindChildrenAndParent_{f, f^{-1}, g}(v)$ returns $(\bot, \bot, c)$, let $v^{*}$ be any vertex other than $v$, and assign $L(v) = v^{*}$, $R(v) = v^{*}$, and $F(v) = c$. This will be an $\emptychild'$ solution of type \ref{item: (Empty Child)} and we can easily compute an $\ILPwithInverse$ solution if such an $\emptychild'$ solution is detected: check $v$, $v'$, and $f(g(v'))$ for solutions. Otherwise, assign
    \[ (L(v), R(v); F(v)) \leftarrow \FindChildrenAndParent_{f, f^{-1}, g}(v). \]
    
    We will show that the only solutions to the $\emptychild'$ instance are those that are associated with $(\bot, \bot, c)$ outputs as described above. Consider any $v$ in the $\emptychild'$ instance and split into two cases: one where $F(v) = L(v) = R(v) = v$ and one where this does not hold.

    In the first case:
    \begin{itemize}
        \item $F(L(v)) = F(v) = v$ and $F(R(v)) = F(v) = v$.
        \item $L(v) = R(v) = v$.
        \item It is not the case that $v = 1$. We must have $f^{-1}(v) = \bot$ but we used \cref{claim: one not bot} to choose vertex 1 such that $f(f(1)) \neq f(1)$. Together, these would imply that $1$ is a solution of type \ref{item: Bad Bot}, but we also chose 1 to not be a solution.
    \end{itemize}
    
    In the second case, we may assume there is no $\ILPwithInverse$ solution at $v$, $v'$, and $f(g(v'))$. Then:
    \begin{itemize}
        \item $F(L(v)) = v$ and $F(R(v)) = v$ by \cref{claim: findchildren inverse properties} \cref{claim: findchildren inverse properties same children}. Note that this is why we need to always return a value for $F$ in $\FindChildrenAndParent$: even if $L(v)$ or $R(v)$ are solutions, they will still correctly point back to $v$.
        \item $L(v) \neq R(v)$ by \cref{claim: findchildren inverse properties} \cref{claim: findchildren inverse properties distinct}.
        \item $L(1) \neq 1$ and $R(1) \neq 1$. We assigned vertex 1 such that $f(f(1)) \neq f(1)$, which implies $f(1) \neq 1$. Thus, $F(L(1)) = 1$ but $F(1) \neq 1$, implying $L(1) \neq 1$ and the same argument works for $R$. \qedhere
    \end{itemize}
\end{proof}

\section{The Strength of \texorpdfstring{\lossycode}{Lossy-Code}}\label{sec: APC1}

    We start with the following candidate problem that, at the beginning of this research, was conjectured to be strictly harder than $\lossycode$. %
    
    \begin{description}
        \item[$c$-\AMGMLC.] Let $c > 1$ be a constant, $V := [2N]$ and $P := [c\cdot N^2]$.  
    
        The input is a coloring function $C : V \rightarrow \{0,1\}$ and two mappings $F : P \rightarrow V \times V$, $G: V \times V \rightarrow P$. Let $H := C^{-1}(0) \times C^{-1}(1)$. The goal is to find solutions of either type:

        \begin{description}
            \item[s1.] a pigeon $x \in P$ such that $G(F(x)) \neq x$; \mylabel{Wrong Encoding-Decoding}{item: (Wrong Enc-Dec)}
            \item[s2.] a pigeon $x \in P$ such that $F(x) \notin H$; \mylabel{Invalid Hole}{item: (Invalid Hole)}
        \end{description}
    \end{description}

	The main result in this section is that $c$-\AMGMLC is, in fact, reducible to $\lossycode$. Moreover, our techniques allow us to reduce problems similar to (and sometimes more complicated than) $c$-\AMGMLC to $\lossycode$ in a \emph{systematic} way; we provide two additional examples later ($c$-Dual-\AMGMLC in \autoref{sec: dual AM GM LC} and a problem capturing the Inclusion-Exclusion principle in \autoref{sec: incl excl}).

	\def\PV{\mathsf{PV}}
	\def\dwPHP{\mathrm{dwPHP}}
	Some of the results in this section are already known in the world of bounded arithmetic: they follow from the machineries underlying \Jerabek's theory of (additive) approximate counting $\APC_1 := \PV_1 + \dwPHP(\PV)$~\cite{Jerabek07}. For example, it is not hard to formalize the AM-GM inequality in $\APC_1$ (more precisely, prove in $\APC_1$ that $c$-\AMGMLC is total); Wilkie's witnessing theorem~(\cite{Thapen-PhD}, \cite[Proposition 1.14]{Jerabek04}) implies that every $\NP$ search problem provably total in $\APC_1$ (including $c$-\AMGMLC) reduces to $\lossycode$. 

	One of the goals of this section is to introduce the ideas of $\APC_1$ to audiences who are less familiar with $\APC_1$ (or bounded arithmetic in general). In \autoref{sec: reconstructive PRG}, we introduce reconstructive pseudorandom generators with \emph{feasible witnesses}, which is the central technique underlying \cite{Jerabek07}. This technique allows us to put a wide range of problems similar to $c$-\AMGMLC inside $\lossycode$.

    \subsection{Basics of \texorpdfstring{$\APC_1$}{APC1}}\label{sec: reconstructive PRG}

This subsection presents some background of $\APC_1$ \cite{Jerabek07} that is needed in our reductions. In \cite{Jerabek07}, \Jerabek uses the \emph{Nisan--Wigderson generator} \cite{NisanW94} to approximate the size of feasible sets. Looking ahead, we will abstract the properties needed for the Nisan--Wigderson generator as \emph{reconstructive PRGs} with ``feasible witnesses'' in some sense.

    \def\PRG{\mathsf{PRG}}
    \def\calN{\mathcal{N}}
    \def\eps{\varepsilon}
    \def\Recon{\mathsf{Recon}}
    \def\seed{\mathsf{seed}}
    \def\cupdot{\mathbin{\dot{\cup}}}
    \def\Comp{\mathsf{Comp}}
    \def\Decomp{\mathsf{Decomp}}

We first define the notion of \emph{injection-surjection pairs}. Let $A, B$ be two sets. We say that there is an \emph{injection-surjection pair} certifying $|A| \le |B|$ if there are polynomial-time computable functions $f: A\to B$ and $g: B\to A$ such that for every $x\in A$, $g(f(x)) = x$. This will be denoted by the following notation
\[A \xrightleftharpoons[g]{f} B.\]
(Note that in the above notation, the injection-surjection pair is implicitly assuming that the left-hand side ($A$) is smaller than the right-hand side ($B$); hence $A\rightleftharpoons B$ and $B\rightleftharpoons A$ have very different meanings.)

Injection-surjection pairs are the basic primitive used in $\APC_1$ to compare the sizes of two sets; roughly speaking, this is because the underlying principle is the \emph{retraction (weak) pigeonhole principle} (over polynomial-time functions; i.e., the totality of $\lossycode$). 

For any sets $A, B$ and two functions $A\xrightleftharpoons[g]{f} B$, we say that an input $x\in A$ \emph{witnesses} that $A\rightleftharpoons B$ is not an injection-surjection pair, if $g(f(x)) \ne x$. Otherwise (i.e., $g(f(x)) = x$), we say that $x$ \emph{maps to itself} via $A\rightleftharpoons B$. We also use $X \cupdot Y$ to denote the \emph{disjoint union} of two sets $X, Y$.

\paragraph{$\APC_1$-provably reconstructive pseudorandom generators.} Let $N = 2^n, D = 2^d, M = 2^m$, a function $\PRG: [N] \times [D] \to [M]$ is called a $(k, \eps)$-reconstructive PRG if for every subset $S\subseteq [M]$, for all but at most $K := 2^k$ values of $f \in [N]$, we have
\begin{equation}
    \mleft|\frac{|S \cap \PRG_f|}{D} - \frac{|S|}{M}\mright| \le \eps, \label{eq: extractor property}
\end{equation}
where we write $\PRG_f = \{\PRG(f, \seed): \seed\in [D]\}$ for convenience. Furthermore, we want this generator to have ``feasible witnesses'' in the following sense:
\begin{definition}
    Let $G_<^{(-)}, H_<^{(-)}, G_>^{(-)}, H_>^{(-)}$ be polynomial-time oracle algorithms. We say that $f \in [N]$ \emph{provides an $\eps$-(additive) approximation of $|S|$ feasibly}, if letting $val := |S\cap \PRG_f| \cdot (M/D)$, then we have the following two injection-surjection pairs:
    \begin{align}
        [val]\times [D]\xrightleftharpoons[G_<^S(f, -)]{H_<^S(f, -)}&\, (S\cupdot [\eps M])\times [D],\text{ and}\label{eq: PRG inj surj pair I}\\
        S\times [D]\xrightleftharpoons[G_>^S(f, -)]{H_>^S(f, -)}&\, ([val]\cupdot [\eps M])\times [D].\label{eq: PRG inj surj pair II}
    \end{align}
    Roughly speaking, these injection-surjection pairs certify \eqref{eq: extractor property}: $(G_<, H_<)$ certifies $\frac{|S\cap \PRG_f|}{D} \le \frac{|S|}{M} + \eps$, while $(G_>, H_>)$ certifies $\frac{|S|}{M} \le \frac{|S\cap\PRG_f|}{D} + \eps$.
\end{definition}

Fix $S \subseteq [M]$, we want that all but $K$ values $f \in [N]$ provide an $\eps$-additive approximation of $|S|$ feasibly. In fact, we can construct a pair of \emph{reconstruction algorithm} $\Comp$ and $\Decomp$, which are deterministic oracle algorithms satisfying the following. Given any witness $w$ that \eqref{eq: PRG inj surj pair I} or \eqref{eq: PRG inj surj pair II} is \emph{not} an injection-surjection pair, $\Comp^S(f, w)$ compresses $f$ into an element $\tilde{f} \in [K]$ (i.e., compresses $f$ into $k$ bits), and $\Decomp^S(\tilde{f})$ decompresses $\tilde{f}$ back to $f$.

The following theorem asserts that some explicit generator (in fact, the Nisan--Wigderson generator~\cite{NisanW94}) has ``feasible witnesses'' in the above sense. It is implicit in \cite{Jerabek07}; for completeness, we provide a proof in \autoref{sec: proof of NW in APC1}.

\def\EitherNotInjSurjPair{either \eqref{eq: PRG inj surj pair I} or \eqref{eq: PRG inj surj pair II}\xspace}
\begin{restatable}{theorem}{ThmNisanWigdersoninAPC}\label{thm: certified Nisan--Wigderson}
    Let $n, m \in \N$, $\eps > 0$, and $\rho > 1$ be parameters. Let $d := O\mleft(\frac{\log^2 (nm/\eps)}{\log\rho}\mright)$ and $k := d + (\rho + 1)(m - 1) + O(\log(m/\eps))$. Let $N := 2^n$, $M := 2^m$, $K := 2^k$, and $D := 2^d$.
    
    Then there is a $(k, \eps)$-reconstructive generator $\PRG: [N] \times [D] \to [M]$ with $\poly(\rho mn/\eps)$-time deterministic oracle algorithms $G_<, H_<, G_>, H_>, \Comp, \Decomp$ such that the following holds. For every set $S\subseteq [M]$, every $f\in [N]$, and every input $w$ witnessing that \EitherNotInjSurjPair is not an injection-surjection pair (where $val := |S\cap \PRG_f|\cdot \frac{M}{D}$), we have $\Comp^S(f, w) \in [K]$ and
    \[\Decomp^S(\Comp^S(f, w)) = f.\]
\end{restatable}

In particular, the above theorem implies that if we take $f$ to be ``hard enough'' (i.e., $f$ is not in the range of $\Decomp^S$) then $f$ provides an $\eps$-additive approximation of $|S|$ feasibly, analogous to the classical theorem that we can use a hard truth table to derandomize $\BPP$. This is reminiscent of the ``compress-or-random'' technique in catalytic computing~\cite{Pyne24, CookLMP25, KMPS25, AgarwalaM25}: Any string $f$ is either ``random'', which means it can be used for derandomization, or ``compressible'' and admits a short description.

\subsection{Reductions to Lossy-Code}
\subsubsection{\texorpdfstring{\AMGMLC}{AMGM-LC}}    

We first show that $c$-\AMGMLC is in $\lossycode$.

\begin{theorem}
    For every constant $c > 1$, there is a decision tree reduction of $\polylog(N)$ query complexity from $c$-AMGM-LC to $\lossycode$.
\end{theorem}
\begin{proof}
    Let $\eps := (c-1)/3$, $\rho := \lceil\log N\rceil$, $d := O(\log\log N)$, $k := d + (\rho + 1)\log |V| + O(\log (|V|/\eps)) \le O(\log^2 N)$, and $n' := 10\lceil \log N\rceil^2$. Let $D := 2^d$, $K := 2^k$, and $N' := 2^{n'}$. Let $\PRG: [N']\times [D] \to V$ be defined in \autoref{thm: certified Nisan--Wigderson}.
    
    Let $f \in [N']$; for now, it would be convenient to assume that $f$ successfully provides $\eps$-approximations feasibly. Then we can estimate the size of $C^{-1}(0)$ as $v_0 := |C^{-1}(0) \cap \PRG_f|\cdot (2N) / D$; the estimation of $|C^{-1}(1)|$ is thus $2N - v_0 = |C^{-1}(1) \cap \PRG_f|\cdot (2N) / D$. Furthermore, we can obtain injection-surjection pairs (that depend on $f$)
    \begin{align}
        C^{-1}(0) \times [D] \rightleftharpoons &\, [v_0 + 2\eps N] \times [D]~\text{and}\label{eq: inj-surj pair for 0}\\
        C^{-1}(1) \times [D] \rightleftharpoons &\, [2N - v_0 + 2\eps N] \times [D].\label{eq: inj-surj pair for 1}
    \end{align}
    Since
    \[|C^{-1}(0)|\cdot |C^{-1}(1)| \le (v_0 + \eps \cdot 2N) \cdot (2N - v_0 + \eps \cdot 2N) \le (1+2\eps)N^2,\]
    we obtain an injection-surjection pair
    \begin{equation}\label{eq: inj-surj pair for 0 and 1}
        C^{-1}(0) \times C^{-1}(1) \times [D]^2 \rightleftharpoons [(1+2\eps)N^2]\times [D]^2.
    \end{equation}
    Combining this with $(F, G)$, the purported injection-surjection pair from $P$ to $H$, we obtain an injection-surjection pair
    \begin{equation}\label{eq: final inj surj pair}
        P\times [D]^2 \rightleftharpoons [(1+2\eps)N^2] \times [D]^2.
    \end{equation}
    Since $|P| = cN^2 > (1+\Omega(1)) \cdot (1+2\eps)N^2$, this would be a contradiction to the retraction weak pigeonhole principle. Hence, by solving $\lossycode$ we can find a witness that \eqref{eq: final inj surj pair} is not an injection-surjection pair, which is also an answer of the original $c$-\AMGMLC instance.

    We have shown that given some $f$ that successfully provides $\eps$-approximations, we can reduce $c$-\AMGMLC to $\lossycode$. But what if $f$ does not provide such approximations? In this case, $\Comp$ allows us to compress such $f$ into a $k$-bit string. In particular:
    \begin{itemize}
        \item For any input $(x, y)$ witnessing that \eqref{eq: inj-surj pair for 0} is not an injection-surjection pair, $\Comp^{C^{-1}(0)}(f, x, y)$ returns some $\tilde{f} \in [K]$ such that $\Decomp^{C^{-1}(0)}(\tilde{f}) = f$.
        \item For any input $(x, y)$ witnessing that \eqref{eq: inj-surj pair for 1} is not an injection-surjection pair, $\Comp^{C^{-1}(1)}(f, x, y)$ returns some $\tilde{f} \in [K]$ such that $\Decomp^{C^{-1}(1)}(\tilde{f}) = f$.
    \end{itemize}

    Now we are ready to formally present our reduction from AMGM-LC to $\lossycode$. Let $X := [N'] \times P \times [D]^2$, $Y_1 := [N'] \times [(1+2\eps)N^2] \times [D]^2$, $Y_2 := [2] \times [K] \times P \times [D]^2$, and $Y := Y_1 \cupdot Y_2$. Then
    \begin{align*}
        |Y| =&\, N' (1+2\eps)N^2 \cdot D^2 + 2K\cdot cN^2 \cdot D^2\\
        \le&\, N'(1+2.1\eps)N^2D^2\\
        \le&\, (1-\Omega(1))N'\cdot cN^2\cdot D^2 = (1-\Omega(1))|X|.
    \end{align*}

    We reduce the AMGM-LC instance to a $\lossycode$ instance that consists of a pair of functions $F^*: X\to Y$ and $G^*: Y\to X$.\begin{itemize}
        \item The function $F^*: X\to Y$ takes as inputs $f \in [N']$ and $(p, u_1, u_2) \in P\times [D]^2$. It considers the injection defined in \eqref{eq: final inj surj pair} that corresponds to $f$ and computes its output $(q, v_1, v_2)$ on input $(p, u_1, u_2)$. Then it checks whether $f$ is ``good'': if $f$ is ``good'' then it outputs $(f, q, v_1, v_2) \in Y_1$, otherwise it outputs some value in $Y_2$ containing a compression of $f$.
        
        More precisely, let $(a, b) := F(p)$, assume that $a \in C^{-1}(0)$ and $b\in C^{-1}(1)$. To evaluate \eqref{eq: final inj surj pair} on input $(p, u_1, u_2)$, we need to evaluate $(a, u_1)$ on \eqref{eq: inj-surj pair for 0} and $(b, u_2)$ on \eqref{eq: inj-surj pair for 1}. If $(a, u_1)$ witnesses that \eqref{eq: inj-surj pair for 0} is not an injection-surjection pair, then $f$ is not ``good'', and $f$ can be compressed into $\tilde{f} := \Comp^{C^{-1}(0)}(f, a, u_1) \in [K]$. We return $(0, \tilde{f}, p, u_1, u_2) \in Y_2$ in this case. Similarly, if $(b, u_2)$ witnesses \eqref{eq: inj-surj pair for 1} is not an injection-surjection pair, then we let $\tilde{f} := \Comp^{C^{-1}(1)}(f, b, u_2) \in [K]$ and return $(1, \tilde{f}, p, u_1, u_2) \in Y_2$. If none of the above happens, then $f$ is ``good'' and we return $(f, q, v_1, v_2) \in Y_1$.
        
        \item The function $G^*: Y\to X$ can be decomposed into functions $G_1^*: Y_1 \to X$ and $G_2^*: Y_2 \to X$. In either case, we are given a ``compressed'' form of some $(f, p, u_1, u_2) \in X$ and need to recover $(f, p, u_1, u_2)$.

        For $G_1^*: Y_1 \to X$, the ``compressed form'' is $(f, q, v_1, v_2) \in Y_1$. This corresponds to the case that $f$ is ``good''. We consider the surjection defined in \eqref{eq: final inj surj pair} corresponding to $f$ and compute its output $(p, u_1, u_2)$ given input $(q, v_1, v_2)$. Then we output $(f, p, u_1, u_2) \in X$.

        For $G_2^*: Y_2 \to X$, the ``compressed form'' is $(b, \tilde{f}, p, u_1, u_2) \in Y_2$. This corresponds to the case that $f$ is ``not good'', and we can directly compute $f := \Decomp^{C^{-1}(b)}(\tilde{f})$. Then we output $(f, p, u_1, u_2) \in X$.
    \end{itemize}

    To see the correctness of this reduction, consider any input $(f, p, u_1, u_2)$ witnessing that $X\xrightleftharpoons[G^*]{F^*}Y$ is not an injection-surjection pair. There are two cases:
    \begin{itemize}
        \item Suppose that $F^*(f, p, u_1, u_2) = (f, q, v_1, v_2) \in Y_1$. Letting $(a, b) := H(p)$, this means that $(a, u_1)$ maps to itself via \eqref{eq: inj-surj pair for 0}, and $(b, u_2)$ maps to itself via \eqref{eq: inj-surj pair for 1}. In other words, $(a, b, u_1, u_2)$ maps to itself via \eqref{eq: inj-surj pair for 0 and 1}. Hence, it has to be the case that $P\xrightleftharpoons[G]{F} C^{-1}(0) \times C^{-1}(1)$ does not map $p$ to itself, which means that $p$ is a solution to AMGM-LC.
        \item The other case is that $F^*(f, p, u_1, u_2) = (b, \tilde{f}, p, u_1, u_2) \in Y_2$. We claim that this case would not have happened. Indeed, in this case, $G_2^*$ always maps $(b, \tilde{f}, p, u_1, u_2)$ back to $(f, p, u_1, u_2)$.
    \end{itemize}
    
    Finally, it is easy to see that both $F^*$ and $G^*$ run in deterministic $\polylog(N)$ time.
\end{proof}

\subsubsection{\texorpdfstring{Dual-\AMGMLC}{Dual-AMGM-LC}}\label{sec: dual AM GM LC}

Now we define the \emph{dual} of the problem $c$-\AMGMLC, which expresses the AM-GM inequality in a different way:
\begin{description}
    \item[$c$-Dual-\AMGMLC.] 
    Let $V = [2N], H = [c\cdot N^2]$. The input consists of a coloring function $C : V \rightarrow \{0,1\}$ and two mappings $F : H \rightarrow V \times V$, $G: V \times V \rightarrow H$. Let
    \[P := \{(u, v) : C(u) = C(v)\} = (C^{-1}(0) \times C^{-1}(0)) \cup (C^{-1}(1) \times C^{-1}(1)).\]
    The goal is to find a pigeon $x\in P$ such that $F(G(x)) \ne x$.
\end{description}

\begin{theorem}
    For every constant $0 < c < 2$, there is a decision tree reduction of $\polylog(N)$ query complexity from $c$-Dual-\AMGMLC to $\lossycode$.
\end{theorem}
\begin{proof}
    Let $\eps := (2-c)/32$ and $c' := c + 8\eps(1+\eps)$, then $c' < 2$. Let $\rho := \lceil\log N\rceil$, $d := O(\log \log N)$, $n' := 10\lceil \log N\rceil^2$, and $k := d + (\rho + 1)(\log|V|-1) + O(\log (|V|/\eps)) \le 2\lceil\log N\rceil^2$. Let $D := 2^d = \polylog(N)$, $K := 2^k$, $N' := 2^{n'}$. Consider the generator $\PRG: [N']\times D \to V$ in \autoref{thm: certified Nisan--Wigderson}.

    Let $f \in [N']$. Let $v_0 := |C^{-1}(0)\cap \PRG_f| \cdot (2N/D)$ and $v_1 := 2N - v_0$ be the estimations of $|C^{-1}(0)|$ and $|C^{-1}(1)|$ provided by $\PRG_f$, respectively. Then we obtain (purported) injection-surjection pairs
    \begin{align*}
        [v_0]\times [D] &\rightleftharpoons (C^{-1}(0) \cupdot [2\eps N])\times [D]\text{ and}\\
        [v_1]\times [D] &\rightleftharpoons (C^{-1}(1) \cupdot [2\eps N])\times [D].
    \end{align*}
    This implies an injection-surjection pair
    \begin{align*}
        [v_0^2 + v_1^2]\times [D]^2\xrightleftharpoons[G_f]{F_f}&~\mleft((C^{-1}(0)\cupdot [2\eps N])^2 \cupdot (C^{-1}(1)\cupdot [2\eps N])^2\mright) \times [D]^2\\
        =&~ \mleft(P \cupdot (C^{-1}(0)\cupdot C^{-1}(1))\times [4\eps N])\cupdot [2(2\eps N)^2]\mright)\times [D]^2\\
        =&~ (P\cupdot [8\eps(1+\eps) N^2])\times [D]^2.
    \end{align*}
    Composing this with our input $P\xrightleftharpoons[G]{F}H$ and noting that $H = [cN^2]$ and $v_0^2 + v_1^2 \ge 2N^2$, we obtain
    \[[2N^2]\times [D]^2 \xrightleftharpoons[G_f\circ G]{F\circ F_f}[(c+8\eps(1+\eps))N^2]\times [D]^2 = [c'N^2]\times [D]^2.\]
    (Note that we are abusing notation here in ``$G_f\circ G$'' and ``$F\circ F_f$''; for example, we are extending the domain of $G$ to $H\cupdot [8\eps(1+\eps)N^2]$ so that $G$ is the identity map on $[8\eps(1+\eps)N^2]$.)
    
    Since $(c+8\eps(1+\eps)) < 2 - \Omega(1)$, $(F\circ F_f, G\circ G_f)$ is a valid $\lossycode$ instance. Solving this $\lossycode$ instance gives us a witness that $(F\circ F_f, G\circ G_f)$ is not an injection-surjection pair. This implies either a witness that $(F, G)$ is not an injection-surjection pair (which is what we need), or a witness that $(F_f, G_f)$ is not an injection-surjection pair (which will allow us to compress $f$).

    Now we formally define the complete reduction from $c$-Dual-\AMGMLC to $\lossycode$. Let $X := [N'] \times [2N^2]\times [D]^2$, $Y_1 := [N']\times [c'N^2]\times [D]^2$, $Y_2 := [K+O(1)]\times [2N^2]\times [D]^2$, and $Y := Y_1 \cupdot Y_2$. Note that
    \begin{align*}
        |Y| =&\, |Y_1| + |Y_2| \le (c'+o(1))N' N^2D^2 \le (2-\Omega(1))N' N^2D^2 \le (1-\Omega(1))|X|.
    \end{align*}
    
    We reduce the $c$-Dual-\AMGMLC instance $(C, F, G)$ to the $\lossycode$ instance $X\xrightleftharpoons[G^*]{F^*}Y$, defined as follows.
    \begin{itemize}
        \item The function $F^*: X\to Y$ takes as inputs $f\in [N']$ and $u\in [2N^2]\times [D]^2$. It tests if $u$ is a witness that $(F_f, G_f)$ is not an injection-surjection pair. If this is the case, then it uses $\Comp$ to compress $f$ into $K + O(1)$ bits as $\tilde{f}$, and returns $(\tilde{f}, u) \in Y_2$. Otherwise it computes $v := F(F_f(u))\in [c'N^2]\times [D]^2$ and returns $(f, v)\in Y_1$.
        \item The function $G^*: Y\to X$ takes either $(f, v)\in Y_1$ or $(\tilde{f}, u)\in Y_2$ as inputs. If the input is $(f, v)\in Y_1$ where $f\in [N']$ and $v\in [c'N^2]\times [D]^2$, then it computes $u := G(G_f(v)) \in [2N^2]\times [D]^2$ and returns $(f, u)\in X$. If the input is $(\tilde{f}, u)\in Y_2$ where $\tilde{f}\in [K + O(1)]$ and $u\in [2N^2]\times [D]^2$, then it decompresses $f\in [N']$ from $\tilde{f}$ and returns $(f, u) \in X$.
    \end{itemize}
    Given any $(f, u)\in X$ witnessing that $(F^*, G^*)$ is not an injection-surjection pair, we have that $F^*(f, u) \in Y_1$ and we can find a witness that $(F, G)$ is not an injection-surjection pair in deterministic $\polylog(N)$ time. Finally, $F^*, G^*$ can be computed in deterministic $\polylog(N)$ time as well.
\end{proof}

\subsubsection{The Inclusion-Exclusion Principle}\label{sec: incl excl}

Finally, as a proof-of-concept, we consider the following more complicated total search problem capturing the \emph{inclusion-exclusion principle} and show that it is in $\Lossy$. Given that \Jerabek already proved this principle in $\APC_1$~\cite[Proposition 2.19]{Jerabek07}, it should come as no surprise that this problem is in $\Lossy$. Needless to say, our proof is just a translation of \Jerabek's proof into the language of black-box total search problems. We include this example as an additional demonstration of how to reduce more complicated problems to $\lossycode$ via \autoref{thm: certified Nisan--Wigderson}.

Consider the following inequality expressing the inclusion-exclusion principle: Let $S_1, S_2, \dots, S_\ell \subseteq [N]$ be sets, $a_i := |S_i|$, $a_{i,j} := |S_i \cap S_j|$, and
\begin{equation}\label{eq: def of tilde a}
	\tilde{a} := \sum_{i=1}^\ell a_i - \sum_{1\le i < j \le \ell} a_{i,j},
\end{equation}
then $\mleft|\bigcup_{i\in [\ell]}S_i\mright| \ge \tilde{a}$.

Now we formalize the above inequality as a $\TFZPP$ problem $\textsc{Inclusion-Exclusion}$. 

\begin{description}
    \item[$\textsc{Inclusion-Exclusion}$.] The input consists of:\begin{itemize}
        \item Parameters $N$, $\ell \le \polylog(N)$, and $\eps > 1/\polylog(N)$.
        \item A table $T \subseteq [\ell]\times [N]$ where $T_{i,j} = 1$ if and only if $j \in S_i$.
        \item Numbers $0\le a_i \le N$ for each $i\in [\ell]$, and $0\le a_{i,j} \le N$ for each $1\le i < j\le \ell$, which are purported estimates for $|S_i|$ and $|S_i \cap S_j|$ respectively. Let $\tilde{a}$ be defined as in \eqref{eq: def of tilde a} and assume $\tilde{a} \ge \eps N$.
        \item Injection-surjection pairs $f_i: [a_i] \to [N]$ and $g_i: [N] \to [a_i]$.
        \item Injection-surjection pairs $f_{i,j}: [a_{i,j}] \to [N]$ and $g_{i,j}: [N] \to [a_{i,j}]$.
        \item Finally, an injection-surjection pair $\tilde{f}: [\tilde{a} - \eps N] \to [N]$ and $\tilde{g}: [N] \to [\tilde{a} - \eps N]$.
    \end{itemize}
    The goal is to find any solution of the following types:
    \begin{description}
        \item [s1.] some $x\in [a_i]$ such that $T_{i, f_i(x)} = 0$ or $g_i(f_i(x)) \ne x$;\mylabel{Violation of $|S_i|\ge a_i$}{item: SOL1}

        \item [s2.] some $x\in [N]$ such that ($T_{i, x} = T_{j, x} = 1$) but $f_{i,j}(g_{i,j}(x)) \ne x$; or
        
        \mylabel{Violation of $|S_i\cap S_j| \le a_{i,j}$}{item: SOL2}

        \item [s3.] some $x\in [N]$ such that $\tilde{f}(\tilde{g}(x)) \ne x$, and there is some $j\in [\ell]$ such that $T_{j, x} = 1$.
            
        \mylabel{Violation of $\mleft|\bigcup_{i\in [\ell]}S_i\mright| \le \tilde{a} - \eps N$}{item: SOL3}
    \end{description}
\end{description}

    \begin{theorem}
        $\textsc{Inclusion-Exclusion}$ reduces to $\lossycode$.
    \end{theorem}
    \begin{proof}
        Let $\eps' := \eps / (10\ell^2)$, $\rho := \lceil \log N\rceil$, $d := O(\log\log N)$, $k := d + (\rho + 1) \log N + O(\log (N/\eps)) \le 3\lceil \log N\rceil^2$, and $n' := \log(100\ell^4\eps^{-1}) + k \le 4\lceil\log N\rceil^2$. Let $D := 2^d$, $K := 2^k$, and $N' := 2^{n'}$. Note that $d' \ge O(\frac{\log^2 (n' \cdot\log N / \eps')}{\log\rho})$, hence it is valid to apply \autoref{thm: certified Nisan--Wigderson} to obtain $\PRG: [N']\times [D] \to [N]$. Recall that we define $S_i := \{j \in [N]: T_{i, j} = 1\}$. We also define $S := \bigcup_{i\in [\ell]} S_i$.

        Let $f\in [N']$. We use $f$ to estimate each $|S_i|$, $|S_i\cap S_j|$, and $|S|$ within additive error $\eps'\cdot N$. Let $v_i := |S_i\cap \PRG_f|\cdot (N/D)$, $v_{i,j} := |S_i\cap S_j \cap \PRG_f|\cdot (N/D)$, and $v := |S\cap \PRG_f|\cdot (N/D)$, then we obtain (purported) injection-surjection pairs
        \begin{align}
            S_i\times [D] &\,\xrightleftharpoons[G_i]{F_i} ([v_i]\cupdot [\eps' N])\times [D],\label{eq: estimate Si}\\
            [v_{i,j}]\times [D] &\,\xrightleftharpoons[F_{i,j}]{G_{i,j}} ((S_i \cap S_j)\cupdot [\eps' N])\times [D],\label{eq: estimate Sij}\\
			[v]\times [D] &\,\xrightleftharpoons[F]{G} (S\cupdot [\eps' N])\times [D].\label{eq: estimate S}
        \end{align}

        If $v_i + 2\eps' N < a_i$, then we can compose \eqref{eq: estimate Si} with the injection-surjection pair $(f_i, g_i)$ to obtain the $\lossycode$ instance
        \begin{equation}
			[a_i]\times [D] \xrightleftharpoons[{g_i\times [D]}]{f_i\times [D]}S_i\times [D] \xrightleftharpoons[G_i]{F_i} [v_i + \eps' N] \times [D].
			\label{eq: lossy-code instance I}
		\end{equation}
        Here, we use the notation $f_i\times [D]$ to denote the function mapping $(x, y)$ to $(f_i(x), y)$, where $x\in [a_i]$ and $y\in [D]$; $g_i\times [D]$ is defined similarly. Note that since $\frac{(v_i + \eps' N)\cdot D}{a_i\cdot D} \le 1-\frac{\eps' N}{a_i} \le \eps'$, \eqref{eq: lossy-code instance I} is a valid $\lossycode$ instance with good stretch; in particular, it reduces to standard $\lossycode$ (of stretch $[2t]\rightleftharpoons [t]$) in decision tree depth $O(1/\eps')$. Solving the instance \eqref{eq: lossy-code instance I} gives us an element $(x, y) \in [a_i]\times [D]$ such that either (1) $x$ is a \ref{item: SOL1} or (2) $(f_i(x), y) \in S_i\times [D]$ witnesses that \eqref{eq: estimate Si} is not a valid injection-surjection pair.

		Similarly, if $v_{i,j} > a_{i,j} + 2\eps' N$, then we compose \eqref{eq: estimate Sij} with the injection-surjection pair $(f_{i, j}, g_{i, j})$ to obtain the $\lossycode$ instance 
		\begin{equation}
			[v_{i,j}]\times [D]\xrightleftharpoons[F_{i,j}]{G_{i,j}} ((S_i\cap S_j)\cupdot[\eps' N]) \times [D] \xrightleftharpoons[{f_{i,j}\times [D]}]{g_{i,j}\times [D]}[a_{i,j} + \eps' N]\times [D].
			\label{eq: lossy-code instance II}
		\end{equation}
		Again, we abuse notation to extend $f_{i,j}$ and $g_{i,j}$ into functions $f_{i,j}:[a_{i,j}]\cupdot[\eps' N]\to [N]\cupdot[\eps' N]$ and $g_{i,j}:[N]\cupdot[\eps' N] \to [a_{i,j}]\cupdot[\eps' N]$ such that they are the identity function over $[\eps' N]$. \eqref{eq: lossy-code instance II} is a valid $\lossycode$ instance that reduces to the standard $\lossycode$ in decision tree depth $O(1/\eps')$. Solving the instance \eqref{eq: lossy-code instance II} gives us an element $(x, y) \in [v_{i,j}]\times [D]$ such that either (1) $x' := G_{i,j}(x)$ is a \ref{item: SOL2} or (2) $(x, y)$ witnesses that \eqref{eq: estimate Sij} is not a valid injection-surjection pair.

		Finally, if $v > \tilde{a} -\eps N + 2\eps' N$, then we compose \eqref{eq: estimate S} with the injection-surjection pair $(\tilde{f}, \tilde{g})$ to obtain the $\lossycode$ instance
		\begin{equation}
			[v]\times [D] \xrightleftharpoons[F]{G}(S\cupdot[\eps' N])\times [D] \xrightleftharpoons[{\tilde{f}\times [D]}]{\tilde{g}\times [D]}[\tilde{a} - \eps N + \eps' N]\times [D].
			\label{eq: lossy-code instance III}
		\end{equation}
		Solving the instance \eqref{eq: lossy-code instance III} gives us an element $(x, y)\in [v]\times [D]$ such that either (1) $x' := G(x)$ is a \ref{item: SOL3} or (2) $(x, y)$ witnesses that \eqref{eq: estimate S} is not a valid injection-surjection pair.

		We now observe that one of the above three cases must happen. This is exactly due to the inclusion-exclusion principle itself: Let $\tilde{S}_i := S_i\cap \PRG_f$, then
		\[\mleft|\bigcup_{i\in [\ell]}\tilde{S}_i\mright| \ge \sum_{i=1}^\ell |\tilde{S}_i| - \sum_{1\le i < j \le \ell} |\tilde{S}_i\cap \tilde{S}_j|.\]
		Recall that $v_i = |\tilde{S}_i|(N/D)$, $v_{i,j} = |\tilde{S}_i\cap \tilde{S}_j|(N/D)$, and $v = \mleft|\bigcup_{i\in [\ell]}\tilde{S}_i\mright| (N/D)$, hence $v \ge \sum_{i=1}^\ell v_i - \sum_{1\le i < j\le \ell} v_{i, j}$. If all three cases above do not happen, then
		\begin{align*}
			\tilde{a} \ge&\, v + \eps N - 2\eps' N & (\because v > \tilde{a} - \eps N + 2\eps'N\text{ does not happen})\\
			\ge&\, \sum_{i=1}^\ell v_i - \sum_{1\le i < j\le \ell} v_{i,j} + \eps N - 2\eps' N\\
			\ge&\, \sum_{i=1}^\ell (a_i - 2\eps' N) - \sum_{1\le i < j\le \ell}v_{i, j} + \eps N - 2\eps' N & (\because v_i + 2\eps' N < a_i\text{ does not happen})\\
			\ge&\, \sum_{i=1}^\ell (a_i - 2\eps' N) - \sum_{1\le i < j\le \ell} (a_{i, j} + 2\eps' N) + \eps N - 2\eps' N & (\because v_{i, j} > a_{i, j} + 2\eps' N\text{ does not happen})\\
			\ge&\, \sum_{i=1}^\ell a_i - \sum_{1\le i < j\le \ell} a_{i, j} + (\eps - 5\ell^2 \eps') N > \tilde{a},
		\end{align*}
        a contradiction.

        Now we are ready to describe our reduction from $\textsc{Inclusion-Exclusion}$ to $\lossycode$. Let $X := [N'] \times [2N]\times [D]$ and $Y$ be some set that we will define later, we will produce a $\lossycode$ instance $X\xrightleftharpoons[G^*]{F^*} Y$. Given an element in $X$, we parse it as $(f, x, y)$ where $f \in [N']$, $x\in [2N]$, and $y\in [D]$. To compute $F^*(f, x, y)$, we find out which of the above three cases happens for $f$.
        \begin{itemize}
            \item {\bf Case I:} Suppose that there exists $i\in [\ell]$ such that $v_i + 2\eps' N < a_i$. If $x\not\in [a_i]$, then we let $x' := x - \eps' N$ and $y' := y$; otherwise we feed $(x, y)$ into
            \begin{equation}
				[a_i]\times [D] \xrightleftharpoons[{g_i\times [D]}]{f_i\times [D]}S_i\times [D] \xrightleftharpoons[G_i]{F_i} [v_i + \eps' N] \times [D]
			\tag{\ref{eq: lossy-code instance I}}
			\end{equation}
            to obtain $(x', y') \in [v_i + \eps' N]\times [D]$. We define $Y_0 := [N'] \times [2N - \eps' N] \times [D]$ and map $F^*(f, x, y) = (f, x', y') \in Y_0$.
            
            Note that the evaluation of \eqref{eq: lossy-code instance I} might be ill-behaved in the following cases. Suppose $x\in [a_i]$ and let $x' := f_i(x)$. If $x'\not\in S_i$ or $g_i(x') \ne x$, then we know that $x'$ is a \ref{item: SOL1} and we simply let $F^*(f, x, y) := \bot$. Otherwise, if $(x', y)$ witnesses that $(F_i, G_i)$ is not an injection-surjection pair, then $S_i$ is a distinguisher for $\PRG_f$. Let $Y_1 := [K]\times [\ell]\times [2N]\times [D]$, we define $F^*(f, x, y) := (\Comp^{S_i}(f, x', y), i, x, y) \in Y_1$. If both cases above do not happen, we define $F^*(f, x, y) := (f, x', y') \in Y_0$ as usual.
            \item {\bf Case II:} Suppose that there exists $1\le i < j\le \ell$ such that $v_{i,j} > a_{i,j} + 2\eps' N$. If $x\not\in [v_{i,j}]$, then we let $x' := x - \eps' N$ and $y' := y$; otherwise we feed $(x, y)$ into
            \begin{equation}
				[v_{i,j}]\times [D]\xrightleftharpoons[F_{i,j}]{G_{i,j}} ((S_i\cap S_j)\cupdot[\eps' N]) \xrightleftharpoons[{f_{i,j}\times [D]}]{g_{i,j}\times [D]}[a_{i,j} + \eps' N]\times [D]
				\tag{\ref{eq: lossy-code instance II}}
			\end{equation}
            to obtain $(x', y') \in [a_{i,j} + \eps' N] \times [D]$. We let $F^*(f, x, y) := (f, x', y') \in Y_0$.

            Similarly, the evaluation of \eqref{eq: lossy-code instance II} might be ill-behaved in the following cases. Suppose $x\in [v_{i, j}]$. If $(x, y)$ witnesses that $(G_{i,j}, F_{i,j})$ is not a valid injection-surjection pair, then $S_i\cap S_j$ is a distinguisher for $\PRG_f$; in this case, letting $Y_2 := [K] \times [\ell]^2 \times [2n]\times [D]$, we map $F^*(f, x, y) := (\Comp^{S_i \cap S_j}(f, x, y), i, j, x, y) \in Y_2$. Otherwise, let $(x', y') := G_{i,j}(x, y)$, if $x' \in S_i\cap S_j$ but $x'$ witnesses that $(g_{i,j}, f_{i,j})$ is not a valid injection-surjection pair, then $x'$ is a \ref{item: SOL2} and we simply let $F^*(f, x, y) := \bot$. If both cases above do not happen, we define $F^*(f, x, y) = (f, x', y') \in Y_0$ as usual.
            \item {\bf Case III:} Suppose that $v > \tilde{a} - \eps N + 2\eps' N$. If $x\not\in [v]$, then we let $x' := x - \eps' n$ and $y' := y$; otherwise we feed $(x, y)$ into
	 		\begin{equation}
				[v]\times [D] \xrightleftharpoons[F]{G}(S\cupdot[\eps' N])\times [D] \xrightleftharpoons[{\tilde{f}\times [D]}]{\tilde{g}\times [D]}[\tilde{a} - \eps N + \eps' N]\times [D]
				\tag{\ref{eq: lossy-code instance III}}
			\end{equation}
			to obtain $(x', y') \in [\tilde{a} - \eps N + \eps' N]\times [D]$. We let $F^*(f, x, y) = (f, x', y') \in Y_0$.

			Similarly, the evaluation of \eqref{eq: lossy-code instance III} on $(x, y)$ might be ill-behaved in the following cases. Suppose $x\in [v]$. If $(x, y)$ witnesses that $(G, F)$ is not a valid injection-surjection pair, then $S$ is a valid distinguisher for $\PRG_f$; in this case, letting $Y_3 := [K]\times [2N]\times [D]$, we map $F^*(f, x, y) := (\Comp^S(f, x, y), x, y) \in Y_3$. Otherwise, let $(x', y') := G(x, y)$, if $x'\in S$ but $x'$ witnesses that $(\tilde{g}, \tilde{f})$ is not a valid injection-surjection pair, then $x'$ is a \ref{item: SOL3} and we let $F^*(f, x, y) := \bot$. If both cases above do not happen, we define $F^*(f, x, y) := (f, x', y') \in Y_0$ as usual.
        \end{itemize}

		Let $Y$ be the disjoint union $Y := Y_0 \cupdot Y_1 \cupdot Y_2 \cupdot Y_3 \cupdot \{\bot\}$, then
		\[\frac{|Y|}{|X|} \le \frac{N'(2N-\eps' N)D + K\cdot 2\ell^2\cdot 2N \cdot D + 1}{N'\cdot 2N\cdot D} \le 1-\eps'/2 + (K\cdot 2\ell^2 + 1)/N' \le 1-\eps'/4.\]
		We have described a function $F^*: X \to Y$ and it remains to describe $G^*: Y\to X$. The $\lossycode$ instance $(F^*, G^*)$ reduces to a standard $\lossycode$ instance (of stretch $[2t]\rightleftharpoons[t]$) in decision tree depth $O(\log (1/\eps')) \le \polylog(N)$. Given an input in $Y$:\begin{itemize}
			\item If the input is $(f, x', y') \in Y_0$, then one of the above three cases happen for $f$; depending on this, we map $(x', y')$ back to $(x, y)$ via the appropriate surjection. For example, in case I: If $x' \not\in [v_i + \eps' n]$ then we let $x := x' + \eps' n$ and $y' := y$; otherwise we let $(x'', y) := G_i(x', y')$ and $x := g_i(x')$.
			\item If the input is $(\hat{f}, i, x, y)\in Y_1$, then we recover $f := \Decomp^{S_i}(\hat{f})$ and return $(f, x, y)$.
			\item If the input is $(\hat{f}, i, j, x, y) \in Y_2$, then we recover $f := \Decomp^{S_i\cap S_j}(\hat{f})$ and return $(f, x, y)$.
			\item If the input is $(\hat{f}, x, y)\in Y_3$, then we recover $f := \Decomp^S(\hat{f})$ and return $(f, x, y)$.
			\item We do not care about the value of $G^*(\bot)$; we map it to an arbitrary value.
		\end{itemize}
		Finally, we need to argue that this is a correct reduction. Let $(f, x, y) \in X$. If $F^*(f, x, y)\in Y_0$, then the injection-surjection pairs (\eqref{eq: lossy-code instance I}, \eqref{eq: lossy-code instance II}, \eqref{eq: lossy-code instance III}) are well-behaved, hence $G^*(F^*(f, x, y)) = (f, x, y)$; if $F^*(f, x, y)\in Y_1\cup Y_2\cup Y_3$, then the compression ($\Comp^O(f, x, y)$ for the suitable oracle $O$) works correctly, hence $G^*(F^*(f, x, y)) = (f, x, y)$ as well. Therefore, if $(f, x, y)$ is a solution of the $\lossycode$ instance, then $F^*(f, x, y) = \bot$. As discussed above, in this case, we will find a solution of $\textsc{Inclusion-Exclusion}$.
    \end{proof}

    \subsubsection{Conclusion: \texorpdfstring{$\Lossy$}{Lossy}-Completeness}
    Finally, it is easy to see that the problems we considered ($\AMGMLC$, Dual-$\AMGMLC$, and $\textsc{Inclusion-Exclusion}$) are at least as hard as $\lossycode$, hence they are $\Lossy$-complete.
    \begin{theorem}
        The following problems are $\Lossy$-complete:
        \begin{itemize}
            \item $c$-$\AMGMLC$ for every constant $c > 1$;
            \item $c$-dual-$\AMGMLC$ for every $0 < c < 2$;
            \item $\textsc{Inclusion-Exclusion}$.
        \end{itemize}
    \end{theorem}

\section{Dense Linear Ordering}\label{sec: DLO}

    Finally, we show that the dense version of the Linear Ordering Principle reduces to $\lossycode$. The Linear Ordering Principle was recently studied by Korten and Pitassi~\cite{KortenP24}, where they showed that $\Avoid$ reduces to Linear Ordering. In the $\TFNP$ world, we show the opposite reduction and prove that Dense Linear Ordering reduces to Lossy Code. This problem has been studied in proof complexity \cite{Riis01, AtseriasD08, Gryaznov19, DBLP:conf/focs/ConnerydRNPR23} as well.

    \begin{description}
        \item[$\DLO$.]
        The input consists of the descriptions of a linear ordering $\prec$ over $N$ elements and a \emph{median function} $\med: [N]\times [N] \to [N]$. Without loss of generality, we may assume that for $x\ne y \in [N]$, exactly one of $(x\prec y)$ and $(y\prec x)$ is true, and that $\med(x, y) = \med(y, x)$. (That is, $\prec$ is represented by a string of $\binom{N}{2}$ bits and $\med$ is represented by a list of $\binom{n}{2}$ elements in $[N]$.) A solution is one of the following.
        \begin{description}
            \item[s1.] $x, y, z \in [N]$ such that $x\prec y$, $y\prec z$, and $z\prec x$; or \mylabel{Transitivity Violation}{item: (Transitivity Violation)}
            \item[s2.] $x, y\in [N]$ such that $x\prec y$, but neither $x\prec \med(x, y)$ nor $\med(x, y)\prec y$.
            
            \mylabel{Invalid Median}{item: (Invalid Median)}
        \end{description}
    \end{description}

    \def\Left{\mathsf{left}}
    \def\Right{\mathsf{right}}
    \def\NonTransitivity{{\bf (Violation of transitivity)}\xspace}
    \def\BadMedian{{\bf (Bad median)}\xspace}
    \begin{theorem}
        $\DLO$ reduces to $\lossycode$.
    \end{theorem}
    \begin{proof}
        Let $(\prec, \med)$ be an input of $\DLO$ over a universe $\calU$ of size $N$. %
        
        Fix some arbitrary $l_0, r_0$ in advance such that $l_0 \prec r_0$, and let $\ell := 4\log N$. Consider the following $\lossycode$ instance $(f, g)$ between $\calU$ and $\{\LL, \RR\}^{\le \ell}$. 
        \begin{figure}[H]
            \begin{minipage}{0.50\linewidth}
                \begin{algorithm}[H]
                    \caption{$f: \calU \to \{\LL, \RR\}^{\le \ell}$}
                    \begin{algorithmic}[1]
                        \Function{$f$}{$m$}
                        \State $\sigma\gets$ the empty string
                        \For {$i\gets 1$ to $\ell$}
                            \State {$m_{i-1} \gets \med(l_{i-1}, r_{i-1})$}
                            \If {$m_{i-1} = m$}
                                \State \Return {$\sigma$}
                            \ElsIf {$m \prec m_{i-1}$}
                                \State $(l_i, r_i)\gets (l_{i-1}, m_{i-1})$, $\sigma \gets \sigma \circ \LL$
                            \Else \Comment{$m_{i-1} \prec m$}
                                \State $(l_i, r_i)\gets (m_{i-1}, r_{i-1})$, $\sigma \gets \sigma \circ \RR$
                            \EndIf
                        \EndFor
                        \State \Return $\sigma$ \Comment{regardless of whether we have arrived at $m$}
                        \EndFunction
                    \end{algorithmic}
                \end{algorithm}
            \end{minipage}
            \hfill
            \begin{minipage}{0.40\linewidth}
                \begin{algorithm}[H]
                    \caption{$g: \{\LL, \RR\}^{\le \ell} \to \calU$}
                    \begin{algorithmic}[1]
                        \Function{$g$}{$\sigma$}
                            \For {$i\gets 1$ to $|\sigma|$}
                                \State {$m_{i-1} \gets \med(l_{i-1}, r_{i-1})$}
                                \If {$\sigma_i = \LL$}
                                    \State $(l_i, r_i)\gets (l_{i-1}, m_{i-1})$
                                \Else \Comment{$\sigma_i = \RR$}
                                    \State $(l_i, r_i)\gets (m_{i-1}, r_{i-1})$
                                \EndIf
                            \EndFor
                            \State \Return $\med(l_{|\sigma|}, r_{|\sigma|})$
                        \EndFunction
                    \end{algorithmic}
                \end{algorithm}
            \end{minipage}
        \end{figure}
    
        In the following discussion, we assume that throughout the executions of $f$ and $g$, we always have $l_i \prec r_i$ for each valid $i$. This is because otherwise we can find an \ref{item: (Invalid Median)}: Let $i$ be the smallest index such that $l_i\not\prec r_i$, then $i\ge 1$ and $l_{i-1}\prec r_{i-1}$. Either we have $(l_i, r_i) = (l_{i-1}, m_{i-1})$, or we have $(l_i, r_i) = (m_{i-1}, r_{i-1})$. In both cases, $(l_{i-1}, r_{i-1})$ has an \ref{item: (Invalid Median)}.
        
        Consider the following scenario. Suppose that $g(\sigma) = v$, but during the execution of $g(\sigma)$, we encountered some interval $(l_i, r_i)$ such that either ($l_i \not\prec v$) or ($v \not\prec r_i$). We argue that in this scenario, it is easy to find a solution of $(\prec, \med)$. Suppose $l_i\not\prec v$; the case that $v\not\prec r_i$ can be handled symmetrically.
        \begin{itemize}
            \item If $l_{|\sigma|} \not \prec v$ then, $(l_{|\sigma|}, r_{|\sigma|})$ has an \ref{item: (Invalid Median)}.
            \item Otherwise there is an index $j \in [i, |\sigma|)$ such that $l_j\not\prec v$ but $l_{j+1} \prec v$. It follows that $l_j \ne l_{j+1}$, which implies that $l_{j+1} = \med(l_j, r_j)$.
            \begin{itemize}
                \item If $l_j \not\prec l_{j+1}$ then we have an \ref{item: (Invalid Median)} $(l_j, r_j)$.
                \item Otherwise, since $l_{j+1}\prec v$ and $l_j\prec l_{j+1}$, we have $v\ne l_j$. Since $l_j\not\prec v$, we have $v\prec l_j$. Now we have a \ref{item: (Transitivity Violation)} $(l_j, l_{j+1}, v)$ where $v\prec l_j$, $l_j\prec l_{j+1}$, and $l_{j+1} \prec v$.
            \end{itemize}
        \end{itemize}
    
        Now we prove the correctness of our reduction. That is, given any string $\sigma \in \{\LL, \RR\}^{\le \ell}$ such that $f(g(\sigma)) \ne \sigma$, we can find a solution for $\DLO$. Let $v := g(\sigma)$ and $\sigma' := f(v)$, then $\sigma' \ne \sigma$. Let $i$ be the smallest index such that $\sigma'_i \ne \sigma_i$. In particular, if $\sigma'$ is a prefix of $\sigma$ then we define $i := |\sigma'| + 1$. (Note that it cannot be the case that $\sigma$ is a prefix of $\sigma'$, as otherwise the execution of $f(v)$ would have reached the element $v$ in the $|\sigma|$-th step and returned $\sigma$ instead of $\sigma'$.) It follows that the first $i-1$ rounds of $f(v)$ and $g(\sigma)$ produce the same intervals $\{(l_j, r_j)\}_{0\le j < i}$. Assume that $\sigma_i = \LL$; the case that $\sigma_i = \RR$ is symmetric. Let $(l^g_i, r^g_i)$ denote the interval $(l_i, r_i)$ in the execution of $g(\sigma)$, then $l^g_i = l_{i-1}$ and $r^g_i = m_{i-1} = \med(l_{i-1}, r_{i-1})$.
        \begin{itemize}
            \item If $i = |\sigma'| + 1$, then $m_{i-1} = v = g(\sigma)$.
            \item Otherwise, $\sigma'_i = \RR$ and hence $m_{i-1} \prec v = g(\sigma)$.
        \end{itemize}
        In either case, we have found some $(l^g_i, r^g_i)$ during the execution of $g(\sigma)$ such that $v\not\prec r^g_i$. By the previous discussion, we can find a solution of $(\prec, \med)$.
    \end{proof}

\section{Open Problems}

We end with some future directions. The main problem left open by this work is to exhibit a natural problem in $\TFZPP^{dt}$ which is not reducible to $\lossycode$; we conjecture that $\ILP$ is such a problem. Some additional open questions are the following:

\begin{itemize}
    \item Find a $\TFZPP$ upper bound for $\textsc{Bertrand-Chebyshev}$, i.e., a natural problem in $\TFZPP$ to which $\textsc{Bertrand-Chebyshev}$ reduces. The best upper bound we are aware of is only $\Lossy^{\textsc{Factoring}}$~\cite{ParisWW88, DBLP:conf/coco/Korten22}, which is in $\Lossy^\PPA$ and $\Lossy^\PWPP$ under the generalized Riemann Hypothesis~\cite{Jerabek16}. Unfortunately, none of these upper bound classes are in $\TFZPP$.
    \item Find a $\TFZPP$ upper bound for the following problem, which we call \textsc{Razborov-Smolensky} \cite{Razborov87, Smolensky87}: The input is an $\AC^0[2]$ circuit $C$ of depth $d$ and size at most $2^{n^{1/10d}}$ that attempts to compute $\mathrm{MAJ}$ (the Majority function), and the goal is to output an instance $x\in \{0, 1\}^n$ such that $C(x) \ne \mathrm{MAJ}(x)$. Since $\mathrm{MAJ}$ is average-case hard against such circuits, this problem sits in $\TFZPP$. This problem is trivially solvable in deterministic quasi-polynomial time (note that the na\"ive algorithm runs in $2^{O(n)}$ time while the input size is $2^{n^{\Omega(1)}}$), hence we are interested in the regime where only polynomial-time reductions are allowed. We are not aware of any syntactic subclass of $\TFZPP$ that contains this problem. %
    \item Is $\lossycodep\in \PPAD$? Observe that $\lossycodep$ is equivalent to $\eolLong$ with half of the vertices designated as distinguished sources. 
    The simple argument showing $\lossycode\in \PPADS$, that is, output the solution returned by the $\PPAD$ solver on the original graph, does not work here, since the $\PPAD$ solver may return one of the distinguished sources.
    Recently, Goldberg and Hollender~\cite{HairyBallPPADComplete} gave an involved proof that $\textsc{End-Of-Line}$ with $k$ distinguished sources belongs to $\PPAD$ when $k=\polylog(N)$, where $N$ denotes the size of the graph.
    However, their argument does not extend to the case that $k=\poly(N)$ since the $\PPAD$ instance they constructed has size $N^{\Omega(k)}$.
    \item Can  $\BEC$ be separated from $\lossycodep$?
    \item What is the relationship between $\DLO$ and $\lossycodep$?
\end{itemize}

\DeclareUrlCommand{\Doi}{\urlstyle{tt}}
\renewcommand{\path}[1]{\small\Doi{#1}}
\renewcommand{\url}[1]{\href{#1}{\Doi{#1}}}

\section*{Acknowledgments} 

We thank Robert Robere, Yuhao Li, and Ben Davis for extensive discussions about the Nephew problem and $\TFZPP$. As well, we thank the reviewers for suggestions which improved the presentation of this paper. Noah Fleming was supported by an NSERC Discovery grant and the Swedish Research Council under grant number 2025-06762. Stefan Grosser was supported by the NSERC CGS D fellowship. Siddhartha Jain was supported by Scott Aaronson's Berkeley CIQC grant and an Amazon AI Fellowship. Jiawei Li was supported by Scott Aaronson's Open Philanthropy grant. 
Morgan Shirley was supported by an NSERC grant and by Knut and Alice Wallenberg grant KAW 2023.0116. Weiqiang Yuan was supported by the Swiss State Secretariat for Education, Research and Innovation (SERI) under contract number MB22.00026.

\bibliographystyle{alphaurl}
{\bibliography{refs}}

\newcommand{\etalchar}[1]{$^{#1}$}
\begin{thebibliography}{CdRN{\etalchar{+}}23}

\bibitem[AD08]{AtseriasD08}
Albert Atserias and V{\'{\i}}ctor Dalmau.
\newblock A combinatorial characterization of resolution width.
\newblock {\em J. Comput. Syst. Sci.}, 74(3):323--334, 2008.
\newblock \href {https://doi.org/10.1016/J.JCSS.2007.06.025}
  {\path{doi:10.1016/J.JCSS.2007.06.025}}.

\bibitem[AKS04]{AKS04}
Manindra Agrawal, Neeraj Kayal, and Nitin Saxena.
\newblock {PRIMES} is in {$\P$}.
\newblock {\em Annals of Mathematics}, 160(2):781--793, 2004.
\newblock \href {https://doi.org/10.4007/annals.2004.160.781}
  {\path{doi:10.4007/annals.2004.160.781}}.

\bibitem[AM25]{AgarwalaM25}
Aryan Agarwala and Ian Mertz.
\newblock Bipartite matching is in catalytic logspace.
\newblock In {\em {FOCS}}, 2025.
\newblock To appear.
\newblock \href {https://arxiv.org/abs/2504.09991} {\path{arXiv:2504.09991}}.

\bibitem[AT14]{DBLP:journals/tocl/AtseriasT14}
Albert Atserias and Neil Thapen.
\newblock The ordering principle in a fragment of approximate counting.
\newblock {\em {ACM} Trans. Comput. Log.}, 15(4):29:1--29:11, 2014.
\newblock \href {https://doi.org/10.1145/2629555} {\path{doi:10.1145/2629555}}.

\bibitem[BCE{\etalchar{+}}98]{beame1995relative}
Paul Beame, Stephen~A. Cook, Jeff Edmonds, Russell Impagliazzo, and Toniann
  Pitassi.
\newblock The relative complexity of {$\NP$} search problems.
\newblock {\em J. Comput. Syst. Sci.}, 57(1):3--19, 1998.
\newblock \href {https://doi.org/10.1006/JCSS.1998.1575}
  {\path{doi:10.1006/JCSS.1998.1575}}.

\bibitem[BCK{\etalchar{+}}14]{DBLP:conf/stoc/BuhrmanCKLS14}
Harry Buhrman, Richard Cleve, Michal Kouck{\'{y}}, Bruno Loff, and Florian
  Speelman.
\newblock Computing with a full memory: catalytic space.
\newblock In {\em {STOC}}, pages 857--866. {ACM}, 2014.
\newblock \href {https://doi.org/10.1145/2591796.2591874}
  {\path{doi:10.1145/2591796.2591874}}.

\bibitem[BFI23]{BussFI23}
Sam Buss, Noah Fleming, and Russell Impagliazzo.
\newblock {$\TFNP$} characterizations of proof systems and monotone circuits.
\newblock In Yael~Tauman Kalai, editor, {\em 14th Innovations in Theoretical
  Computer Science Conference, {ITCS} 2023, January 10-13, 2023, MIT,
  Cambridge, Massachusetts, {USA}}, volume 251 of {\em LIPIcs}, pages
  30:1--30:40. Schloss Dagstuhl - Leibniz-Zentrum f{\"{u}}r Informatik, 2023.
\newblock \href {https://doi.org/10.4230/LIPICS.ITCS.2023.30}
  {\path{doi:10.4230/LIPICS.ITCS.2023.30}}.

\bibitem[BG01]{BonetG01}
Maria~Luisa Bonet and Nicola Galesi.
\newblock Optimality of size-width tradeoffs for resolution.
\newblock {\em Comput. Complex.}, 10(4):261--276, 2001.
\newblock \href {https://doi.org/10.1007/S000370100000}
  {\path{doi:10.1007/S000370100000}}.

\bibitem[BGG01]{Borger01}
Egon B{\"o}rger, Erich Gr{\"a}del, and Yuri Gurevich.
\newblock {\em The classical decision problem}.
\newblock Springer Science \& Business Media, 2001.

\bibitem[BGSD25]{BGS25}
Huck Bennett, Surendra Ghentiyala, and Noah Stephens-Davidowitz.
\newblock The more the merrier! {O}n total coding and lattice problems and the
  complexity of finding multicollisions.
\newblock In {\em 16th {I}nnovations in {T}heoretical {C}omputer {S}cience
  {C}onference}, volume 325 of {\em LIPIcs. Leibniz Int. Proc. Inform.}, pages
  Art. No. 14, 22. Schloss Dagstuhl. Leibniz-Zent. Inform., Wadern, 2025.
\newblock \href {https://doi.org/10.4230/lipics.itcs.2025.14}
  {\path{doi:10.4230/lipics.itcs.2025.14}}.

\bibitem[BHP01]{BHP01}
R.~C. Baker, G.~Harman, and J.~Pintz.
\newblock The difference between consecutive primes. {II}.
\newblock {\em Proc. London Math. Soc. (3)}, 83(3):532--562, 2001.
\newblock \href {https://doi.org/10.1112/plms/83.3.532}
  {\path{doi:10.1112/plms/83.3.532}}.

\bibitem[BJ12]{buss2012propositional}
Samuel~R Buss and Alan~S Johnson.
\newblock Propositional proofs and reductions between np search problems.
\newblock {\em Annals of Pure and Applied Logic}, 163(9):1163--1182, 2012.

\bibitem[BKT14]{BussKT14}
Samuel~R. Buss, Leszek~Aleksander Kołodziejczyk, and Neil Thapen.
\newblock Fragments of approximate counting.
\newblock {\em J. Symb. Log.}, 79(2):496--525, 2014.
\newblock \href {https://doi.org/10.1017/JSL.2013.37}
  {\path{doi:10.1017/JSL.2013.37}}.

\bibitem[BO06]{buresh2006tfnp}
Joshua Buresh-Oppenheim.
\newblock On the {$\TFNP$} complexity of factoring.
\newblock {\em Unpublished}, 2006.

\bibitem[Bus97]{BussSurvey}
Samuel~R. Buss.
\newblock Bounded arithmetic, cryptography and complexity.
\newblock {\em Theoria}, 63(3):147--167, 1997.
\newblock \href {https://doi.org/10.1111/j.1755-2567.1997.tb00745.x}
  {\path{doi:10.1111/j.1755-2567.1997.tb00745.x}}.

\bibitem[CD06]{DBLP:conf/focs/ChenD06}
Xi~Chen and Xiaotie Deng.
\newblock Settling the complexity of two-player {N}ash equilibrium.
\newblock In {\em {FOCS}}, pages 261--272. {IEEE} Computer Society, 2006.
\newblock \href {https://doi.org/10.1109/FOCS.2006.69}
  {\path{doi:10.1109/FOCS.2006.69}}.

\bibitem[CdRN{\etalchar{+}}23]{DBLP:conf/focs/ConnerydRNPR23}
Jonas Conneryd, Susanna~F. de~Rezende, Jakob Nordstr{\"{o}}m, Shuo Pang, and
  Kilian Risse.
\newblock Graph colouring is hard on average for {P}olynomial {C}alculus and
  {N}ullstellensatz.
\newblock In {\em {FOCS}}, pages 1--11. {IEEE}, 2023.
\newblock \href {https://doi.org/10.1109/FOCS57990.2023.00007}
  {\path{doi:10.1109/FOCS57990.2023.00007}}.

\bibitem[CHLR23]{ChenHLR23}
Yeyuan Chen, Yizhi Huang, Jiatu Li, and Hanlin Ren.
\newblock Range avoidance, remote point, and hard partial truth table via
  satisfying-pairs algorithms.
\newblock In {\em {STOC}}, pages 1058--1066. {ACM}, 2023.
\newblock \href {https://doi.org/10.1145/3564246.3585147}
  {\path{doi:10.1145/3564246.3585147}}.

\bibitem[CHR24]{CHR24}
Lijie Chen, Shuichi Hirahara, and Hanlin Ren.
\newblock Symmetric exponential time requires near-maximum circuit size.
\newblock In {\em S{TOC}'24---{P}roceedings of the 56th {A}nnual {ACM}
  {S}ymposium on {T}heory of {C}omputing}, pages 1990--1999. ACM, New York,
  [2024] \copyright 2024.
\newblock \href {https://doi.org/10.1145/3618260.3649624}
  {\path{doi:10.1145/3618260.3649624}}.

\bibitem[CJSW24]{CJSW24}
Lijie Chen, Ce~Jin, Rahul Santhanam, and Ryan Williams.
\newblock Constructive separations and their consequences.
\newblock {\em TheoretiCS}, 3, 2024.
\newblock \href {https://doi.org/10.46298/THEORETICS.24.3}
  {\path{doi:10.46298/THEORETICS.24.3}}.

\bibitem[CK98]{DBLP:journals/jsyml/ChiariK98}
Mario Chiari and Jan Kraj{\'{\i}}cek.
\newblock Witnessing functions in bounded arithmetic and search problems.
\newblock {\em J. Symb. Log.}, 63(3):1095--1115, 1998.
\newblock \href {https://doi.org/10.2307/2586729} {\path{doi:10.2307/2586729}}.

\bibitem[CLMP25]{CookLMP25}
James Cook, Jiatu Li, Ian Mertz, and Edward Pyne.
\newblock The structure of catalytic space: Capturing randomness and time via
  compression.
\newblock In {\em {STOC}}, pages 554--564. {ACM}, 2025.
\newblock \href {https://doi.org/10.1145/3717823.3718112}
  {\path{doi:10.1145/3717823.3718112}}.

\bibitem[CLO{\etalchar{+}}23]{ChenLORS23}
Lijie Chen, Zhenjian Lu, Igor~C. Oliveira, Hanlin Ren, and Rahul Santhanam.
\newblock Polynomial-time pseudodeterministic construction of primes.
\newblock In {\em 64th {IEEE} Annual Symposium on Foundations of Computer
  Science, {FOCS} 2023, Santa Cruz, CA, USA, November 6-9, 2023}, pages
  1261--1270. {IEEE}, 2023.
\newblock \href {https://doi.org/10.1109/FOCS57990.2023.00074}
  {\path{doi:10.1109/FOCS57990.2023.00074}}.

\bibitem[CLO24]{CLO24}
Lijie Chen, Jiatu Li, and Igor~C. Oliveira.
\newblock Reverse mathematics of complexity lower bounds.
\newblock In {\em {FOCS}}, pages 505--527. {IEEE}, 2024.
\newblock \href {https://doi.org/10.1109/FOCS61266.2024.00040}
  {\path{doi:10.1109/FOCS61266.2024.00040}}.

\bibitem[CLW20]{CLW20}
Lijie Chen, Xin Lyu, and R.~Ryan Williams.
\newblock Almost-everywhere circuit lower bounds from non-trivial
  derandomization.
\newblock In {\em {FOCS}}, pages 1--12. {IEEE}, 2020.
\newblock \href {https://doi.org/10.1109/FOCS46700.2020.00009}
  {\path{doi:10.1109/FOCS46700.2020.00009}}.

\bibitem[Cra36]{cramer1936order}
Harald Cram{\'e}r.
\newblock On the order of magnitude of the difference between consecutive prime
  numbers.
\newblock {\em Acta Arithmetica}, 2:23--46, 1936.
\newblock URL: \url{http://eudml.org/doc/205441}.

\bibitem[CTW23]{CTW23}
Lijie Chen, Roei Tell, and Ryan Williams.
\newblock Derandomization vs refutation: {A} unified framework for
  characterizing derandomization.
\newblock In {\em {FOCS}}, pages 1008--1047. {IEEE}, 2023.
\newblock \href {https://doi.org/10.1109/FOCS57990.2023.00062}
  {\path{doi:10.1109/FOCS57990.2023.00062}}.

\bibitem[DGP09]{DGC09}
Constantinos Daskalakis, Paul~W. Goldberg, and Christos~H. Papadimitriou.
\newblock The complexity of computing a {N}ash equilibrium.
\newblock {\em SIAM J. Comput.}, 39(1):195--259, 2009.
\newblock \href {https://doi.org/10.1137/070699652}
  {\path{doi:10.1137/070699652}}.

\bibitem[DR23]{DavisR23}
Ben Davis and Robert Robere.
\newblock Colourful {$\TFNP$} and propositional proofs.
\newblock In Amnon Ta{-}Shma, editor, {\em 38th Computational Complexity
  Conference, {CCC} 2023, July 17-20, 2023, Warwick, {UK}}, volume 264 of {\em
  LIPIcs}, pages 36:1--36:21. Schloss Dagstuhl - Leibniz-Zentrum f{\"{u}}r
  Informatik, 2023.
\newblock \href {https://doi.org/10.4230/LIPICS.CCC.2023.36}
  {\path{doi:10.4230/LIPICS.CCC.2023.36}}.

\bibitem[FGH{\etalchar{+}}24]{Foletal24}
Luk\'{a}\v{s} Folwarczn\'{y}, Mika G\"{o}\"{o}s, Pavel Hub\'{a}\v{c}ek, Gilbert
  Maystre, and Weiqiang Yuan.
\newblock One-way functions vs. {$\TFNP$}: simpler and improved.
\newblock In {\em 15th {I}nnovations in {T}heoretical {C}omputer {S}cience
  {C}onference}, volume 287 of {\em LIPIcs. Leibniz Int. Proc. Inform.}, pages
  Art. No. 50, 14. Schloss Dagstuhl. Leibniz-Zent. Inform., Wadern, 2024.
\newblock \href {https://doi.org/10.4230/lipics.itcs.2024.50}
  {\path{doi:10.4230/lipics.itcs.2024.50}}.

\bibitem[FGHS23]{FearnleyGHS23}
John Fearnley, Paul Goldberg, Alexandros Hollender, and Rahul Savani.
\newblock The complexity of gradient descent: {$\CLS = \PPAD \cap \PLS$}.
\newblock {\em J. {ACM}}, 70(1):7:1--7:74, 2023.
\newblock \href {https://doi.org/10.1145/3568163} {\path{doi:10.1145/3568163}}.

\bibitem[FGMS20]{FearnleyGMS20}
John Fearnley, Spencer Gordon, Ruta Mehta, and Rahul Savani.
\newblock Unique end of potential line.
\newblock {\em J. Comput. Syst. Sci.}, 114:1--35, 2020.
\newblock \href {https://doi.org/10.1016/J.JCSS.2020.05.007}
  {\path{doi:10.1016/J.JCSS.2020.05.007}}.

\bibitem[FGPR24]{FlemingGPR24}
Noah Fleming, Stefan Grosser, Toniann Pitassi, and Robert Robere.
\newblock Black-box {$\PPP$} is not {T}uring-closed.
\newblock In Bojan Mohar, Igor Shinkar, and Ryan O'Donnell, editors, {\em
  Proceedings of the 56th Annual {ACM} Symposium on Theory of Computing, {STOC}
  2024, Vancouver, BC, Canada, June 24-28, 2024}, pages 1405--1414. {ACM},
  2024.
\newblock \href {https://doi.org/10.1145/3618260.3649769}
  {\path{doi:10.1145/3618260.3649769}}.

\bibitem[FIM25]{FlemingMD25}
Noah Fleming, Deniz Imrek, and Christophe Marciot.
\newblock Provably total functions in the polynomial hierarchy.
\newblock In Srikanth Srinivasan, editor, {\em 40th Computational Complexity
  Conference, {CCC} 2025, August 5-8, 2025, Toronto, Canada}, volume 339 of
  {\em LIPIcs}, pages 28:1--28:40. Schloss Dagstuhl - Leibniz-Zentrum f{\"{u}}r
  Informatik, 2025.
\newblock \href {https://doi.org/10.4230/LIPICS.CCC.2025.28}
  {\path{doi:10.4230/LIPICS.CCC.2025.28}}.

\bibitem[FKP19]{FlemingKP19}
Noah Fleming, Pravesh Kothari, and Toniann Pitassi.
\newblock Semialgebraic proofs and efficient algorithm design.
\newblock {\em Found. Trends Theor. Comput. Sci.}, 14(1-2):1--221, 2019.
\newblock \href {https://doi.org/10.1561/0400000086}
  {\path{doi:10.1561/0400000086}}.

\bibitem[GG11]{GatG11}
Eran Gat and Shafi Goldwasser.
\newblock Probabilistic search algorithms with unique answers and their
  cryptographic applications.
\newblock {\em Electron. Colloquium Comput. Complex.}, {TR11-136}, 2011.
\newblock URL: \url{https://eccc.weizmann.ac.il/report/2011/136/}.

\bibitem[GGM86]{DBLP:journals/jacm/GoldreichGM86}
Oded Goldreich, Shafi Goldwasser, and Silvio Micali.
\newblock How to construct random functions.
\newblock {\em J. {ACM}}, 33(4):792--807, 1986.
\newblock \href {https://doi.org/10.1145/6490.6503}
  {\path{doi:10.1145/6490.6503}}.

\bibitem[GGNS23]{DBLP:conf/approx/GajulapalliGNS23}
Karthik Gajulapalli, Alexander Golovnev, Satyajeet Nagargoje, and Sidhant
  Saraogi.
\newblock Range avoidance for constant depth circuits: Hardness and algorithms.
\newblock In {\em {APPROX/RANDOM}}, volume 275 of {\em LIPIcs}, pages
  65:1--65:18. Schloss Dagstuhl - Leibniz-Zentrum f{\"{u}}r Informatik, 2023.
\newblock \href {https://doi.org/10.4230/LIPICS.APPROX/RANDOM.2023.65}
  {\path{doi:10.4230/LIPICS.APPROX/RANDOM.2023.65}}.

\bibitem[GH21]{HairyBallPPADComplete}
Paul~W. Goldberg and Alexandros Hollender.
\newblock The hairy ball problem is {$\PPAD$}-complete.
\newblock {\em Journal of Computer and System Sciences}, 122:34--62, 2021.
\newblock \href {https://doi.org/10.1016/j.jcss.2021.05.004}
  {\path{doi:10.1016/j.jcss.2021.05.004}}.

\bibitem[GHJ{\etalchar{+}}22]{GHJ+22}
Mika Göös, Alexandros Hollender, Siddhartha Jain, Gilbert Maystre, William
  Pires, Robert Robere, and Ran Tao.
\newblock Separations in proof complexity and {$\TFNP$}.
\newblock In {\em 2022 IEEE 63rd Annual Symposium on Foundations of Computer
  Science (FOCS)}, pages 1150--1161, 2022.
\newblock \href {https://doi.org/10.1109/FOCS54457.2022.00111}
  {\path{doi:10.1109/FOCS54457.2022.00111}}.

\bibitem[GHJ{\etalchar{+}}24]{GoosHJMPRT24}
Mika G{\"{o}}{\"{o}}s, Alexandros Hollender, Siddhartha Jain, Gilbert Maystre,
  William Pires, Robert Robere, and Ran Tao.
\newblock Further collapses in {$\TFNP$}.
\newblock {\em {SIAM} J. Comput.}, 53(3):573--587, 2024.
\newblock \href {https://doi.org/10.1137/22M1498346}
  {\path{doi:10.1137/22M1498346}}.

\bibitem[GKRS19]{GoosKRS19}
Mika G{\"{o}}{\"{o}}s, Pritish Kamath, Robert Robere, and Dmitry Sokolov.
\newblock Adventures in monotone complexity and {$\TFNP$}.
\newblock In Avrim Blum, editor, {\em 10th Innovations in Theoretical Computer
  Science Conference, {ITCS} 2019, January 10-12, 2019, San Diego, California,
  {USA}}, volume 124 of {\em LIPIcs}, pages 38:1--38:19. Schloss Dagstuhl -
  Leibniz-Zentrum f{\"{u}}r Informatik, 2019.
\newblock \href {https://doi.org/10.4230/LIPICS.ITCS.2019.38}
  {\path{doi:10.4230/LIPICS.ITCS.2019.38}}.

\bibitem[GL25]{GhentiyalaLi25}
Surendra Ghentiyala and Zeyong Li.
\newblock Hierarchies within {$\TFNP$:} building blocks and collapses.
\newblock {\em CoRR}, 2025.
\newblock \href {https://arxiv.org/abs/2507.21550} {\path{arXiv:2507.21550}}.

\bibitem[GLW25]{DBLP:journals/toct/GuruswamiLW25}
Venkatesan Guruswami, Xin Lyu, and Xiuhan Wang.
\newblock Range avoidance for low-depth circuits and connections to
  pseudorandomness.
\newblock {\em {ACM} Trans. Comput. Theory}, 17(2):14:1--14:23, 2025.
\newblock \href {https://doi.org/10.1145/3718745} {\path{doi:10.1145/3718745}}.

\bibitem[GP18]{GoldbergP18}
Paul~W. Goldberg and Christos~H. Papadimitriou.
\newblock Towards a unified complexity theory of total functions.
\newblock {\em J. Comput. Syst. Sci.}, 94:167--192, 2018.
\newblock \href {https://doi.org/10.1016/J.JCSS.2017.12.003}
  {\path{doi:10.1016/J.JCSS.2017.12.003}}.

\bibitem[Gry19]{Gryaznov19}
Svyatoslav Gryaznov.
\newblock Notes on resolution over linear equations.
\newblock In {\em {CSR}}, volume 11532 of {\em Lecture Notes in Computer
  Science}, pages 168--179. Springer, 2019.
\newblock \href {https://doi.org/10.1007/978-3-030-19955-5\_15}
  {\path{doi:10.1007/978-3-030-19955-5\_15}}.

\bibitem[Han04]{Hanika-PhD}
Jiří Hanika.
\newblock {\em Search Problems and Bounded Arithmetic}.
\newblock PhD thesis, Charles University, Prague, 2004.

\bibitem[HKKS20]{Hubacek20}
Pavel Hub\'{a}\v{c}ek, Chethan Kamath, Karel Kr\'{a}l, and Veronika
  Sl\'{\i}vov\'{a}.
\newblock On average-case hardness in {$\TFNP$} from one-way functions.
\newblock In {\em Theory of cryptography. {P}art {III}}, volume 12552 of {\em
  Lecture Notes in Comput. Sci.}, pages 614--638. Springer, Cham, [2020]
  \copyright 2020.
\newblock \href {https://doi.org/10.1007/978-3-030-64381-2\_22}
  {\path{doi:10.1007/978-3-030-64381-2\_22}}.

\bibitem[HKT24]{HubacekKT24}
Pavel Hub{\'{a}}cek, Erfan Khaniki, and Neil Thapen.
\newblock {$\TFNP$} intersections through the lens of feasible disjunction.
\newblock In Venkatesan Guruswami, editor, {\em 15th Innovations in Theoretical
  Computer Science Conference, {ITCS} 2024, January 30 to February 2, 2024,
  Berkeley, CA, {USA}}, volume 287 of {\em LIPIcs}, pages 63:1--63:24. Schloss
  Dagstuhl - Leibniz-Zentrum f{\"{u}}r Informatik, 2024.
\newblock \href {https://doi.org/10.4230/LIPICS.ITCS.2024.63}
  {\path{doi:10.4230/LIPICS.ITCS.2024.63}}.

\bibitem[HL22]{HopkinsL22}
Max Hopkins and Ting{-}Chun Lin.
\newblock Explicit lower bounds against {$\Omega(n)$}-rounds of sum-of-squares.
\newblock In {\em 63rd {IEEE} Annual Symposium on Foundations of Computer
  Science, {FOCS} 2022, Denver, CO, USA, October 31 - November 3, 2022}, pages
  662--673. {IEEE}, 2022.
\newblock \href {https://doi.org/10.1109/FOCS54457.2022.00069}
  {\path{doi:10.1109/FOCS54457.2022.00069}}.

\bibitem[HV25]{HV25}
Edward~A. Hirsch and Ilya Volkovich.
\newblock Upper and lower bounds for the linear ordering principle.
\newblock {\em CoRR}, 2025.
\newblock \href {https://arxiv.org/abs/2503.19188} {\path{arXiv:2503.19188}}.

\bibitem[HY20]{HubacekY20}
Pavel Hub{\'{a}}cek and Eylon Yogev.
\newblock Hardness of continuous local search: Query complexity and
  cryptographic lower bounds.
\newblock {\em {SIAM} J. Comput.}, 49(6):1128--1172, 2020.
\newblock \href {https://doi.org/10.1137/17M1118014}
  {\path{doi:10.1137/17M1118014}}.

\bibitem[IW97]{ImpagliazzoW97}
Russell Impagliazzo and Avi Wigderson.
\newblock {$\P = \BPP$} if {$\E$} requires exponential circuits: Derandomizing
  the {XOR} lemma.
\newblock In {\em {STOC}}, pages 220--229. {ACM}, 1997.
\newblock \href {https://doi.org/10.1145/258533.258590}
  {\path{doi:10.1145/258533.258590}}.

\bibitem[Je{\v{r}}04]{Jerabek04}
Emil Je{\v{r}}{\'{a}}bek.
\newblock Dual weak pigeonhole principle, {Boolean} complexity, and
  derandomization.
\newblock {\em Ann. Pure Appl. Log.}, 129(1-3):1--37, 2004.
\newblock \href {https://doi.org/10.1016/j.apal.2003.12.003}
  {\path{doi:10.1016/j.apal.2003.12.003}}.

\bibitem[Je{\v{r}}05]{Jerabek-phd}
Emil Je{\v{r}}{\'{a}}bek.
\newblock {\em Weak pigeonhole principle and randomized computation}.
\newblock PhD thesis, Charles University in Prague, 2005.

\bibitem[Je{\v{r}}07a]{Jerabek07}
Emil Je{\v{r}}{\'{a}}bek.
\newblock Approximate counting in bounded arithmetic.
\newblock {\em J. Symb. Log.}, 72(3):959--993, 2007.
\newblock \href {https://doi.org/10.2178/JSL/1191333850}
  {\path{doi:10.2178/JSL/1191333850}}.

\bibitem[Je{\v{r}}07b]{Jerabek-independence}
Emil Je{\v{r}}{\'{a}}bek.
\newblock On independence of variants of the weak pigeonhole principle.
\newblock {\em J. Log. Comput.}, 17(3):587--604, 2007.
\newblock \href {https://doi.org/10.1093/LOGCOM/EXM017}
  {\path{doi:10.1093/LOGCOM/EXM017}}.

\bibitem[Je{\v{r}}16]{Jerabek16}
Emil Je{\v{r}}{\'{a}}bek.
\newblock Integer factoring and modular square roots.
\newblock {\em J. Comput. Syst. Sci.}, 82(2):380--394, 2016.
\newblock \href {https://doi.org/10.1016/J.JCSS.2015.08.001}
  {\path{doi:10.1016/J.JCSS.2015.08.001}}.

\bibitem[JPY88]{DBLP:journals/jcss/JohnsonPY88}
David~S. Johnson, Christos~H. Papadimitriou, and Mihalis Yannakakis.
\newblock How easy is local search?
\newblock {\em J. Comput. Syst. Sci.}, 37(1):79--100, 1988.
\newblock \href {https://doi.org/10.1016/0022-0000(88)90046-3}
  {\path{doi:10.1016/0022-0000(88)90046-3}}.

\bibitem[KKMP21]{DBLP:conf/innovations/KleinbergKMP21}
Robert Kleinberg, Oliver Korten, Daniel Mitropolsky, and Christos~H.
  Papadimitriou.
\newblock Total functions in the polynomial hierarchy.
\newblock In {\em {ITCS}}, volume 185 of {\em LIPIcs}, pages 44:1--44:18.
  Schloss Dagstuhl - Leibniz-Zentrum f{\"{u}}r Informatik, 2021.
\newblock \href {https://doi.org/10.4230/LIPICS.ITCS.2021.44}
  {\path{doi:10.4230/LIPICS.ITCS.2021.44}}.

\bibitem[KMPS25]{KMPS25}
Michal Kouck{\'{y}}, Ian Mertz, Edward Pyne, and Sasha Sami.
\newblock Collapsing catalytic classes.
\newblock In {\em {FOCS}}, 2025.
\newblock To appear.
\newblock \href {https://arxiv.org/abs/2504.08444} {\path{arXiv:2504.08444}}.

\bibitem[Kor21]{Kor21}
Oliver Korten.
\newblock The hardest explicit construction.
\newblock In {\em 2021 {IEEE} 62nd {A}nnual {S}ymposium on {F}oundations of
  {C}omputer {S}cience---{FOCS} 2021}, pages 433--444. IEEE Computer Soc., Los
  Alamitos, CA, 2021.
\newblock \href {https://doi.org/10.1109/FOCS52979.2021.00051}
  {\path{doi:10.1109/FOCS52979.2021.00051}}.

\bibitem[Kor22]{DBLP:conf/coco/Korten22}
Oliver Korten.
\newblock Derandomization from time-space tradeoffs.
\newblock In {\em {CCC}}, volume 234 of {\em LIPIcs}, pages 37:1--37:26.
  Schloss Dagstuhl - Leibniz-Zentrum f{\"{u}}r Informatik, 2022.
\newblock \href {https://doi.org/10.4230/LIPICS.CCC.2022.37}
  {\path{doi:10.4230/LIPICS.CCC.2022.37}}.

\bibitem[Kor25]{Korten25}
Oliver Korten.
\newblock Range avoidance and the complexity of explicit constructions.
\newblock {\em Bull. {EATCS}}, 145, 2025.
\newblock URL: \url{http://eatcs.org/beatcs/index.php/beatcs/article/view/825}.

\bibitem[KP24]{KortenP24}
Oliver Korten and Toniann Pitassi.
\newblock Strong vs. weak range avoidance and the linear ordering principle.
\newblock In {\em {FOCS}}, pages 1388--1407. {IEEE}, 2024.
\newblock \href {https://doi.org/10.1109/FOCS61266.2024.00089}
  {\path{doi:10.1109/FOCS61266.2024.00089}}.

\bibitem[Kra95]{Krajicek95}
Jan Kraj{\'{\i}}cek.
\newblock Extensions of models of {$\mathsf{PV}$}.
\newblock In {\em Logic Colloquium}, volume~11 of {\em Lecture Notes in Logic},
  pages 104--114. Springer, 1995.
\newblock \href {https://doi.org/10.1007/978-3-662-22108-2\_8}
  {\path{doi:10.1007/978-3-662-22108-2\_8}}.

\bibitem[Kra04]{Krajicek04a}
Jan Kraj\'{\i}\v{c}ek.
\newblock Dual weak pigeonhole principle, pseudo-surjective functions, and
  provability of circuit lower bounds.
\newblock {\em J. Symb. Log.}, 69(1):265--286, 2004.
\newblock \href {https://doi.org/10.2178/jsl/1080938841}
  {\path{doi:10.2178/jsl/1080938841}}.

\bibitem[Kri85]{Krishnamurthy85}
Balakrishnan Krishnamurthy.
\newblock Short proofs for tricky formulas.
\newblock {\em Acta Informatica}, 22(3):253--275, 1985.
\newblock \href {https://doi.org/10.1007/BF00265682}
  {\path{doi:10.1007/BF00265682}}.

\bibitem[KT22]{KolodziejczykT22}
Leszek~Aleksander Kolodziejczyk and Neil Thapen.
\newblock Approximate counting and {NP} search problems.
\newblock {\em J. Math. Log.}, 22(3):2250012:1--2250012:31, 2022.
\newblock \href {https://doi.org/10.1142/S021906132250012X}
  {\path{doi:10.1142/S021906132250012X}}.

\bibitem[Li24]{Li24}
Zeyong Li.
\newblock Symmetric exponential time requires near-maximum circuit size:
  Simplified, truly uniform.
\newblock In {\em {STOC}}, pages 2000--2007. {ACM}, 2024.
\newblock \href {https://doi.org/10.1145/3618260.3649615}
  {\path{doi:10.1145/3618260.3649615}}.

\bibitem[LLR24]{LiLR24}
Jiawei Li, Yuhao Li, and Hanlin Ren.
\newblock Metamathematics of resolution lower bounds: {A} {TFNP} perspective.
\newblock {\em CoRR}, abs/2411.15515, 2024.
\newblock \href {https://arxiv.org/abs/2411.15515} {\path{arXiv:2411.15515}}.

\bibitem[LO87]{LagariasO87}
J.~C. Lagarias and Andrew~M. Odlyzko.
\newblock Computing {$\pi(x)$}: An analytic method.
\newblock {\em J. Algorithms}, 8(2):173--191, 1987.
\newblock \href {https://doi.org/10.1016/0196-6774(87)90037-X}
  {\path{doi:10.1016/0196-6774(87)90037-X}}.

\bibitem[LOS21]{LuOS21}
Zhenjian Lu, Igor~C. Oliveira, and Rahul Santhanam.
\newblock Pseudodeterministic algorithms and the structure of probabilistic
  time.
\newblock In {\em {STOC}}, pages 303--316. {ACM}, 2021.
\newblock \href {https://doi.org/10.1145/3406325.3451085}
  {\path{doi:10.1145/3406325.3451085}}.

\bibitem[LPR24]{PR24}
Yuhao Li, William Pires, and Robert Robere.
\newblock Intersection classes in {$\TFNP$} and proof complexity.
\newblock In Venkatesan Guruswami, editor, {\em 15th Innovations in Theoretical
  Computer Science Conference, {ITCS} 2024, January 30 to February 2, 2024,
  Berkeley, CA, {USA}}, volume 287 of {\em LIPIcs}, pages 74:1--74:22. Schloss
  Dagstuhl - Leibniz-Zentrum f{\"{u}}r Informatik, 2024.
\newblock \href {https://doi.org/10.4230/LIPICS.ITCS.2024.74}
  {\path{doi:10.4230/LIPICS.ITCS.2024.74}}.

\bibitem[LPT24]{DBLP:conf/focs/LiPT24}
Jiatu Li, Edward Pyne, and Roei Tell.
\newblock Distinguishing, predicting, and certifying: On the long reach of
  partial notions of pseudorandomness.
\newblock In {\em {FOCS}}, pages 1--13. {IEEE}, 2024.
\newblock \href {https://doi.org/10.1109/FOCS61266.2024.00095}
  {\path{doi:10.1109/FOCS61266.2024.00095}}.

\bibitem[Mer87]{DBLP:conf/crypto/Merkle87}
Ralph~C. Merkle.
\newblock A digital signature based on a conventional encryption function.
\newblock In {\em {CRYPTO}}, volume 293 of {\em Lecture Notes in Computer
  Science}, pages 369--378. Springer, 1987.
\newblock \href {https://doi.org/10.1007/3-540-48184-2\_32}
  {\path{doi:10.1007/3-540-48184-2\_32}}.

\bibitem[Mer23]{DBLP:journals/eatcs/Mertz23}
Ian Mertz.
\newblock Reusing space: Techniques and open problems.
\newblock {\em Bull. {EATCS}}, 141, 2023.
\newblock URL: \url{http://eatcs.org/beatcs/index.php/beatcs/article/view/780}.

\bibitem[M{\"{u}}l21]{Muller21}
Moritz M{\"{u}}ller.
\newblock Typical forcings, {NP} search problems and an extension of a theorem
  of riis.
\newblock {\em Ann. Pure Appl. Log.}, 172(4):102930, 2021.
\newblock URL: \url{https://doi.org/10.1016/j.apal.2020.102930}, \href
  {https://doi.org/10.1016/J.APAL.2020.102930}
  {\path{doi:10.1016/J.APAL.2020.102930}}.

\bibitem[NW94]{NisanW94}
Noam Nisan and Avi Wigderson.
\newblock Hardness vs randomness.
\newblock {\em J. Comput. Syst. Sci.}, 49(2):149--167, 1994.
\newblock \href {https://doi.org/10.1016/S0022-0000(05)80043-1}
  {\path{doi:10.1016/S0022-0000(05)80043-1}}.

\bibitem[OS17]{OliveiraS17}
Igor~C. Oliveira and Rahul Santhanam.
\newblock Pseudodeterministic constructions in subexponential time.
\newblock In {\em {STOC}}, pages 665--677, 2017.
\newblock \href {https://doi.org/10.1145/3055399.3055500}
  {\path{doi:10.1145/3055399.3055500}}.

\bibitem[Pap94]{DBLP:journals/jcss/Papadimitriou94}
Christos~H. Papadimitriou.
\newblock On the complexity of the parity argument and other inefficient proofs
  of existence.
\newblock {\em J. Comput. Syst. Sci.}, 48(3):498--532, 1994.
\newblock \href {https://doi.org/10.1016/S0022-0000(05)80063-7}
  {\path{doi:10.1016/S0022-0000(05)80063-7}}.

\bibitem[Pic15]{Pich15}
J{\'{a}}n Pich.
\newblock Circuit lower bounds in bounded arithmetics.
\newblock {\em Ann. Pure Appl. Log.}, 166(1):29--45, 2015.
\newblock \href {https://doi.org/10.1016/J.APAL.2014.08.004}
  {\path{doi:10.1016/J.APAL.2014.08.004}}.

\bibitem[Pot20]{Potechin20}
Aaron Potechin.
\newblock Sum of squares bounds for the ordering principle.
\newblock In Shubhangi Saraf, editor, {\em 35th Computational Complexity
  Conference, {CCC} 2020, July 28-31, 2020, Saarbr{\"{u}}cken, Germany (Virtual
  Conference)}, volume 169 of {\em LIPIcs}, pages 38:1--38:37. Schloss Dagstuhl
  - Leibniz-Zentrum f{\"{u}}r Informatik, 2020.
\newblock \href {https://doi.org/10.4230/LIPICS.CCC.2020.38}
  {\path{doi:10.4230/LIPICS.CCC.2020.38}}.

\bibitem[PPY23]{PPY23}
Amol Pasarkar, Christos~H. Papadimitriou, and Mihalis Yannakakis.
\newblock Extremal combinatorics, iterated pigeonhole arguments and
  generalizations of {$\PPP$}.
\newblock In Yael~Tauman Kalai, editor, {\em 14th Innovations in Theoretical
  Computer Science Conference, {ITCS} 2023, January 10-13, 2023, MIT,
  Cambridge, Massachusetts, {USA}}, volume 251 of {\em LIPIcs}, pages
  88:1--88:20. Schloss Dagstuhl - Leibniz-Zentrum f{\"{u}}r Informatik, 2023.
\newblock \href {https://doi.org/10.4230/LIPICS.ITCS.2023.88}
  {\path{doi:10.4230/LIPICS.ITCS.2023.88}}.

\bibitem[PRZ23]{PyneRZ23}
Edward Pyne, Ran Raz, and Wei Zhan.
\newblock Certified hardness vs. randomness for log-space.
\newblock In {\em {FOCS}}, pages 989--1007. {IEEE}, 2023.
\newblock \href {https://doi.org/10.1109/FOCS57990.2023.00061}
  {\path{doi:10.1109/FOCS57990.2023.00061}}.

\bibitem[PS23]{PichS23}
J{\'{a}}n Pich and Rahul Santhanam.
\newblock Towards {$\P\ne\NP$} from {E}xtended {F}rege lower bounds.
\newblock {\em Electron. Colloquium Comput. Complex.}, {TR23-199}, 2023.
\newblock URL: \url{https://eccc.weizmann.ac.il/report/2023/199}.

\bibitem[PT19]{PudlakT19}
Pavel Pudl{\'{a}}k and Neil Thapen.
\newblock Random resolution refutations.
\newblock {\em Comput. Complex.}, 28(2):185--239, 2019.
\newblock \href {https://doi.org/10.1007/S00037-019-00182-7}
  {\path{doi:10.1007/S00037-019-00182-7}}.

\bibitem[Pud15]{Pudlak15}
Pavel Pudl{\'{a}}k.
\newblock On the complexity of finding falsifying assignments for {{Herbrand}}
  disjunctions.
\newblock {\em Arch. Math. Log.}, 54(7-8):769--783, 2015.
\newblock \href {https://doi.org/10.1007/S00153-015-0439-6}
  {\path{doi:10.1007/S00153-015-0439-6}}.

\bibitem[PWW88]{ParisWW88}
Jeff~B. Paris, A.~J. Wilkie, and Alan~R. Woods.
\newblock Provability of the pigeonhole principle and the existence of
  infinitely many primes.
\newblock {\em J. Symb. Log.}, 53(4):1235--1244, 1988.
\newblock \href {https://doi.org/10.1017/S0022481200028061}
  {\path{doi:10.1017/S0022481200028061}}.

\bibitem[Pyn24]{Pyne24}
Edward Pyne.
\newblock Derandomizing logspace with a small shared hard drive.
\newblock In {\em {CCC}}, volume 300 of {\em LIPIcs}, pages 4:1--4:20. Schloss
  Dagstuhl - Leibniz-Zentrum f{\"{u}}r Informatik, 2024.
\newblock \href {https://doi.org/10.4230/LIPICS.CCC.2024.4}
  {\path{doi:10.4230/LIPICS.CCC.2024.4}}.

\bibitem[Raz87]{Razborov87}
Alexander~A Razborov.
\newblock Lower bounds on the size of bounded depth circuits over a complete
  basis with logical addition.
\newblock {\em Mathematical Notes of the Academy of Sciences of the USSR},
  41(4):333--338, 1987.

\bibitem[Rii01]{Riis01}
S{\o}ren Riis.
\newblock A complexity gap for tree resolution.
\newblock {\em Comput. Complex.}, 10(3):179--209, 2001.
\newblock \href {https://doi.org/10.1007/S00037-001-8194-Y}
  {\path{doi:10.1007/S00037-001-8194-Y}}.

\bibitem[RRV02]{DBLP:journals/jcss/RazRV02}
Ran Raz, Omer Reingold, and Salil~P. Vadhan.
\newblock Extracting all the randomness and reducing the error in {T}revisan's
  extractors.
\newblock {\em J. Comput. Syst. Sci.}, 65(1):97--128, 2002.
\newblock \href {https://doi.org/10.1006/JCSS.2002.1824}
  {\path{doi:10.1006/JCSS.2002.1824}}.

\bibitem[RSW22]{RenSW22}
Hanlin Ren, Rahul Santhanam, and Zhikun Wang.
\newblock On the range avoidance problem for circuits.
\newblock In {\em {FOCS}}, pages 640--650. {IEEE}, 2022.
\newblock \href {https://doi.org/10.1109/FOCS54457.2022.00067}
  {\path{doi:10.1109/FOCS54457.2022.00067}}.

\bibitem[Smo87]{Smolensky87}
Roman Smolensky.
\newblock Algebraic methods in the theory of lower bounds for boolean circuit
  complexity.
\newblock In {\em {STOC}}, pages 77--82. {ACM}, 1987.
\newblock \href {https://doi.org/10.1145/28395.28404}
  {\path{doi:10.1145/28395.28404}}.

\bibitem[STV01]{DBLP:journals/jcss/SudanTV01}
Madhu Sudan, Luca Trevisan, and Salil~P. Vadhan.
\newblock Pseudorandom generators without the {XOR} lemma.
\newblock {\em J. Comput. Syst. Sci.}, 62(2):236--266, 2001.
\newblock \href {https://doi.org/10.1006/JCSS.2000.1730}
  {\path{doi:10.1006/JCSS.2000.1730}}.

\bibitem[TCH12]{TaoCH12}
Terence Tao, Ernest Croot, III, and Harald Helfgott.
\newblock Deterministic methods to find primes.
\newblock {\em Mathematics of Computation}, 81(278):1233--1246, 2012.
\newblock \href {https://doi.org/10.1090/S0025-5718-2011-02542-1}
  {\path{doi:10.1090/S0025-5718-2011-02542-1}}.

\bibitem[Tha02]{Thapen-PhD}
Neil Thapen.
\newblock {\em The weak pigeonhole principle in models of bounded arithmetic}.
\newblock PhD thesis, University of Oxford, 2002.

\bibitem[Tha24]{Thapen24}
Neil Thapen.
\newblock How to fit large complexity classes into {$\TFNP$}.
\newblock {\em CoRR}, abs/2412.09984, 2024.
\newblock \href {https://arxiv.org/abs/2412.09984} {\path{arXiv:2412.09984}}.

\bibitem[Uma03]{Umans03}
Christopher Umans.
\newblock Pseudo-random generators for all hardnesses.
\newblock {\em J. Comput. Syst. Sci.}, 67(2):419--440, 2003.
\newblock \href {https://doi.org/10.1016/S0022-0000(03)00046-1}
  {\path{doi:10.1016/S0022-0000(03)00046-1}}.

\end{thebibliography}

\appendix

\section{Herbrandization}\label{appendix: Herbrandization}
Herbrandization is a basic construction in logic.
Roughly speaking, any $\mathsf{TF}\Sigma_2^P$ problem $L$ can be captured by a logical formula $\forall x\exists y\forall z~\varphi(x, y, z)$,\footnote{Formally, they are called ``$\forall\Sigma_2^b$-formulas'', where the subscript $2$ stands for two alternations ($\forall\exists\forall$) and the superscript $b$ stands for ``bounded'', i.e., the lengths of $x, y, z$ are polynomially related and $\varphi$ is a polynomial-time predicate.} 
where the task of $L$ is, given $x$, to find a $y$ such that $\forall z~\varphi(x, y, z)$ holds. To Herbrandize this formula, we add another function $h$ (treated as an input oracle) and consider the formula $\forall x\exists y~\varphi(x, y, h(y))$. One can then define a $\TFNP$ problem by treating $h$ as a first-order variable like $x$, where the task is, given $x$ and $h$, to find a $y$ such that $\varphi(x, y, h(y))$ holds.

\begin{table}[H]
    \begin{tabular}{|c|c|}
        \hline
        \textbf{$\mathsf{TF}\Sigma_2^P$ search problem} & \textbf{$\TFNP$ problem via Herbrandization}\\
        \hline
        $\forall x\exists y\forall z~\varphi(x, y, z)$ & $\forall x, h\exists y~\varphi(x, y, h(y))$\\
        \hline
        \makecell{$\Avoid$:\\ $\forall D\exists x\forall y~D(y) \ne x$} & \makecell{$\lossycode$:\\ $\forall C, D\exists x~D(C(x)) \ne x$}\\
        \hline
        \makecell{\eqref{eq: ilp-from-borger}:\\ $\forall F\exists x \forall y~\mleft(F(F(y)) \ne F(x) \lor F(y) = x\mright)$} & \makecell{$\ILP$:\\ $\forall F, G\exists x~\mleft(F(F(G(x))) \ne F(x) \lor F(G(x)) = x\mright)$}\\
        \hline
    \end{tabular}
    \caption{Some examples of $\TFNP$ problems via Herbrandization.}
\end{table}

\section{Proof Complexity Characterizations of Randomized Reductions}
\label{sec:appendixPC}

In this section we prove \autoref{thm:mainChar}, which we restate next.

\begin{theorem}
    \label{thm:metaRandomProof}
    If a proof system ${\cal P}$ is characterized by the total search problems reducible to $R \in \TFNP^{dt}$, then $r{\cal P}$ is characterized by the total search problems that are randomized-reducible to $R$.
\end{theorem}

\begin{proof}
Fix a complexity $c$ $r{\cal P}$ proof $\cal D$ of a CNF formula $H$. We will construct a complexity $\Theta(c)$ randomized reduction from $\Search_H$ to a complete problem $\Search_F$ for $\cal C$. On input $x \in \{0,1\}^n$ the reduction first samples $(\Pi, B)$ from the distribution $\cal D$ given by the $r{\cal P}$ proof, where $\Pi$ is a complexity $c$ $\cal P$-proof of $H \wedge B$. Since by assumption $\cal P$ is characterized by ${\cal C}$, this implies that there is a complexity $O(c)$ reduction ${\cal T}=(T, \{T_o\})$ from $\Search_{H \wedge B}$ to $\Search_F$. Relabel each leaf of an output decision tree $T_o$ in $\cal T$, which is labelled with a clause of $B$, by $\bot$, indicating a failure event of the randomized reduction. 

Observe that property (1) (correctness) of a randomized reduction to $\Search_H$ is satisfied. It remains to argue that property (2) (error probabilities) is also satisfied. This follows since the reduction was constructed from a $r \cal P$ proof. Indeed, the probability that the reduction fails is the probability that we sampled a reduction ${\cal T}=(T, \{T_o\})$ and the leaf of the output decision tree that we arrive at when following $x$ is labelled by $\bot$. By construction, we arrive at a $\bot$ leaf only if we falsify a clause of the corresponding CNF $B$. However, by the definition of an $r {\cal P}$ proof, for every $x$ the probability that every clause in $B$ is satisfied by $x$ is at least $2/3$.

For the converse, let $\Search_F$ be a ${\cal C}$-complete problem and let ${\cal D}$ be a randomized reduction from $\Search_H$ to 	$\Search_F$ of complexity $c$. We will argue that each ${\cal T} \sim {\cal D}$ is a deterministic reduction from $\Search_{H \wedge B}$ to $\Search_F$ for some CNF formula $B$ of width $O(c^2)$. Because $\cal C$ is characterized by ${\cal P}$, we will obtain a $\cal P$ proof of $H \wedge B$, and putting these together, a $r {\cal P}$ proof of $H$. 

Let $F =C_1 \wedge \ldots \wedge C_m$ and let ${\cal T} \sim {\cal D}$ where ${\cal T} := (\{T_i\}_{i \in [n]}, \{T_o\}_{o \in [m]})$.
Following \cite{BussFI23}, let the \emph{reduced formula} $F_{\cal T}$ be the CNF formula obtained as follows: for each clause $C \in F$, let the decision tree $T^C$ be obtained by sequentially running the decision trees $T_i$ for each $i \in vars(C)$ to obtain an assignment $\alpha \in \{0,1\}^{vars(C)}$. 
If $C(\alpha)=b \in \{0,1\}$ then label this a ``$b$-leaf'', for $b \in \{0,1\}$. Say that a root-to-leaf path $p \in T^C$ is a $b$-path if it ends at a $b$-leaf. Define 
\begin{align*} 
C({\cal T})&:= \bigwedge_{0\mbox{-path }p \in T^C} \neg p %
\end{align*}
In words, $C(T)$ says that the clause $C$, after substituting $\cal T$ for the variables, is never falsified. The reduced CNF formula is 
\[F_{\cal T}:= \bigwedge_{o \in [m]} \left( C_o({\cal T})  \lor \bigwedge_{p \in T_o} \neg p \right). \]
$F_{\cal T}$ formalizes the definition of a reduction (\autoref{eq:red}) by $\cal T$ to $\Search_F$---if we falsify $C_o({\cal T})$ and follow path $p$ in $T_o$ then the label of the leaf of $p$ is a valid solution to the search problem reducing to $\Search_F$ by ${\cal T}$. Hence,
$\cal T$ is a reduction from a CNF formula $H \wedge B$ to $F$ iff each clause of $F_{\cal T}$ is a weakening\footnote{A clause $C$ is a weakening of a clause $D$ if the literals of $D$ are a subset of the literals of $D$.} of a clause of $H \wedge B$.

Then, knowing $H$, we can recover $B$ as follows: let $B$ be the set of all clauses of $F_{\cal T}$ which are not a weakening of any clause of $H$. Then ${\cal T}$ is a reduction from $\Search_{H \wedge B}$ to $\Search_F$. Note that the width of the clauses in $B$ is at most $O(c^2)$ as every clause in $F$ has width $O(c)$, the trees in ${\cal T}$ have depth $O(c)$, and we have substituted the decision trees for the variables of the clause.

It remains to argue that for every $x \in \{0,1\}^n$, $\Pr_{(\Pi,B) \sim {\cal D}}[B(x)=1] \geq 2/3$. This is immediate from the fact that $\cal T$ is a randomized reduction to $\Search_F$. Indeed, each clause $K \in F_{\cal T}$ comes from some $C_o({\cal T}) \lor \neg p$ for $p \in T_o$ and some $o \in [m]$. That is, $K= \neg p^* \lor \neg p$ for some $0$-path $p^* \in T^C$. %
Hence, by the definition of a randomized reduction,
\begin{align*}
   2/3 &\leq  \Pr_{{\cal T} \sim {\cal D}} \left[ \forall o \in [m] : (o,T(x)) \not \in \Search_F \lor T_o(x) \neq \bot \right]  \\
   &=\Pr_{{\cal T} \sim {\cal D}} \left[ \forall o \in [m] : C_o({\cal T}(x)) \neq 0 \lor T_o(x) \neq \bot  \right]\\
   &= \Pr_{{\cal T} \sim {\cal D}} \left[ \forall o \in [m], \forall~\mbox{$0$-paths }  p^* \in T^{C_o}, \forall~\mbox{$\bot$-paths } p \in T_o :  \neg p^*(x)=1  \lor \neg p(x)=1  \right] \\ 
   &=\Pr_{{\cal T} \sim {\cal D}} \left[ \forall K \in B: K(x)=1 \right]. \tag{$B$ are clauses of $F_{\cal T}$}
\end{align*}
\end{proof}

\section{Proof of \autoref{thm: certified Nisan--Wigderson}}\label{sec: proof of NW in APC1}

\paragraph{Weak designs.} We say that $I_1, I_2, \dots, I_m\subseteq [d]$ is a \emph{weak $(\ell, \rho)$-design} if:\begin{itemize}
    \item For every $i\le m$, $|I_i| = \ell$; and
    \item For every $i\le m$, 
    \[\sum_{j < i}2^{|S_i\cap S_j|} \le \rho \cdot (m-1).\]
\end{itemize}
\begin{theorem}[{\cite{DBLP:journals/jcss/RazRV02}}]\label{thm: weak designs}
    For every $\ell, m\in\N$ and $\rho > 1$, there is a weak $(\ell, \rho)$-design $S_1, S_2, \dots, S_m\subseteq [d]$ with
    \[d = \mleft\lceil \frac{\ell}{\ln \rho} \mright\rceil \cdot \ell.\]
    Moreover, such a family can be found in deterministic time $\poly(m, d)$.
\end{theorem}
    
\paragraph{List-decodable codes.} A pair of functions $(\Enc, \Dec)$ is called an \emph{$(L, 1/2-\eps)$-list-decodable code} if:\begin{itemize}
    \item $\Enc: \{0, 1\}^n \to \{0, 1\}^{2^\ell}$ and $\Dec: \{0, 1\}^{2^\ell} \to (\{0, 1\}^n)^L$ are computable in deterministic polynomial time, and
    \item for every $x \in \{0, 1\}^n$ and $y\in\{0, 1\}^{2^\ell}$ such that $y$ is $(1/2 - \eps)$-close to $\Enc(x)$, $\Enc(x)$ appears in the list $\Dec(y)$.
\end{itemize}

\begin{theorem}[{\cite{DBLP:journals/jcss/SudanTV01}}]\label{thm: list decodable codes}
    For every $n\in\N$ and $\eps > 0$, there exists an $(L, 1/2-\eps)$-list-decodable code with $\ell = O(\log(n/\eps))$ and $L = \poly(1/\eps)$.
\end{theorem}

In what follows, for notational convenience, we will think of length-$2^\ell$ strings $f \in \{0, 1\}^{2^\ell}$ as (the truth tables of) $\ell$-bit Boolean functions $f: \{0, 1\}^\ell\to\{0, 1\}$.

\def\Hyb{{\mathsf{Hyb}}}
{
\def\EitherNotInjSurjPair{either
    \[
        [val]\times [D]\xrightleftharpoons[G_<^S(f, -)]{H_<^S(f, -)} (S\cupdot [\eps M])\times [D]\quad\text{or}\quad
        S\times [D]\xrightleftharpoons[G_>^S(f, -)]{H_>^S(f, -)} ([val]\cupdot [\eps M])\times [D]
    \]
}
\ThmNisanWigdersoninAPC*
}

\begin{proof}
    Let $\eps'$ be the biggest (inverse) power of $2$ such that $\eps' \le \eps / m$. Let $(\Enc, \Dec)$ be an $(L, 1/2-\eps')$-list-decodable code guaranteed by \autoref{thm: list decodable codes} with $\ell = O(\log n/\eps') = O(\log (nm/\eps))$ and $L \le \poly(1/\eps') \le \poly(m/\eps)$. Let $I_1, I_2, \dots, I_m\subseteq [d]$ be a weak $(\ell, \rho)$-design guaranteed by \autoref{thm: weak designs} with $d = O(\ell^2 / \log\rho)$. Given $f \in \{0, 1\}^n$ and $z \in \{0, 1\}^d$ as inputs, the generator first computes $\tilde{f} := \Enc(f)$, and then outputs
	\[\PRG(f, z) := (\tilde{f}(z|_{I_1}), \tilde{f}(z|_{I_2}), \dots, \tilde{f}(z|_{I_m})).\]
    The classical proof of Nisan--Wigderson \cite{NisanW94} shows that for every $S\subseteq \{0, 1\}^m$ and every $f \in \{0, 1\}^n$ that is worst-case hard against $S$-oracle circuits, $f$ provides an additive approximation of $|S|$, i.e.,
    \[\mleft|\frac{|S\cap \PRG_f|}{D} - \frac{|S|}{M}\mright| \le \eps.\]
    The \emph{proof} of the above fact goes through a \emph{hybrid} argument and considers the following intermediate generators. For each $0\le i\le m$, let $\Hyb_i:\{0, 1\}^d \times \{0, 1\}^m \to \{0, 1\}^m$ denote the generator that takes $z\in\{0, 1\}^d$ and $r\in\{0, 1\}^m$ as inputs, and outputs
	\[\Hyb_i(z, r) := (\tilde{f}(z|_{I_1}), \dots, \tilde{f}(z|_{I_i}), r_{i+1}, \dots, r_m).\]
	Let $V_i := \{(z, r): \Hyb_i(z, r) \in S\}$; intuitively, $|V_i|$ is an estimation of $|S|\cdot 2^d$ by the generator $\Hyb_i$. Suppose $f$ is hard, then for every $1\le i \le m$, $|V_{i-1}|$ is close to $|V_i|$. It follows that $|V_0|$ is close to $|V_m|$. Since $|V_0| = 2^d\cdot |S|$ and $|V_m| = |S\cap \PRG_f| \cdot 2^m = val\cdot 2^d$, $val$ is a good estimation of $|S|$.

    In a nutshell, the proof of \cite[Theorem 2.7]{Jerabek07} proceeds by constructing injection-surjection pairs witnessing $|V_{i-1}|\lesssim |V_i|$ and $|V_{i-1}|\gtrsim |V_i|$ for each $i$. Composing all these injection-surjection pairs gives the final injection-surjection pairs witnessing $|V_0| \lesssim |V_m|$ and $|V_0|\gtrsim |V_m|$ respectively.

    \def\diff{{\mathit{diff}}}
    Fix $1\le i \le m$, we now aim to construct injection-surjection pairs witnessing the inequalities $|V_{i-1}|\lesssim |V_i|$ and $|V_{i-1}| \gtrsim |V_i|$. That is, letting $\diff := \eps' \cdot 2^{m+d}$, we will construct injection-surjection pairs $V_i \xrightleftharpoons[G_i]{H_i} V_{i-1}\cupdot[\diff]$ and $V_{i-1} \xrightleftharpoons[G'_i]{H'_i} V_i\cupdot[\diff]$ that depends on $f$. Moreover, if $f$ is indeed a ``hard function'', then these injection-surjection pairs will be valid, i.e., $G_i\circ H_i$ is the identity map on $V_i$ and $G'_i \circ H'_i$ is the identity map on $V_{i-1}$.

    \def\aux{{\mathsf{aux}}}
    Let $z\in\{0, 1\}^d$ and $r\in\{0, 1\}^m$, we denote $z' := z|_{I_i}$ and $\aux := (z|_{[d]\setminus I_i}, r_{[m]\setminus i})$. Note that there is a one-to-one correspondence between $(z, r)$ and $(z', r_i, \aux)$. Now given $i, \aux$, we define 
	\begin{align*}
	X_{i, \aux} :=&\,\{(z', r_i): (\tilde{f}(z|_{I_1}), \tilde{f}(z|_{I_2}), \dots, \tilde{f}(z|_{I_{i-1}}), r_i, r_{i+1}, \dots, r_n) \in S\}\text{ and}\\
	Y_{i, \aux} :=&\, \{(z', b): (z', \tilde{f}(z')) \in X_{i, \aux}\}.
	\end{align*}
	
	Then, $V_{i-1} = \bigcup_\aux \mleft(X_{i, \aux}\times \{\aux\}\mright)$ and $V_i = \bigcup_\aux \mleft(Y_{i, \aux}\times \{\aux\}\mright)$.
	To show that $|V_{i-1}|\approx |V_i|$, it suffices to show that $|X_{i, \aux}| \approx |Y_{i, \aux}|$ for every $\aux$ (under the assumption that $f$ is ``hard''); this is exactly what the next claim shows.
    \begin{claim}
		If $\mleft||X_{i, \aux}| - |Y_{i, \aux}|\mright| > \eps'2^{\ell+1}$, then given $(i, \aux)$ and additional $\rho\cdot (m-1) + 2$ advice bits, it is possible to recover a string $\tilde{f}_{\sf apx}$ that is $(1/2 + \eps')$-close to $\tilde{f}$ in deterministic $\poly(\rho m, 2^\ell)$ time with oracle access to $S$. Moreover, given $(i, \aux)$ and $\tilde{f}$, the additional advice bits can be computed in deterministic $\poly(\rho m, 2^\ell)$ time with oracle access to $S$.
	\end{claim}
	\begin{claimproof}[Proof Sketch]
        For every $b \in \{0, 1\}$, define a string $\tilde{f}^b \in \{0, 1\}^{2^\ell}$, where for every $z' \in \{0, 1\}^\ell$, the $z'$-th bit of $\tilde{f}^b$ is $1$ if and only if $(z', b) \in X_{i, \aux}$. It can be shown that
        \[|X_{i, \aux}| - |Y_{i, \aux}| = \Delta(\tilde{f}^1, \tilde{f}) - \Delta(\tilde{f}^0, \tilde{f}),\]
        where $\Delta(\cdot, \cdot)$ denotes the Hamming distance of two binary strings.
        Since $\mleft||X_{i, \aux}| - |Y_{i, \aux}|\mright| > \eps'2^{\ell+1}$, there must be some $b \in \{0, 1\}$ such that $\Delta(\tilde{f}^b, \tilde{f}) \notin [(1/2-\eps')2^\ell, (1/2+\eps')2^\ell]$.

        Now suppose that $\tilde{f}$ is fixed and $(i, \aux)$ is given. For each $j < i$, $\tilde{f}(z|_{I_j})$ is a function over $z'$ that only depends on $|S_i\cap S_j|$ bits of $z'$. Hence, the truth table of this function can be recorded in $2^{|S_i\cap S_j|}$ bits. If we write down the truth tables of $\tilde{f}(z|_{I_j})$ for every $j < i$ as advice, this only costs
        \[\sum_{j < i}2^{|S_i \cap S_j|}\le \rho\cdot (m-1)\]
        advice bits. We append two additional advice bits $b, b' \in \{0, 1\}$, where $b$ indicates that $\Delta(\tilde{f}^b, \tilde{f}) \notin [(1/2-\eps')2^\ell, (1/2+\eps')2^\ell]$ and $b'$ indicates whether $\Delta(\tilde{f}^b, \tilde{f})$ is above $1/2$ or not. It is easy to see that we can recover a string that is $(1/2+\eps')$-close to $\tilde{f}$ given $(i, \aux)$ and these advice bits; moreover, these advice bits can be computed in deterministic polynomial time given $(i, \aux)$ and $f$ as inputs.
	\end{claimproof}
    Composing the above claim with the list-decodable code $(\Enc, \Dec)$, we obtain the following corollary:
    \begin{corollary}\renewcommand\qedsymbol{$\diamond$}\label{cor: Nisan Wigderson compression}
        If $||X_{i, \aux}|-|Y_{i, \aux}|| > \eps'2^{\ell+1}$, then given $(i, \aux)$ and additional $\rho\cdot (m-1) + \log L + 2$ advice bits, it is possible to compute $f$ in deterministic $\poly(\rho m, 2^\ell)$ time with oracle access to $S$. Moreover, given $(i, \aux)$ and $f$, the additional advice bits can be computed in deterministic $\poly(\rho m, 2^\ell)$ time with oracle access to $S$.
        \qed
    \end{corollary}

    Let $\Decomp^S(i, \aux, \alpha)$ be the procedure for computing $f$ from $(i, \aux)$ and the advice bits $\alpha$ as asserted in \autoref{cor: Nisan Wigderson compression}. We now say $f$ is ``hard'' if $f$ is not in the range of $\Decomp^S$. (Note that $\Decomp^S$ takes $k := \log m + d-\ell+m-1+\rho(m-1) + \log L + 2 \le d + (\rho + 1)(m-1) + O(\log(m/\eps))$ bits. As long as this is less than $n$ bits, a hard $f$ must exist.)
	
	Suppose that we are given some $f$ that is hard. We can create injection-surjection pairs between $X_{i, \aux}$ and $Y_{i, \aux}$ by brute force; this only takes deterministic $\poly(2^\ell)$ time with oracle access to $S$. As an example, we construct 
    \[Y_{i, \aux} \xrightleftharpoons[G_{i, \aux}]{H_{i, \aux}} X_{i, \aux}\cupdot [\eps' \cdot 2^{\ell+1}].\]
    \begin{itemize}
        \item $G_{i, \aux}(v)$: If $v \in X_{i, \aux}$ then let $p$ be the integer such that $v$ is the (lexicographically) $p$-th smallest element of $X_{i, \aux}$ (the smallest element is the $0$-th); if $v\in [\eps' \cdot 2^{\ell+1}]$ then let $p := |X_{i, \aux}| + v$. Return the $p$-th smallest element of $Y_{i, \aux}$; if $|Y_{i, \aux}| \ge p$ then return an arbitrary element (say the smallest one).
        \item $H_{i, \aux}(v)$: Suppose that $v$ is the $p$-th smallest element of $Y_{i, \aux}$. If $p < |X_{i, \aux}|$ then return the $p$-th smallest element of $X_{i, \aux}$; otherwise return $p - |X_{i, \aux}| \in [\eps' \cdot 2^{\ell+1}]$.
        \item It is straightforward to verify that $G_{i, \aux} \circ H_{i, \aux}$ is the identity map and that $G_{i, \aux}$ and $H_{i, \aux}$ are computable in deterministic $\poly(2^\ell)$ time.
    \end{itemize}

	We can similarly construct $X_{i, \aux} \xrightleftharpoons[G'_{i, \aux}]{H'_{i, \aux}} Y_{i, \aux} \cupdot [\eps' 2^{\ell+1}]$ such that $G'_{i, \aux} \circ H'_{i, \aux}$ is the identity map and $G'_{i, \aux}$ and $H'_{i, \aux}$ are computable in deterministic $\poly(2^\ell)$ time.
	
	Now we describe the functions $G_i, H_i, G'_i, H'_i$.
	\begin{itemize}
		\item Let $v$ be the input of $G_i$. If $v\in V_{i-1}$ then write $v = (z, r) = (z', r_i, \aux)$; if $v \in [\diff]$ then let $\aux := \lfloor v / (\eps' 2^{\ell+1})\rfloor$ (treated as both a number in $[2^{m+d-\ell-1}]$ and a length-$(m+d-\ell-1)$ string) and $v' := v - \aux \cdot \eps' 2^{\ell+1}$. Assuming $f$ is hard, we have $v \in X_{i, \aux} \cupdot [\eps' 2^{\ell+1}]$. Let $u := G_{i, \aux}(v) \in Y_{i, \aux}$, write $u = (z^u, r_i^u)$ and return $G_i(v) := (z^u, r_i^u, \aux)$.
		\item Let $u \in V_i$ be the input of $H_i$, and write $u := (z', r_i, \aux)$ where $(z', r_i) \in Y_{i, \aux}$. We can compute $v := H_{i, \aux}(z', r_i) \in X_{i, \aux} \cupdot [\eps' 2^{\ell+1}]$. If $v\in X_{i, \aux}$ then we write $v = (z^v, r_i^v)$ and return $(z^v, r_i^v, \aux) \in V_{i-1}$; if $v\in [\eps' 2^{\ell+1}]$ then we return $\aux \cdot \eps' 2^{\ell+1} + v \in [\diff]$.
		\item The definitions of $G'_i, H'_i$ are analogous.
	\end{itemize}

    It is easy to see that if $f$ is indeed hard, then $G_i\circ H_i$ is the identity map. However, if $f$ is \emph{not hard}, there is no guarantee that $G_i\circ H_i$ is the identity map. Nevertheless we still gain something: Given any witness $u \in V_i$ such that $G_i(H_i(u)) \ne u$, if we write $u = (z', r_i, \aux)$ then we have $||X_{i, \aux}| - |Y_{i, \aux}|| > \eps' 2^{\ell+1}$. By \autoref{cor: Nisan Wigderson compression}, we can compute $(i, \aux, \alpha)$ from this witness $u$ deterministically such that $\Decomp^S(i, \aux, \alpha) = f$, i.e., we found a \emph{witness for the non-hardness of $f$} as well! In summary, let $f$ be a purported hard function, we can \emph{either} use $f$ to perform approximate counting and obtain injection-surjection pairs $(G_i, H_i)$ \emph{or}, if $(G_i, H_i)$ fails to be an injection-surjection pair, \emph{exploit this failure} to compress $f$.

	Finally, we can compose the functions $\{G_i\}$, $\{H_i\}$, $\{G'_i\}$, and $\{H'_i\}$ to obtain $G_<, H_<, G_>, H_>$.
	
	\begin{itemize}
		\item Let $v \in (S\cupdot [\eps M])\times [D] = V_0 \cupdot [m\cdot \diff]$ be the input of $G_<$. For each $i$ from $1$ to $m$, if currently we have $v\in V_{i-1} \cupdot [\diff]$, then we update $v \gets G_i(v)$; otherwise $v\in [\diff, m\cdot \diff)$ and we update $v\gets v - \diff$.
		\item Let $v \in V_m = [val] \times [D]$ be the input of $H_<$. For each $i$ from $m$ downto $1$, if currently we have $v\in V_i$, then we update $v\gets H_i(v)$; otherwise $v$ is some number in $[m\cdot \diff]$ and we update $v\gets v + \diff$.
		\item The definitions of $G_>, H_>$ are analogous.
	\end{itemize}

	It is easy to see that the functions
	\[[val] \times [D]\xrightleftharpoons[G_<]{H_<}(S\cupdot[\eps M])\times [D]\]
	satisfy the following property: given any $w \in [val]\times [D]$ such that $G_<(H_<(w)) \ne w$, we can compute in deterministic $\poly(n, 2^\ell) = \poly(n/\eps)$ time a ``compression'' $\Comp^S(f, w)$ of $f$ that decompresses to $f$ via $\Decomp^S$. Furthermore, $G_<$ and $H_<$ themselves can be computed in deterministic $\poly(n/\eps)$ time. The conclusions for $G_>$ and $H_>$ can be proved similarly.
\end{proof}

\end{document}